\documentclass[12 pt, journal, onecolumn]{IEEEtran}

\usepackage{cite}
\usepackage{enumitem}
\usepackage{bm}
\usepackage{tikz}
\usepackage{soul}
\usetikzlibrary{automata,positioning}
\usepackage[colorinlistoftodos]{todonotes}
\usepackage{float}
\usepackage{subfig}
\usepackage{cancel}
\usepackage{amsmath,amssymb,amsthm}
\usepackage{graphicx,color}
\usepackage[figuresright]{rotating}
\usepackage[utf8]{inputenc} 
\usepackage[T1]{fontenc}    
\usepackage{url}            
\usepackage{booktabs}       
\usepackage{enumitem}
\usepackage{amsfonts}       
\usepackage{amssymb,mathtools}
\usepackage{amsfonts}
\usepackage{nicefrac}       
\usepackage{microtype}      
\usepackage{amsmath,amssymb,amsthm}

\usepackage{pbox}
\usepackage{tkz-euclide}
\usetkzobj{all}
\usepackage{tikz}
\usepackage{url}
\usepackage{todonotes}

\usepackage{bm}

\usepackage{times}

\usepackage{amsmath}
\usepackage{amsfonts}
\usepackage{amssymb}
\usepackage{graphicx}
\usepackage{algorithm}
\usepackage[noend]{algpseudocode}

\usepackage[bookmarks=false,colorlinks=true,urlcolor=blue,citecolor=blue,linkcolor=blue]{hyperref}
\usepackage{cleveref}

\newtheorem{lem}{Lemma}
\newtheorem{thm}{Theorem}
\newtheorem{defn}{Definition}

\newtheorem{clm}{Claim}

\title{CodeNet: Training Large Scale Neural Networks in Presence of Soft-Errors} 

\author{Sanghamitra~Dutta, Ziqian~Bai,   Tze Meng Low and Pulkit Grover\\ 
\thanks{ S. Dutta, T. M. Low and P. Grover are with Department of ECE at Carnegie Mellon University, PA, United States. Z. Bai is with Simon Fraser University, BC, Canada. [Corresponding Author Contact: sanghamd@andrew.cmu.edu]}
}

\begin{document} 
\bstctlcite{IEEEexample:BSTcontrol}

\baselineskip24pt


\maketitle



\begin{abstract}
This work proposes the first strategy to make distributed training of neural networks resilient to computing errors, a problem that has remained unsolved despite being first posed in 1956 by von Neumann. He also speculated that the efficiency and reliability of the human brain is obtained by allowing for low power but error-prone components with redundancy for error-resilience. It is surprising that this  problem remains open, even as massive artificial neural networks are being trained on increasingly low-cost and unreliable processing units. Our coding-theory-inspired strategy, ``CodeNet,'' solves this problem by addressing three  challenges in the science of reliable computing: (i) Providing the first strategy for error-resilient neural network training by encoding each layer separately; (ii) Keeping the overheads of coding (encoding/error-detection/decoding) low by obviating the need to re-encode the updated parameter matrices after each iteration from scratch. (iii) Providing a completely decentralized implementation with no central node (which is a single point of failure), allowing all primary computational steps to be error-prone. We theoretically demonstrate that CodeNet has higher error tolerance than replication, which we leverage to speedup computation time. Simultaneously, CodeNet requires lower redundancy than replication, and equal computational and communication costs in scaling sense. We first theoretically demonstrate the benefits of using CodeNet over replication in reducing the expected computation time accounting for checkpointing. Our experiments show that CodeNet achieves the best accuracy-runtime tradeoff compared to both replication and uncoded strategies. CodeNet is a significant step towards biologically plausible neural network training, that could hold the key to orders of magnitude efficiency improvements.

\end{abstract}


\section{Introduction}

Inspired by the success of Shannon's theory of information~\cite{shannon1948} in addressing errors in communication, and the remarkable efficiency and speed of the human brain in processing information with seemingly error-prone components, von Neumann began the study of computing in presence of noisy and erroneous computational elements in 1956~\cite{von1956probabilistic}, as is also evident from the influence of the McCulloch-Pitts model of a neuron~\cite{mcculloch1943logical} in his work. It is often speculated~\cite{sreenivasan2011error} that the error-prone nature of brain's hardware helps it be more efficient: rather than paying the cost at a component level, it may be more efficient to accept component-level errors, and utilize sophisticated error-correction mechanisms for overall reliability of the computation\footnote{See also the ``Efficient Coding Hypothesis'' of Barlow~\cite{barlow1961possible} and a recent validation of this concept in computing linear transforms~\cite{yang2017ITTrans}.}. The  brain operates at a surprisingly low power of about $15$ W~\cite{yanushkevich2013introduction}, and attains high accuracy and speeds, despite individual neurons in the brain being slow and error-prone~\cite{sreenivasan2011error,olshausen1996emergence,barlow1961possible,whatmough2018dnn,wang2018training}. Even today, the brain's system-level energy requirement is orders of magnitude smaller than the most efficient computers, despite substantial efforts in imitating the brain, going as far as using spiking neural networks~\cite{truenorth}. While there is growing interest in training using low-cost and unreliable hardware~\cite{binaryconnect,sakr2017analytical,sakr2019accumulation,wang2019deep, yang2018bit}, they still use significantly more power-consuming and reliable components than that used in the brain. Thus, von Neumann's original motivation of training neural networks in presence of noise and errors still remains open today, and might hold the key to orders of magnitude improvements in efficiency while maintaining high speed in training neural networks. Towards addressing this important intellectual question, this work provides the first unified strategy that addresses error-resilience in every operation during the training of Deep Neural Networks (DNNs), which are a form of Artificial Neural Networks (ANNs). 



ANNs, proposed in mid 1900s~\cite{rosenblatt1958perceptron,rumelhart1986learning}, have revolutionized modern machine learning and data mining. 
Training large-scale neural networks with millions of parameters~\cite{krizhevsky2012imagenet,taigman2014deepface,he2016deep} often requires large training time exceeding a few days. 
The ever-increasing size of neural networks creates a pressing demand for resources and power for fast and reliable training. In our experiments that appear later, we demonstrate that ignoring errors entirely during neural network training can severely degrade the performance\footnote{Interestingly, errors can sometimes be useful in non-convex minimization because they can help the computation to exit local minima and saddle points. However, this occurs only when the error-magnitude is small. Soft-errors can cause bit-flips in most-significant bits, and thus the corresponding error-magnitudes can be quite large.}, even with a probability of error as low as $3\times 10^{-4}$ (see Fig.~\ref{fig:accuracy_results}). Instead, by embracing errors in computing, one may be able to reduce the power budget of each individual computational node. And, use of system-level error-correction can provide high reliability on error-prone components as long as the overhead is kept small.

In this paper, we propose \textit{CodeNet}, a novel strategy that enables fast and reliable training of Deep Neural Networks (DNNs) in distributed and parallelized architectures that use unreliable processing components. We advance on ideas from information and coding theory to design novel error-correction mechanisms that use redundancy to compute reliably in presence of ``soft-errors,'' \textit{i.e.}, undetected errors that can corrupt the computation of a node, producing garbage outputs that are far from the true (noiseless) output. CodeNet is completely decentralized, allowing for all primary computational operations to be error-prone, including the nonlinear step which is an obstacle for coding in computing because most coding techniques are linear. Further, even error-detection and decoding are allowed to be erroneous and are accomplished in a decentralized manner, by introducing some very low-complexity verification steps that are assumed to be error-free. Error-prone detection and decoding is important because of two reasons: (a) it avoids having a single point-of-failure in the system (if central node fails, so does the algorithm) and allows for all computations to be error-prone including encoding/decoding/nonlinear activation/diagonal matrix post-multiplication; (b) it is also an important requirement for \textit{biological plausibility} of any neural computation algorithm (see e.g.~\cite{olshausen1996emergence}). Lastly, CodeNet also ensures that the additional overheads due to coding are kept low and comparable to replication, since matrices are not required to be encoded afresh even though they update at each iteration, and only vectors are encoded at each iteration which is much cheaper computationally.

\textbf{Summary of Theoretical and Experimental Results:} We first analytically characterize the worst-case error tolerance of CodeNet in Theorem~\ref{thm:recovery_threshold}. Then, we show in Theorem~\ref{thm:time} that, accounting for errors and checkpointing time, CodeNet can offer unboundedly large reductions in the expected training time under fixed storage per node over classical replication, while also requiring fewer redundant nodes. We also compare the computational and communication complexity of CodeNet with classical replication in Theorem~\ref{thm:complexity}, showing that the additional overhead due to coding is negligible using an efficient decentralized implementation with standard communication protocols~\cite{chan2007collective} across the nodes. 
Finally, our experiments on Amazon EC2 clusters (with artificially introduced errors) reveal that even seemingly small error probabilities can severely degrade the performance of classical non-error-resilient training. In specific instances, we observe that CodeNet offers $6$x speedup over replication under equal storage per node (see Fig.~\ref{fig:accuracy_results}) to complete the same number of iterations (and achieve the same accuracy) even while using fewer redundant nodes.
\begin{figure}
\centering
\includegraphics[height=8cm]{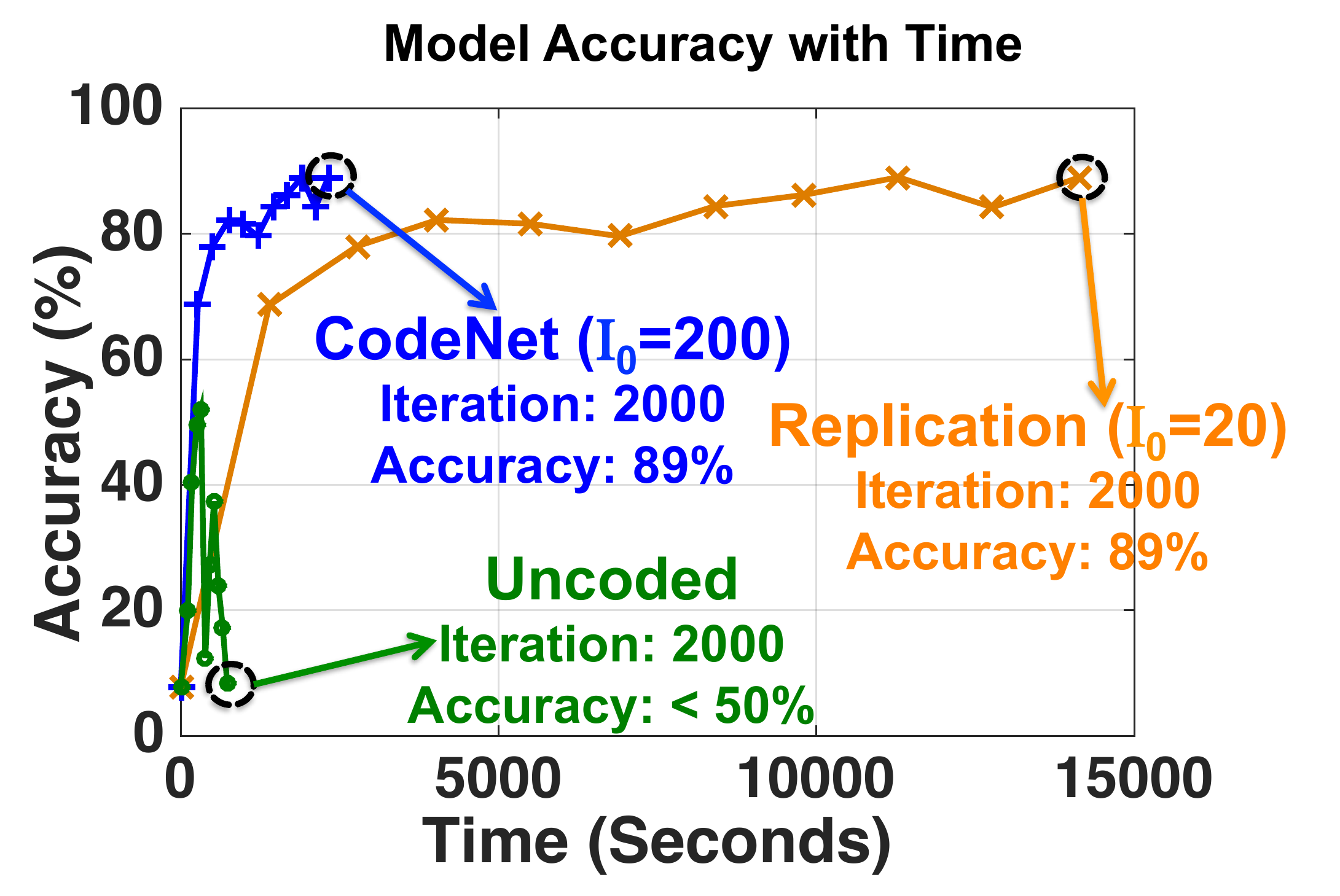}
\caption{Experimental results on Amazon EC2: we train a $3$ layer DNN of configuration $[784 \ 10^4 \ 10^4 \ 10]$ on MNIST dataset using backpropagation with Stochastic Gradient Descent (batch size $1$). We assume that each node can be affected by soft-errors with probability $0.0003$ independently, either during the matrix-vector product in feedforward or backpropagation stage, or during the rank-$1$ update in update stage at any layer, which happen to be the three most computationally intensive operations at any layer during training. CodeNet uses $38$ nodes and completes $2000$ iterations of training in $2322$ seconds, achieving an accuracy of $89\%$. A comparable replication strategy, with equal storage per node, uses more nodes ($40$) but takes much longer ($14140$ seconds) to complete $2000$ iterations and achieve the same accuracy. However, an Uncoded strategy with no error-correction has very poor accuracy. CodeNet thus improves the wall-clock runtime over comparable strategies while also providing acceptable performance. It enables forward computing when an error occurs and reduces the need to roll-back to the iteration of the last checkpoint. }
\label{fig:accuracy_results}
\end{figure}

\subsection{Related Works}
Because of their potential to go undetected, soft-errors are receiving increasing attention (sometimes even regarded as ``the scariest aspect of supercomputing’s monster in the closet'' \cite{geist2016supercomputing}). Common causes for soft-errors include: (i) Exposure of chips to cosmic rays from outer space causing unpredictable bit-flips; (ii) Manufacturing and soldering defects; and (iii) Memory and storage faults etc.~\cite{geist2016supercomputing,ziegler1996terrestrial}. Even for specialized nanoscale circuits, the \textit{International Technology Roadmap} for semiconductors predicts that as devices become smaller, thermal noise itself will cause systems to fail catastrophically during normal operation even without supply voltage scaling \cite{yanushkevich2013introduction}, thus increasing the need for fault-tolerant training.

Fault tolerance has been actively studied since von Neumann's  work~\cite{von1956probabilistic} (see~\cite{Tay_Bel_68,hadjicostis2005coding,Pip_TIT_91,Spi_FCS_96,yang2017ITTrans}). Existing techniques fall into two categories: roll-backward and roll-forward error correction. Roll-backward error correction refers to different forms of \textit{checkpointing}~\cite{faultbook}, where the computation-state is transmitted to, and stored in, a disk at regular programmer-defined intervals. When errors are detected, the last stored state is retrieved from the disk, and the computation is resumed from the previous checkpoint. However checkpointing comes with immense communication costs~\cite{faultbook}. Retrieving the state of the system from the disk is extremely time-intensive, and can significantly slow down the computation if the errors are frequent.

An alternative (and often complementary) approach is roll-forward error correction where redundancy is introduced into the computation itself, and detected errors are corrected prior to proceeding. Use of sophisticated (i.e., non-replication) based error-correcting codes in roll-forward error correction dates at least as far back as 1984, when Algorithm-Based-Fault-Tolerance (ABFT) was proposed by Huang and Abraham~\cite{ABFT1984} for certain linear algebraic operations. ABFT techniques \cite{ABFT1984,faultbook} mainly use parity checks to detect and sometimes correct faults, so that computation can proceed without the need to roll back, when the number of errors are limited. Here, we are interested in soft-errors in a completely decentralized setup, which impose further difficulties because \textit{the error correction mechanisms are themselves required to be {decentralized} and error-prone}.

Recently, ``Coded Computation'' \cite{NewsletterPaper,gauri2014delay,gauri2015straggler,gauri2014efficient,lee2016speeding,lee2018speeding,dutta2016short,dutta2017coded,lee2017matrix,dutta2018optimal,dutta2018unified,yu2017polynomial,GC1,GC2,dutta2018optimal,yu2018entangled,li2015coded,GC3,yang2016logistic,yang2017encoded,yang2016convolution,yang2016encoded,Salman1,Salman2,Salman3,Salman4,Emina1,Emina2,Virtualization,heterogeneousclusters,GC4,Suhas1,Suhas2,yang2017NIPS,Ramtin1,lee2017multicore,jeongFFT,baharav2018straggler,suh2017matrix,mallick2018rateless, wang2018coded,aliasgari2019coded, wang2018fundamental,severinson2018block,ye2018communication,haddadpour2018codes,haddadpour2018straggler,yang2018ISIT,ferdinand2016anytime,ferdinand2018hierarchical,song2017pliable,kosaian2018learning,seth2018bigdata,jeong2018locally} has emerged as an evolution on ABFT techniques to address the problem of straggling nodes, \textit{i.e.}, when a few delay-prone nodes delay the entire computation~\cite{straggler_tail} as the master node has to wait for all to finish. These works provide strategies that are, in some cases~\cite{dutta2016short,yu2017polynomial,dutta2018optimal,GC2}, scaling sense improvements over ABFT, and also obtain fundamental limits on resilience under given storage or communication constraints. Please refer to \cite{dutta2018DNN2,dutta2018unified} for our follow-up work on coded neural network training.

\begin{figure*}
\centering
\subfloat[Feedforward stage: Generates an estimated label.]{\includegraphics[height=2cm]{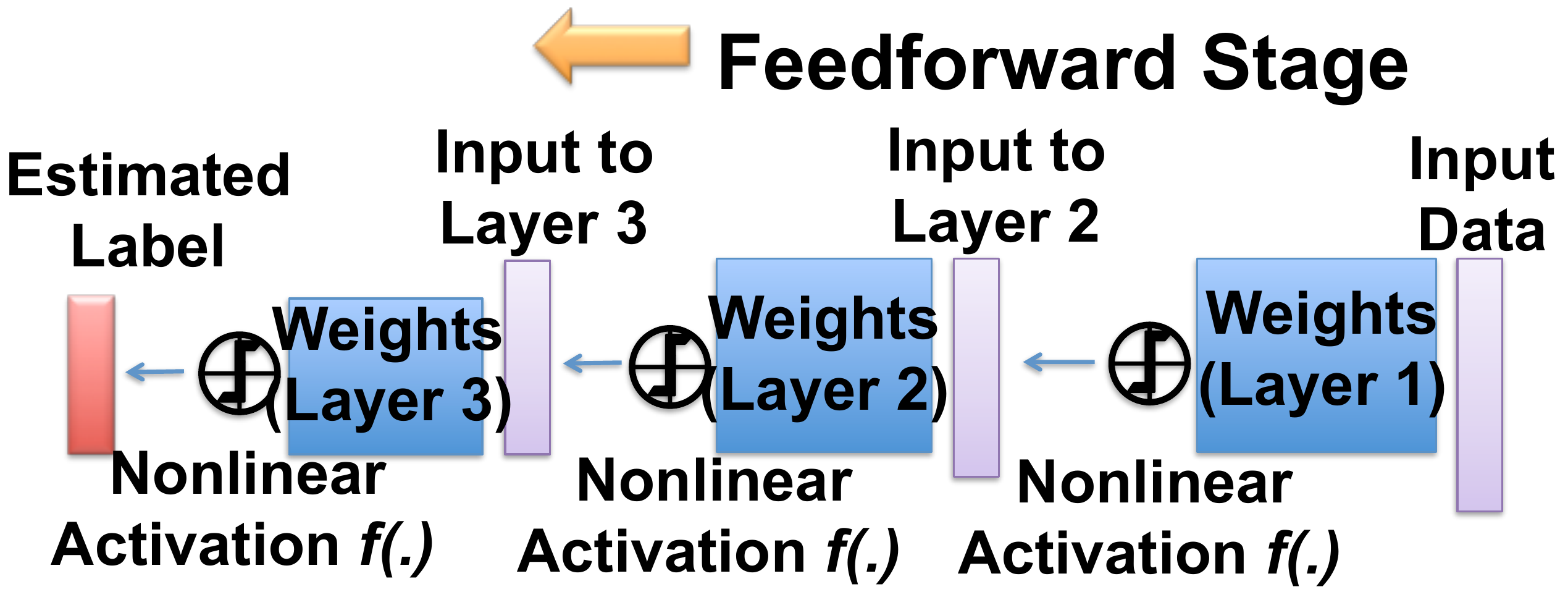}%
\label{fig:DNN_step1}}
\hspace{1.4cm}
\subfloat[Transition from feedforward to backpropagation stage at last layer.]{\includegraphics[height=2cm]{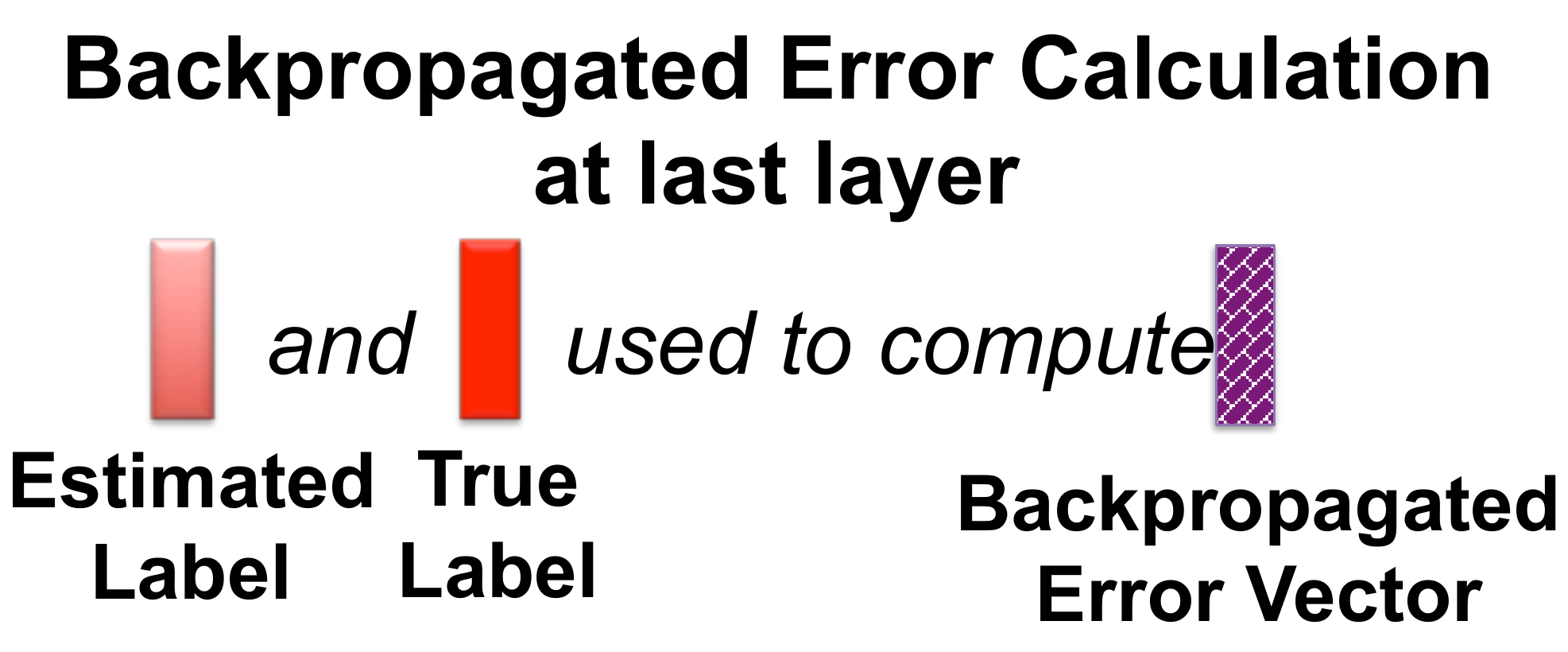}%
\label{fig:DNN_step2}}\\
\centering
\subfloat[Backpropagation stage: Error vector propagates backward.]{\includegraphics[height=2cm]{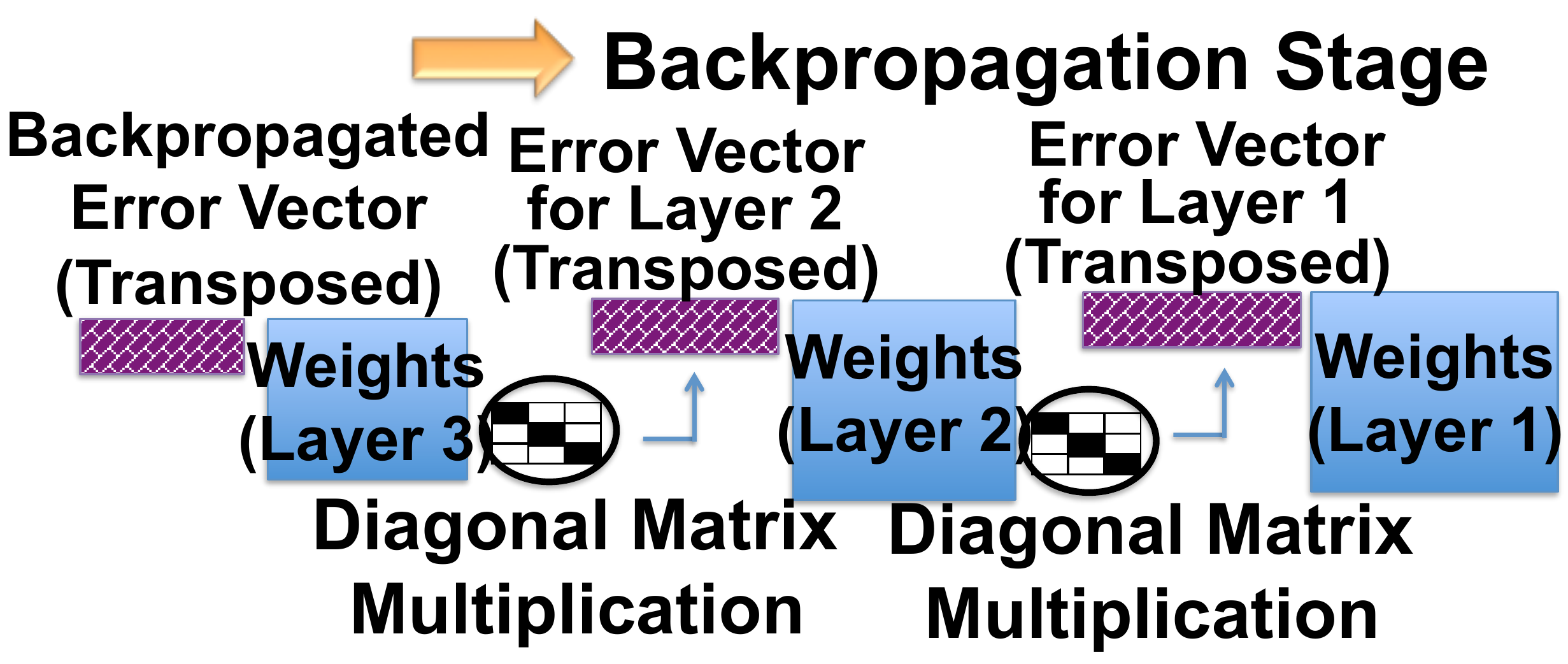}%
\label{fig:DNN_step3}}
\hspace{1.8cm}
\subfloat[Update stage: Rank $1$ updates on weight matrix. ]{\includegraphics[height=2cm]{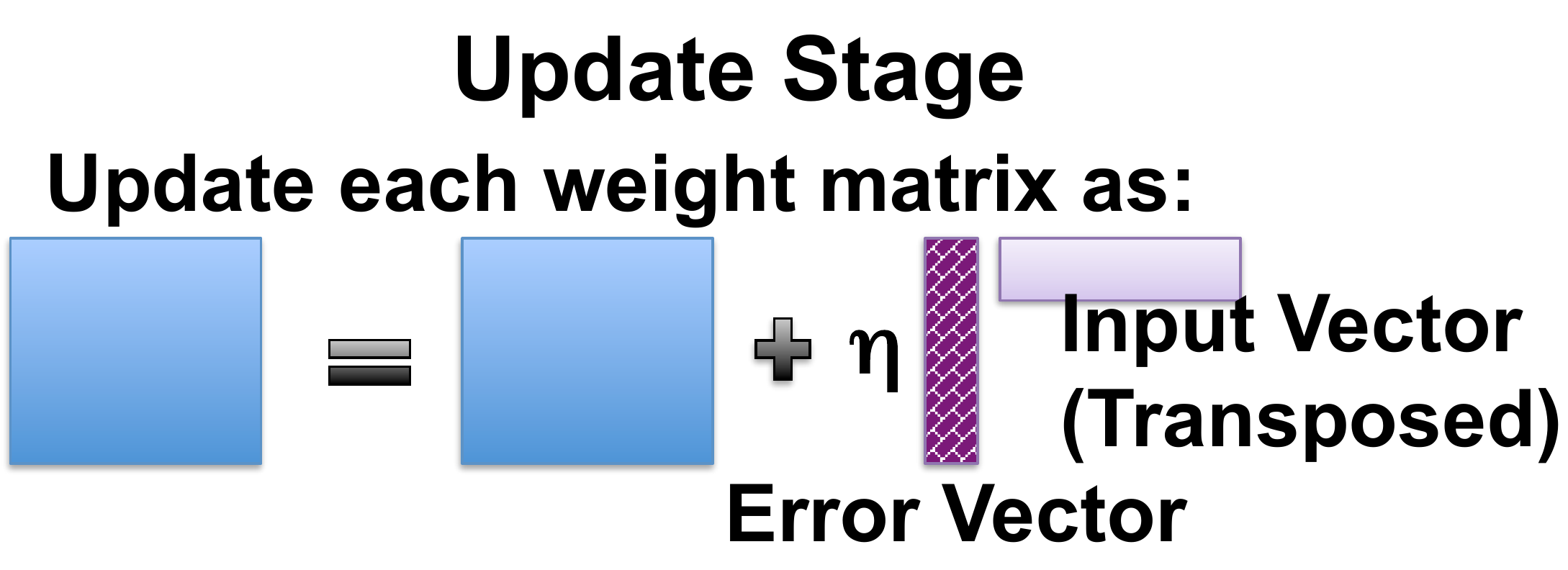}%
\label{fig:DNN_step4}}
\caption{DNN training: (a) Feedforward stage: the data vector is passed forward through all the layers (a matrix-vector product followed by an activation at each layer) producing an estimated label vector. (b) The backpropagated error vector is generated at the last layer comparing the estimated label with the true label. (c) Next, in the Backpropagation stage, the backpropagated error vector propagates backward across the layers (a matrix-vector product followed by a multiplication with a diagonal matrix) to generate the backpropagated error vector for every layer. (d) Finally, in the Update stage, each layer updates itself using its backpropagated error vector and its own input vector. \label{fig:DNN_training} }
\end{figure*}
\subsection{Key Novelties of CodeNet}
We are required to make all the stages of DNN training, namely, feedforward, backpropagation and update (see next section), resilient to errors. We note that for each of these stages, there are similarities that our strategy shares with the existing ABFT and coded computing literature (e.g.~\cite{ProductCodes}). The strategies for feedforward and backpropagation resilience are similar to those for matrix-vector products, and that for the update stage is similar to rank-$1$ updates that arise in some matrix multiplication algorithms (e.g., see~\cite{bouteiller2015algorithm,summa}). Our key contribution is in observing that a unified coding strategy, that advances on these existing strategies, can be weaved into the architecture of DNN training, that automatically makes all three stages error-resilient, without significant overhead. Furthermore, our strategy incorporates three key advances over the existing literature in coded computing and ABFT that are critical to achieving error-resilience in DNN training in practical systems:\\
(i) \textbf{Dealing with nonlinear activations:} An important conceptual difficulty is the nonlinearity of the element-wise activation function after matrix-vector computation in the feedforward stage. Because the most commonly used techniques of error-correcting codes are linear, nonlinear operations can prevent them from preserving codeword structures. Here, we first circumvent this problem by encoding the most computationally intensive operations~\cite{dally2015} (matrix-vector and matrix-update steps) separately at each layer. We demonstrate the utility of CodeNet under two error models, one where only the longer, most computationally intensive steps are vulnerable to errors and one where all the primary steps \textit{including the nonlinear activations} are allowed to be error-prone, by introducing some extremely low-complexity verification steps in the latter, which are assumed to be error-free.\\
(ii) \textbf{Decentralized architecture allowing for all primary operations to be error-prone:} Unlike much of the literature in coded computing (with few exceptions, e.g.~\cite{yang2017ITTrans}), we assume that there is \textit{no single master node} that is available to distribute tasks and collect outputs from parallel nodes. In centralized systems, the master node is a single point-of-failure: if it is {erroneous}, the entire computation can fail. Even in ABFT literature, where decentralization is common, it is assumed that encoding, error-detection and decoding are reliable. In this work, our implementation is completely decentralized, and allows all the primary steps, including the less computationally intensive steps, such as, encoding, error-detection, decoding, as well as, the  nonlinear activation and diagonal matrix post-multiplication, to incur errors. To achieve error-resilient detection and decoding, we only introduce two verification steps of very low-complexity as compared to the primary steps, that are assumed to be error-free.\\
(iii) \textbf{Reducing communication and computing overhead of coding:} Our careful choice of unified coding strategy enables low overhead of coding in two ways: (a) the same initial encoding on the weight matrices introduces resilience in the feedforward, backpropagation and update stages; (b) the overhead of this encoding, which can be prohibitively large for matrices that evolve across iterations, is kept low by performing \textit{encoded updates} that preserve the coded structure of weight matrices while only requiring encoding of \textit{vectors} used in the update stage\footnote{Straightfoward use of existing works on ABFT \cite{bouteiller2015algorithm,ABFT1984} and coded computing~\cite{dutta2016short,lee2018speeding,GC2} for a matrix-vector product (e.g.~ $\bm{W}_{N\times N} \bm{x}_{N \times 1}$) requires encoding $\bm{W}$ from scratch in each iteration. 
Because encoding, in scaling sense, is of the same computational complexity as the matrix-vector product itself, \textit{i.e.}, $\Theta(N^2)$, it can introduce significant overhead and unreliability.}. 

We compare with two natural competing strategies: (i) an \textit{uncoded} strategy that uses no redundancy. This strategy performs no error detection, and simply ignores soft-errors at every iteration. This can cause errors to keep accumulating, which can severely degrade the performance of the neural network, as demonstrated by our experimental results. (ii)  \textit{replication} of the entire network, followed by looking for mismatch between outputs of nodes that perform identical computations at every iteration to detect errors. While this provides fault-tolerance: (a) it requires high redundancy (a factor of 2); and (b) it can increase the training time substantially because the training has to return to the last checkpoint every time an error occurs and restart from that state. 
CodeNet also uses checkpointing when the the number of errors exceeds its error tolerance, but it requires much less frequent checkpointing as compared to replication. This is because CodeNet
detects and corrects errors without the expensive roll-back to previous checkpoint (when the number of errors are limited), and is thus able to proceed forward with the computation.

\section{Background and Problem Formulation}
In this section, we first discuss the primary computational operations that are required to be performed at each layer during error-free DNN training and then move on to the problem formulation and goals. The background we provide is limited; we refer to \Cref{appendix:DNN_background} for more details. Throughout the paper, we use bold-face letters to denote matrices and vectors. 

\subsection{Background: Primary Operations of DNN Training}
Assume that we are training a DNN with $l=1,2,\ldots,L$ layers (excluding the input layer that can be thought of as layer $0$) using backpropagation algorithm \cite{rumelhart1986learning} with Stochastic Gradient Descent (SGD) with a mini-batch size of $1$. The DNN thus consists of $L$ parameter matrices (also called weight matrices), one for each layer, as illustrated in Fig.~\ref{fig:DNN_training}. Let $N_{l}$ denote the number of neurons in the $l$-th layer. Thus, for layer $l$, the weight matrix to be trained is of dimension $N_{l} \times N_{l-1}$ as it represents the connections between the neurons of layer $l$ and $l-1$. Note that for $l=1$, the weight matrix is of dimension $N_1\times N_0$, where $N_0$ is the dimension of the input data vector. 
In every iteration, the DNN (\textit{i.e.} the $L$ weight matrices) is trained based on a single data point and its true label through three stages, namely, feedforward, backpropagation and update, as shown in Fig.~\ref{fig:DNN_training}. At the beginning of every iteration, the first layer (layer $1$) accesses the data vector (input to layer $1$) from memory, and starts the feedforward stage which propagates recursively across layers, from $l=1$ to $L$. For a generic layer, we denote the weight matrix and feedforward input to the layer as $\bm{W}^l$ and $\bm{x}^l$ respectively. The operations performed at layer $l$ during feedforward stage (see Fig.~\ref{fig:DNN_step1}) are:

\noindent \textbf{[Step O1]} Compute matrix-vector product: $\bm{s}^l=\bm{W}^l\bm{x}^l$. \\
\noindent \textbf{[Step C1]} Compute input for $(l+1)$-th layer as: $\bm{x}^{(l+1)}=f(\bm{s}^l)$, by applying nonlinear activation function $f(\cdot)$ element-wise.

At the last layer (layer $l=L$), the backpropagated error vector is
generated by accessing the true label from memory and the estimated label (see Fig.~\ref{fig:DNN_step2}), which is the final feedforward output of the last layer. Then, the backpropagation stage propagates recursively across layers, from $l=L$ to $1$, generating backpropagated error vectors for the next layer. Let the backpropagated error vector for a layer be denoted by $\bm{\delta}^l$. The operations in the backpropagation stage (see Fig.~\ref{fig:DNN_step3}) for $l$-th layer are:

\noindent \textbf{[Step O2]} Compute matrix-vector product: $(\bm{c}^l)^T=(\bm{\delta}^l)^T\bm{W}^l$. \\
\noindent \textbf{[Step C2]} Compute backpropagated error vector for $(l-1)$th layer as: $(\bm{\delta}^{(l-1)})^T=(\bm{c}^l)^T\bm{D}^l$, where $\bm{D}^l$ is a diagonal matrix whose $i$-th diagonal element is a function $g(\cdot)$ of the $i$-th element of $\bm{x}^l$, such that $g(f(u))=f'(u)$ for the chosen nonlinear activation function $f(\cdot)$ in the feedforward stage.

Finally, every weight matrix moves on to the update stage (see Fig.~\ref{fig:DNN_step2}) as follows:

\noindent \textbf{[Step O3]} Perform a rank-$1$ update: $\bm{W}^l \leftarrow \bm{W}^l + \eta \bm{\delta}^l(\bm{x}^l)^T$ where $\eta$ is the learning rate.

\subsection{Desirable Parallelization Schemes} We are interested in completely decentralized, model-parallel architectures where each layer is parallelized across multiple nodes (that can be reused across layers) because the nodes cannot store the entire matrix $\bm{W}^l$ for each layer. Because the steps O1, O2 and O3 are the most computationally intensive steps at each layer, we restrict ourselves to schemes where these three steps for each layer are parallelized across the $P$ nodes. A small number of additional nodes can be used as redundant nodes for error resilience.

\noindent \textbf{Remark:} We consider model-parallel strategies where during an iteration at any layer, error detection or correction will happen only twice, \textit{i.e.}, at step C1 (after O1) and step C2 (after O2). As there is no communication immediately after O3, all errors in step O3 will be detected and corrected in the next iteration at that layer, when that erroneous node produces an output again for the first time for that layer, either after step O1 or O2. Thus, in the worst case, we should be able to correct $(t_1+t_3)$ errors after step O1 and $(t_2+t_3)$ errors after step O2 under Error Models $1$ or $2$. Moreover, under Error Model $2$, we should be able to detect if more errors have occurred even if we cannot correct them.

\subsection{Assumptions}
\paragraph{Data and Label Access} We assume that there is a source or shared memory from where all nodes can access the data (for the first layer, during feedforward stage) and its label (for the last layer, while transitioning from feedforward stage to backpropagation stage). 

\paragraph{Error Models} We use two error models:

\noindent \textbf{Error Model $1$ (Worst-case model):} \textit{Any} node can be affected by soft-errors but \textit{only} during the steps O1, O2 and O3, which are the most computationally intensive operations in DNN training\cite{dally2015}. There are no errors in encoding, error-detection, decoding, nonlinear activation or diagonal matrix post-multiplication, which require negligible time and number of operations\footnote{In general, the shorter the computation, the lower is the probability of soft-errors. E.g., the occurrence of soft-errors is assumed to be a Poisson process in~\cite{li2007memory}, so that the number of soft-errors in a time interval becomes a Poisson random variable with mean proportional to the interval length.} because most of the time and resources are spent on steps O1, O2 and O3. The total number of nodes in error at any layer during O1, O2 and O3 are bounded by $t_1$, $t_2$ and $t_3$ respectively. There is no assumption on the distribution of the errors for this model.\\
\noindent \textbf{Error Model $2$ (Probabilistic model):} \textit{Any} node can be affected by soft-errors during any primary operation (including encoding/error-detection/decoding/nonlinear activation/diagonal matrix post-multiplication), and there is no upper bound on the number of errors. For conceptual simplicity, the output of an erroneous node is assumed to be the correct output corrupted by an additive continuous valued random noise. However, we introduce two verification steps under Error Model $2$ to check for errors during error-detection and decoding, that have very low complexity as compared to any other primary step, and hence, are assumed to be error-free.\\
\noindent \textbf{Remark:} Error Model $1$ is a ``worst-case'' abstraction which is used when it is difficult to place probabilistic priors on errors. Error Model $2$, while allowing for errors in all of these operations, also makes simplifying assumptions. In particular, the continuous distribution of noise simplifies our analyses by avoiding complicated probability distributions that arise in finite number of bits representations. This simplification can  lead to somewhat optimistic conclusions, e.g., it allows us to detect the occurrence of errors (``garbage outputs'') in a coded computation with probability $1$ even if they are too many errors to be corrected (because the noise takes any specific value with probability zero; see [Lemma~\ref{lem:error_tolerance},  \Cref{appendix:error_detection}]). This model is increasingly accurate in the limit of large number of bits of precision. In practical implementations, our results that hold with probability $1$ should be interpreted as holding with high probability (e.g. it is unlikely, but possible, that two erroneous nodes produce the same garbage output). Further, we note that, because both replication and coding are able to exploit Error Model $2$ for error-detection, it does not bias our results towards CodeNet relative to replication.

\paragraph{Error-free checkpoint} We allow for checkpointing the entire DNN (\textit{i.e.}, saving the weight matrices for every layer) at a disk periodically, so that we can revert to the last checkpoint when errors can be detected but not corrected, under Error Model $2$. As in most existing literature, we assume this process of checkpointing is error-free. However, the cost of checkpointing is very high as compared to the computation or communication costs of the algorithm otherwise, as one has to access the disk. Thus, we aim to employ checkpointing less frequently.

\subsection{Goal} Our goal is to design a unified, model-parallel  DNN training strategy resilient to errors, that uses $P$ ``base'' nodes, along with few redundant nodes, and checkpointing, such that: every node can store only a $\frac{1}{P}$ fraction of the number of elements of \textit{each} of the $L$ weight matrices. Thus, the storage per node for layer $l$ is $\frac{N_l N_{l-1}}{P}$, with negligible extra storage of $o\big(\frac{N_l N_{l-1}}{P} \big)$ to store vectors, e.g., $\bm{x}^l$, $\bm{\delta}^l$.

Essentially, we are required to perform all the operations of DNN training given $P$ ``base'' nodes (i.e., the minimal number of nodes needed by the zero-redundancy ``uncoded'' counterpart) with minimum number of additional redundant nodes to be able to correct errors at each layer. Additionally, it is desirable that the additional communication complexities as well computational overheads,~e.g., encoding, error-detection, decoding etc. are kept as low as possible. It is preferable that both the computational and communication costs are comparable to classical replication in scaling sense.

\section{Main Results}

In this section, we first provide our theoretical results that highlight the three benefits of our proposed strategy CodeNet over classical replication strategy with error-detection and checkpointing, namely, (A) Error-tolerance with lower resource overhead; (B) Speedup in computation time; and, (C) Comparable computational and communication costs. Our findings are complemented with experimental results.

\subsection{Error-tolerance with Lower Resource Overhead}
The error-tolerance of our proposed strategy CodeNet is as follows:

\begin{thm}[Error Tolerance] 
\label{thm:recovery_threshold} Let $m$ and $n$ be two integers such that $mn=P$. The CodeNet strategy uses a total of  $\hat{P}=P+2n(t_1+t_3)+2m(t_2+t_3)$ nodes to detect and correct \textit{any} $t_1+t_3$ erroneous nodes after step O1, and \textit{any} $t_2+t_3$ erroneous nodes after step O2, at a single layer during an iteration, under both the Error Models $1$ and $2$. Moreover, under Error Model $2$, if there are more errors, it is able to detect the occurrence of errors with probability $1$ even though it cannot correct them.  
\end{thm}

\noindent When $t_1+t_3=t_2+t_3=t$, the total number of nodes required to be able to correct \textit{any} $t$ errors after steps O1 or O2 is thus $\hat{P}=P+2(m+n)t$, that again reduces to $\hat{P}=P+4\sqrt{P}t$ for $m=n=\sqrt{P}$. We include a proof of this result in \Cref{appendix:error_detection}. Throughout the rest of the paper, we will primarily restrict ourselves to the simpler (and symmetric) case of $t_1+t_3=t_2+t_3=t$, though the strategy easily generalizes. In the Materials and Methods section, we describe in detail the proposed CodeNet strategy (with toy examples for $t=1$), which will also provide intuition on the proof of Theorem~\ref{thm:recovery_threshold}. 

\noindent \textbf{The Comparable Replication Strategy:} In the replication strategy, we assume that the system uses a total of $2P$ nodes, where each node has the same memory limitations, \textit{i.e.}, it can store only a fraction $\frac{1}{P}$ of each of the weight matrices. For every layer, the matrix $\bm{W}^l$ is divided across a grid of $m\times n$ base nodes (such that $mn=P$) and a replica of this entire grid is created using the other set of $P$ nodes. After every matrix-vector product, the replicas performing the same computation exchange their computational outputs and compare. If the outputs do not match, the system detects an error and resumes computation from the iteration of its last checkpoint, retrieving the saved DNN from disk. 

\noindent \textbf{Comparison with CodeNet:} First observe that the replication strategy requires $2P$ nodes as compared to CodeNet which uses $P+4\sqrt{P}$ nodes (for the case of $t=1$). And more importantly, in spite of using more resources, the replication strategy is only able to detect and not correct any errors, while CodeNet can detect as well as correct errors. For error correction in the replication strategy, we would need to have more than $2$ replicas to be able to do a majority voting, but this would require even higher resource overhead.

\subsection{Speedup in Computation Time}
\label{sec:main results}

\begin{figure}[!ht]
\centering
\includegraphics[height=6cm]{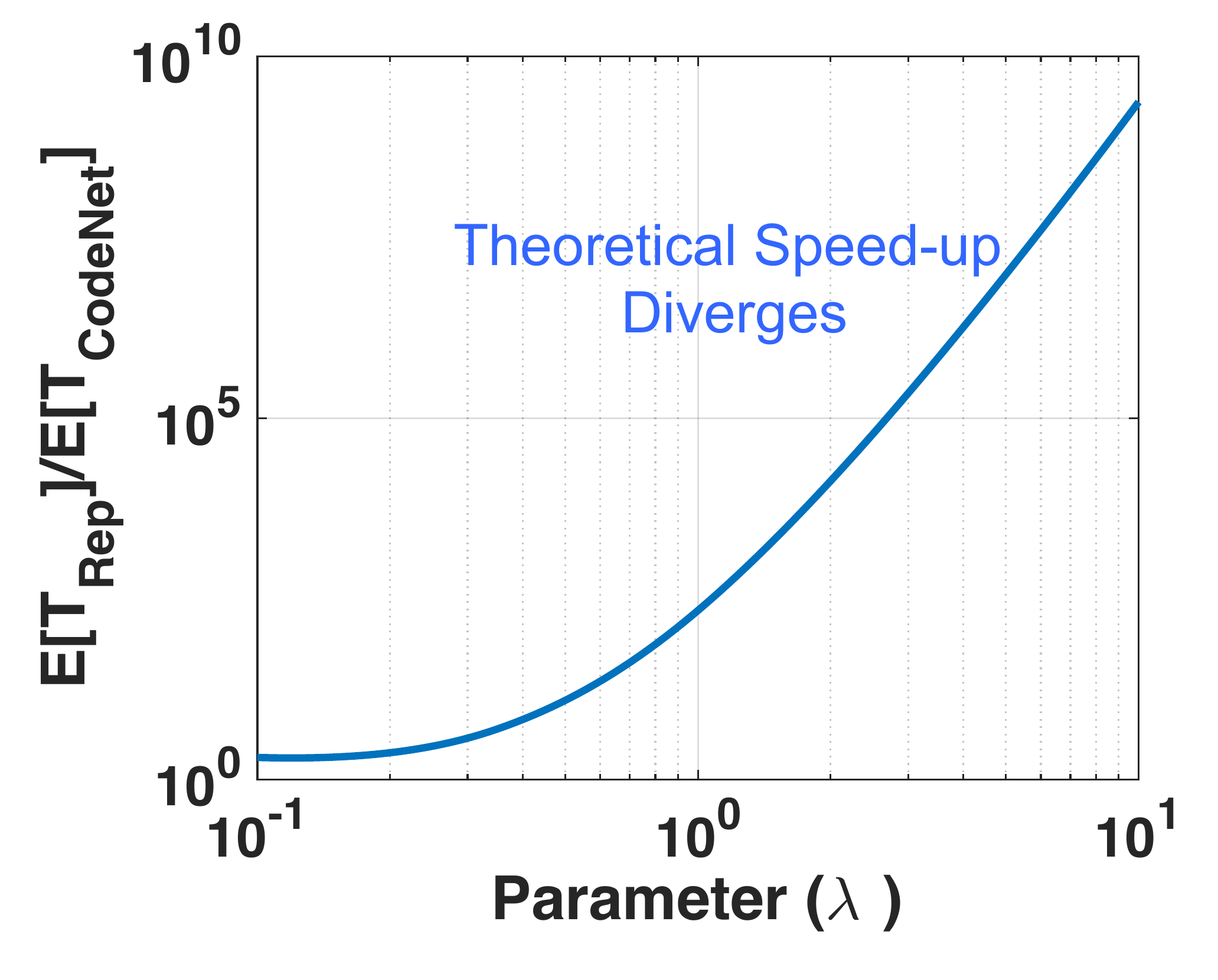}
\caption{Theoretical Plot: We assume the total number of soft-errors in each iteration to follow a Poisson distribution \cite{li2007memory} with parameter $\lambda$. For $t=1$, thus $p_0=e^{-\lambda}$ and $p_1 \geq \lambda e^{-\lambda}$. We also fix $\tau_{cpt}=\tau_{b}=1000$ and $\tau_{f}=1$ and vary $\lambda$ from $0.1$ to $10$. For each case, we also vary the period of checkpointing, \textit{i.e.}, $I_0$ and choose the $I_0$ that minimizes the expected time. The plot shows that as $\lambda$ scales, the ratio of the expected time of replication with CodeNet diverges to infinity. }
\label{fig:tradeoff}
\end{figure}

Here we show that CodeNet has significant speedups in expected computation time over replication that scales as the error probability scales. This is because every time a single error occurs under Error Model $2$, the replication based strategy rolls backward to the iteration of its last checkpoint while CodeNet is able to correct errors and proceed forward to the next iteration. 

For this theorem, we make some assumptions: 

\noindent (i) Errors (under Error Model $2$) at any node may or may not have dependence on other nodes. Error events in an iteration are simply divided into three disjoint sets: (1) zero errors (probability $p_0$); (2) error patterns correctable by CodeNet (probability $p_1$); and
(3) error patterns not correctable by CodeNet (probability $p_2=1-p_0-p_1$). E.g., if CodeNet can correct any $t$ errors, then $p_1$ includes the probability of all error-patterns with at most $t$ errors. Observe that, for the errors captured by $p_1$, CodeNet proceeds forward to the next iteration after regenerating the corrupt sub-matrices, but the replication strategy simply reverts to the iteration of its last checkpoint. For error patterns captured by $p_2$, both strategies have to revert to their last checkpoint. \\ 
\noindent (ii)  Period of checkpointing (number of iterations after which we checkpoint), $I_0$ is fixed, but can vary for different strategies. \\
\noindent (iii)  Time taken by both replication and CodeNet for an error-free iteration is $\tau_f$. We justify this in Theorem~\ref{thm:complexity} by showing that the communication and computation complexities of replication and CodeNet are comparable. Denoting the time to resume from previous checkpoint (by reading from disc) by $\tau_b$, we assume that $\tau_b \gg \tau_f $ as fetching data from the disk is extremely time-intensive. We also pessimistically assume that the time to run an iteration with error correction and regeneration by CodeNet be $\tau_b$. When an error occurs, CodeNet proceeds forward in computation, but it also regenerates some of the sub-matrices of $\bm{W}^l$. This regeneration comes with extra communication cost, increasing the time of the iteration. But because it does not require reading from the disk, letting it be as high as $\tau_b$ is a actually a pessimistic assumption in evaluation of the performance of our strategy. \\
\noindent (iv) Time to save the entire state using checkpointing is $\tau_{cpt}$.

\begin{thm}
\label{thm:time}
The ratio of the expected time taken to complete $M$ iterations by replication to CodeNet scales as 
\vspace{-0.3cm}
\begin{align*}\frac{E[T_{Rep}]}{E[T_{CodeNet}]}=\frac{ \min_{I_0} \frac{M}{I_0}\tau_{cpt} + \frac{M}{I_0}(\tau_f p_0 + \tau_b(1-p_0)) \frac{\frac{1}{(p_0)^{I_0}}-1}{\frac{1}{(p_0)}-1}}{ \min_{I_0} \frac{M}{I_0}\tau_{cpt} + \frac{M}{I_0}(\tau_f p_0 + \tau_b(1-p_0)) \frac{\frac{1}{(p_0+p_1)^{I_0}}-1}{\frac{1}{(p_0+p_1)}-1}}. 
\end{align*}
\end{thm}

The proof is provided in \Cref{appendix:runtime}. In Fig.~\ref{fig:tradeoff} we show how this ratio of theoretical expected times can diverge to infinity as $\frac{p_1}{p_0}$ scales. In particular, if the number of errors in an iteration follow Poisson distribution\cite{li2007memory} with parameter $\lambda$, then the ratio of expected time scales with $\lambda$. Again, if any node fails independently with probability $p$ whenever it is used for one of the $3$ most complexity-intensive ($\Theta(N_lN_{l-1}/P)$) computations (O1, O2, or O3), in each layer, in an iteration, then $p_0=(1-p)^3\hat{P}L$ and $p_1 \geq\binom{3\hat{P}L}{1}p^1(1-p)^{3\hat{P}L-1}$ where $\hat{P}$ is the total number of nodes used and $3L$ is the number of times a node is used in an iteration, \textit{i.e.}, thrice for each layer.

\subsection{Comparable Computational and Communication Overheads}
Next, we compare the computational and communication complexity of CodeNet and replication in an error-free iteration. The result has two major implications: firstly, it shows that the computational and communication complexity of both CodeNet and replication are comparable in a scaling sense and hence, justifies the assumption that $\tau_f$ might be assumed to be the same for both. Secondly, it also shows that all other computational steps in an error-free iteration like addition of partial results (all-reduce), consistency check contribute negligible overhead when compared with the per-node computational complexity.

\begin{thm}
For $m=n=\sqrt{P}$, the following hold for each layer in a single error-free iteration. $(i)$ The ratio of the communication complexity of CodeNet to replication is less than:
 
$$ \frac{6\alpha\log{(\hat{P})} + \beta(6t+3)\frac{(N_l + N_{l-1})}{\sqrt{P}} +2\beta\hat{P}t}{2\alpha\log{(P)} + 2\beta\frac{(N_l + N_{l-1})}{\sqrt{P}}} $$ \text{ which scales as } $\mathcal{O}(3t)$ as $P,\hat{P}, N_l, N_{l-1}\to \infty$ in the regime $\hat{P}^{3/2}=o(\min\{N_l, N_{l-1} \} ) $.\\ $(ii)$ The ratio of the computational complexity of CodeNet to replication is less than:
 
$$ \frac{6 \frac{N_l N_{l-1} }{P} + (1+4t)\left(\frac{N_l}{\sqrt{P} }+\frac{N_{l-1}}{\sqrt{P}}\right) + 2\hat{P}t }{6 \frac{N_l N_{l-1} }{P} +  \left(\frac{N_l}{\sqrt{P} }+\frac{N_{l-1}}{\sqrt{P}}\right) } $$
which is $\Theta(1)$ as $P, N_l, N_{l-1}\to \infty$ with $\hat{P}^{3/2}=o( \min \{N_l,\ N_{l-1} \})$.
\label{thm:complexity}
\end{thm}

The proof is provided in \Cref{appendix:complexity}.

\subsection{Experimental Results}
\label{sec:experiments}

\begin{table}[!ht]
\centering
\caption{Comparison of Accuracy and Runtime for different strategies after $2000$ iterations}
\begin{tabular}{p{2.2cm}p{6cm}p{2.4cm}p{2.2cm}}
Strategy & Checkpointing Period $I_0$ (No. of Iterations)  & Accuracy $(\%)$ & Runtime (Seconds) \\
\midrule
CodeNet & 200 & 89 &  2322\\
Replication & 10 & 89 & 18453 \\
Replication & 20 & 89 & 14140\\
Replication & 30 & 89 & 19238\\
Uncoded & N.A. & <50 & 730\\
\bottomrule
\end{tabular}
\label{table:experiments}
\end{table}
We perform experiments on Amazon EC2 clusters (m3.medium instances) using Starcluster to test the performance of CodeNet on a three layer DNN. The dimensions of the weight matrices are $10^4 \times 784$, $10^4 \times 10^4$ and $10 \times 10^4$ respectively. We train this DNN on the MNIST dataset\cite{lecun1998mnist}. We divide each of these matrices individually into smaller sub-matrices over a $5 \times 4$ grid. For CodeNet, we add two more rows of $5$ nodes each and two more columns of $4$ nodes each for error tolerance, resulting in a total of $20+2\times(5+4)=38$ nodes. For comparison, we also implemented  the competing replication strategy using $40$ nodes (greater than CodeNet), with equal storage per node. For the replication strategy, all the three matrices of dimensions $10^4 \times 784$, $10^4 \times 10^4$ and $10 \times 10^4$  are also divided across a $5 \times 4$  grid, \textit{i.e.}, across $20$ nodes and another $20$ nodes are used to store a replica. We also implemented an uncoded strategy, which simply ignores soft-errors entirely. To compare uncoded and CodeNet assuming equal number of nodes, we allow the uncoded strategy to use either $20$ nodes or all $40$ nodes (while utilizing less storage per node in the latter case, and hence less computational complexity per node). However, for both  uncoded strategies, the experimental behavior is similar. It finishes its iterations faster, but the accuracy is very poor as it ignores soft-errors. Thus, uncoded has a poor performance in accuracy-time tradeoff. {In Fig.~\ref{fig:experiments}, we only show the case using equal nodes, \textit{i.e.}, $40$ nodes. The matrices are divided across these 40 nodes across a $5 \times 8$  grid of nodes. } 

We also assume that every node is affected by a soft-error with probability $3 \times 10^{-4}$ during the three most computationally intensive steps: matrix-vector products of the feedforward/backpropagation stage (O1/O2) and the rank-$1$ update (O3). When a soft-error occurs, we add a random sparse matrix (density $0.005$) to the stored sub-matrix of $\bm{W}^l$ at any node whose non-zero values are drawn from a uniform distribution between $-5$ and $5$. Unlike our theoretical model, we choose a sparse matrix-error here so that only some elements of $\bm{W}^l$ are corrupted by soft-errors. In the event of error, this reduces the extent to which convergence of the training is affected for the uncoded strategy (\textit{i.e.}, this assumption is more favorable to the uncoded strategy in comparison with our ``garbage-output'' model). The performance of coded and replication strategies remains the same because they correct any errors before proceeding. While pessimistic for our results, this different model helps us better understand the effect of small errors in convergence rate of the uncoded strategy. However, even with this favorable modeling for uncoded, the convergence of the strategy is substantially affected to the degree that it does not seem to converge at all.

\begin{figure}[ht]
\centering
\includegraphics[height=6cm]{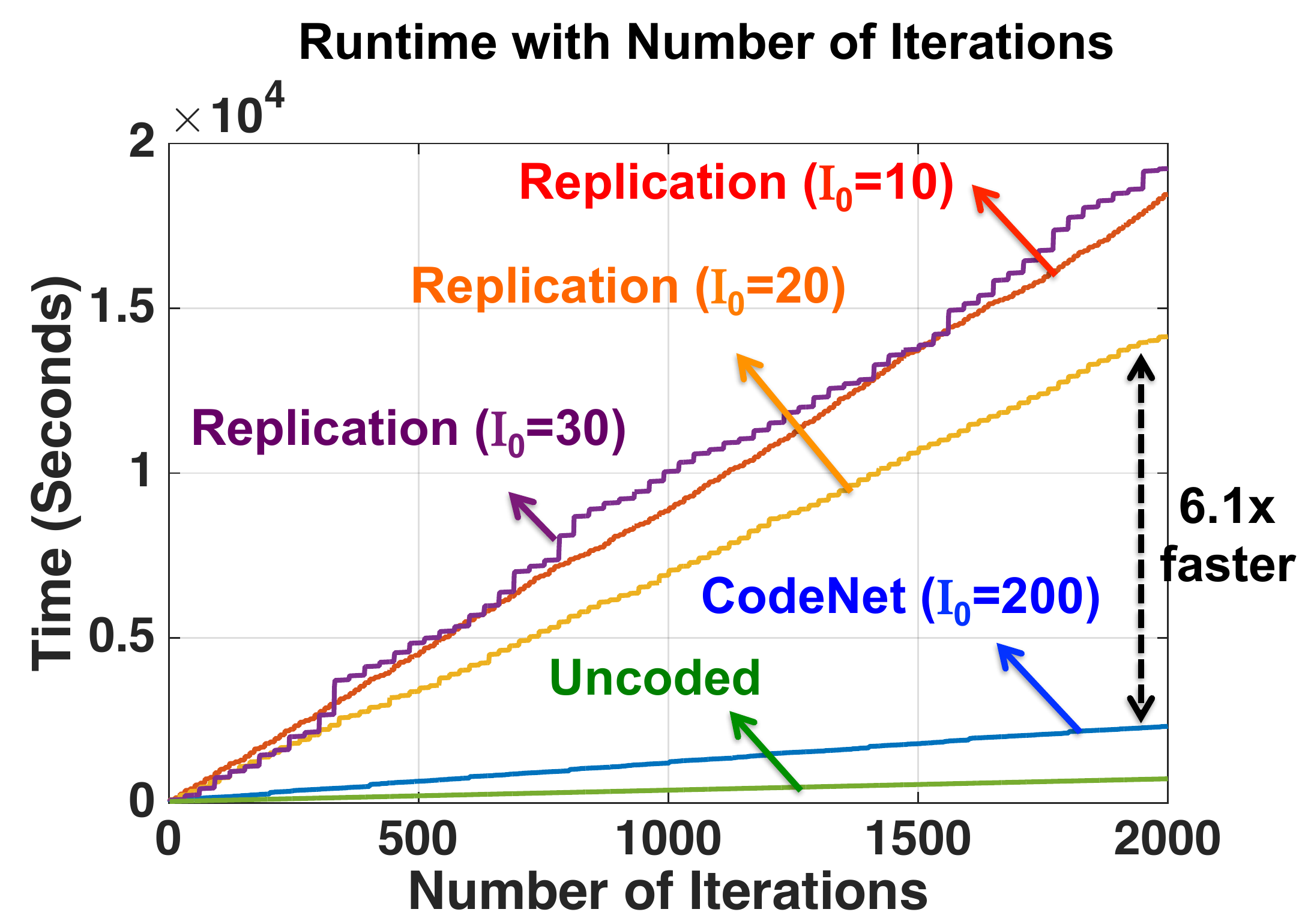}
\caption{The experiments compare the runtime and performance of three strategies: Uncoded, Replication and CodeNet. In this plot, we illustrate the wall-clock runtime required by the three strategies to complete $2000$ iterations of training. For replication, we consider three values of the period of checkpointing, \textit{i.e.} $I_0 = 10,20 \text{ and } 30$. Of these, $I_0=20$ gives the lowest values of runtime. We observe that CodeNet (with $I_0=200$) offers $6.1\times$ speed-up compared to Replication with $I_0=20$. For every plot, there are steep jumps at iteration indices that are multiples of $I_0$ as every time an error occurs, the system reverts to the iteration where it checkpointed last. Uncoded is slightly faster as it proceeds with forward computing ignoring errors, but its performance is much worse as we highlight in Table \ref{table:experiments}. }
\label{fig:experiments}
\end{figure}

Our results are shown in Fig.~\ref{fig:experiments}. The experiments show that ignoring soft-errors entirely (as in Uncoded) is not a good idea as it severely degrades the convergence. Alternately, using replication slows down the runtime substantially as every time an error occurs, the system has to revert to the last checkpoint and start afresh from that iteration. CodeNet has the best accuracy-runtime tradeoff as it can correct and proceed forward in case of a single soft-error, and needs to revert to the last checkpoint only in case of more errors. For the replication strategy, we also varied the period of checkpointing, but CodeNet still outperform replication with the best chosen checkpoint period by a factor of more than $6$.


\section{The Proposed CodeNet Strategy}
In this section, we describe in detail our proposed strategy -- CodeNet -- that is based on a class of error-correcting codes called Maximum Distance Separable (MDS) codes~\cite{ryan2009channel}. As mentioned before, a naive extension of existing coded computing or ABFT techniques would require encoding of updated $\bm{W}^l$ at each layer, in every iteration. By carefully choosing a unified coding strategy, we are able to exploit the \textit{check-sum invariance}~\cite{bouteiller2015algorithm} of DNN training and update operations. This allows the matrix $\bm{W}^l$ to be encoded only once at the beginning of the training and obviates the need to encode $\bm{W}^l$ afresh at every iteration. Instead, at every iteration, we encode the feedforward input \textit{vectors} ($\bm{x}^l$) and the backpropagated error \textit{vectors} ($\bm{\delta}^l$), and perform coded updates that still maintain the coded structure of initial $\bm{W}^l$ across updates, adding only a negligible overhead to the overall computation. 

\subsection{Some Notations} 

We choose two integers $m$ and $n$ such that $mn=P$. We use the notation $\bm{W}^l_{i,j}$ to denote the sub-matrix (block of size $\frac{N_l}{m} \times \frac{N_{l-1}}{n}$) of the matrix $\bm{W}^l$, at block-index $(i,j)$, when $\bm{W}^l$ is block-partitioned in $2$ dimensions across an $m\times n$ base grid. When we consider only the horizontal or vertical partitioning of $\bm{W}^l$, we use the notation $\bm{W}^l_{i,:}$ or $\bm{W}^l_{:,j}$ respectively, to denote the entire $i$-th row block (of size $\frac{N_l}{m} \times N_{l-1}$) or the entire $j$-th column block (of size $  N_l \times \frac{N_{l-1}}{n} $) respectively. E.g., 
\begin{align*}  \bm{W}^l & = \begin{bmatrix}
            \bm{W}^l_{0,0} & \ldots &\bm{W}^l_{0,n-1} \\
            \vdots & \ddots & \vdots \\
            \bm{W}^l_{m-1,0} & \ldots  &\bm{W}^l_{m-1,n-1}
          \end{bmatrix} = \begin{bmatrix}
            \bm{W}^l_{0,:}   \\
            \vdots \\
            \bm{W}^l_{m-1,:} 
          \end{bmatrix} \Bigg\rbrace \text{Row Blocks} \\
          & = \underbrace{\begin{bmatrix}
             \ \bm{W}^l_{:,0} \ & \ldots &  \ \bm{W}^l_{:,n-1} \      \end{bmatrix}}_{\text{Column Blocks}} .
\end{align*}

\noindent We will use \textit{Systematic} MDS Codes~\cite{ryan2009channel} to encode these sub-matrices. The generator matrix of a \textit{systematic} $(b,a)$ MDS Code with $b\geq a$ is a matrix of dimension $a \times b$ such that: (i) The first $a$ columns form an $a \times a$ identity matrix; (ii) Any $a$ out of $b$ columns are linearly independent. 

The coded blocks (or parity blocks) arising due to encoding of the sub-matrices of size $\frac{N_l}{m} \times \frac{N_{l-1}}{n}$ in the grid that are denoted by $\widetilde{\bm{W}}^l_{i,j}$ where $(i,j)$ denotes their index in the grid. Note that, these parity blocks are redundant blocks that will lie outside the $m\times n$ grid of base nodes. The total number of nodes used by a strategy, including base nodes and redundant nodes, is denoted by $\hat{P}$.

Similarly, the vectors $\bm{x}^l$ and $\bm{c}^l$ are also divided into $n$ equal sub-vectors, each of length $\frac{N_{l-1}}{n}$ while the vectors $\bm{\delta}^l$ and $\bm{s}^l$ are divided into $m$ equal sub-vectors, each of length $\frac{N_{l}}{m}$ respectively. A superscript $\widetilde{\cdot}$ denotes a coded sub-vector (parity sub-vector). E.g.
$\bm{x}^l=\begin{bmatrix}
\bm{x}^l_0 \\
\vdots\\
\bm{x}^l_{n-1}
\end{bmatrix}$  and $ \widetilde{\bm{x}}^l_j$  denotes the additional coded or parity sub-vectors for index $j \ > \ m-1$. 
 
\begin{figure}
\centering
\includegraphics[height=2.7cm]{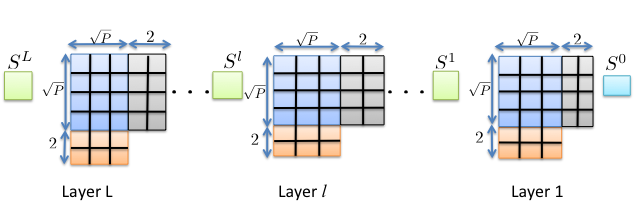}
\caption{The node layout for CodeNet: The same $\hat{P}=P+4\sqrt{P}$ nodes are used in the grid for every layer. The nodes $S^0$, $S^1$,\ldots,$S^L$ are error-free (virtual) nodes introduced here for the ease of explanation, but the actual implementation is decentralized with no single source of failure as described in details later.}
\label{fig:DNN_node_layout}\vspace{-0.2cm}
\end{figure}

\subsection{Layout of Nodes}

We use a set of $P$ base nodes arranged in an $m \times n$ grid. We add $2(t_1+t_3)$ redundant rows of $n$ nodes each for error correction after step O1, and add $2(t_2+t_3)$ redundant columns of $m$ nodes each for error correction after step O2. For simplicity and symmetry, let us choose $t_1+t_3=t_2+t_3=t$. In Fig.~\ref{fig:DNN_node_layout}, we illustrate the node layout for the special case of $m=n=\sqrt{P}$ and $t=1$, thus using $4\sqrt{P}$ redundant nodes.

{For every layer, the weight matrix $\bm{W}^l$ of dimension $N_{l} \times N_{l-1}$ is divided into $P$ equal sub-matrices across a grid of $m \times n$ base nodes, thus satisfying the storage constraint (each node can only store $\frac{N_{l} N_{l-1} }{P}$ elements of the weight matrix for each layer). This is illustrated in Fig.~\ref{fig:DNN_node_layout}, again for the special case of $m=n=\sqrt{P}$ and $t=1$. We will also add redundant rows and columns of redundant nodes to obtain error resilience.}

For conceptual simplicity, between any two consecutive layers of the DNN, let us first assume that there is an error-free ``virtual'' node. These error-free virtual nodes  will \textit{not} be used in our actual algorithm, where we will allow all operations to be error prone under Error Model $2$, including encoding, error-detection, decoding, nonlinear activation and diagonal matrix post-multiplication as we will clarify afterwards.

In Fig.~\ref{fig:DNN_node_layout}, the virtual node between layers $l$ and $l-1$, \textit{i.e.}, node $S^{l-1}$ acts as the sink (which also performs error-check/decoding) for the previous $(l-1)$-th layer and the source (which also performs encoding) for the $l$-th layer in the feedforward stage. In the backpropagation stage, the flow of data is reversed. So, now the virtual node $S^{l-1}$ collects partial results and acts as the error-detector/decoder for the $l$-th layer, and the encoder for the $(l-1)$-th layer. Errors in the update stage of layer $l$ are detected and corrected after step O1 or O2 of the next iteration at layer $l$, when the erroneous node produces an output for the first time for that layer.

\subsection{Strategy Description (Using Virtual Nodes)}

\subsubsection{Pre-processing: Initial Encoding of weight matrices once prior to the start of training}
All the weight matrices ($\bm{W}^l$ of dimension $N_l \times N_{l-1}$ at the $l$-th layer) are first initialized with a random initial value at the start of training (same as the uncoded algorithm). For CodeNet, this initial matrix is then encoded and stored in appropriate nodes for the first iteration. For all subsequent iterations, CodeNet only encodes and decodes vectors instead of matrices. Surprisingly, since the code is designed to be that way, we will show that the weight matrices are able to update themselves in each iteration, while maintaining their coded structure, using just these coded vectors. 

Let us consider an example with $t=1$ and $P=4$ base nodes, arranged in a $2 \times 2$ grid. Assume that $\bm{W}^l$ is block-partitioned row-wise and column-wise to get $2 \times 2$ blocks of dimension $\frac{N_l}{2}  \times \frac{N_{l-1}}{2}$ each, as follows:
$
\bm{W}^l= \begin{bmatrix}
            \bm{W}^l_{0,0} & \bm{W}^l_{0,1} \\
            \bm{W}^l_{1,0} & \bm{W}^l_{1,1}
          \end{bmatrix}.
$
Now, we add $2$ redundant rows and $2$ redundant columns of processing nodes, each containing $\frac{N_l}{2}  \times \frac{N_{l-1}}{2}$ sized coded blocks (or parity blocks) for a single error correction ($t=1$) under Error Models $1$ or $2$. 

Observe the encoding for the redundant rows of processing nodes. Each coded block is an independent linear combination of all the blocks of $\bm{W}^l$ in the same column in the grid. For example, here we choose 
$\widetilde{\bm{W}}^l_{2,j} = \bm{W}^l_{0,j} + \bm{W}^l_{1,j}$, and 
$\widetilde{\bm{W}}^l_{3,j}= \bm{W}^l_{0,j} - \bm{W}^l_{1,j}$ for all $j=0,1$. Here $\bm{W}^l_{0,j}$ and $\bm{W}^l_{1,j} $ denote the first and second blocks of $\bm{W}^l$, while $\widetilde{\bm{W}}^l_{2,j}$ and $\widetilde{\bm{W}}^l_{3,j}$ denote the two coded blocks for column $j$ in the grid.  

The encoding for the redundant columns of nodes is also similar. Each coded block is an independent linear combination of all the blocks of $\bm{W}^l$ in the same row. Here, we choose
$\widetilde{\bm{W}}^l_{i,2}= \bm{W}^l_{i,0} + \bm{W}^l_{i,1} $ and 
$\widetilde{\bm{W}}^l_{i,3}= \bm{W}^l_{i,0} + 2\bm{W}^l_{i,1} $ for $i=0,1$. Here, $\widetilde{\bm{W}}^l_{i,2}$ and $\widetilde{\bm{W}}^l_{i,3}$ denote the two coded blocks for row $i$ in the grid. Rewriting in matrix notation, the sub-matrices that are stored in the grid can be written as:
\begin{align*}
&
\begin{bmatrix}
\bm{W}^l_{0,0} & \bm{W}^l_{0,1} & \widetilde{\bm{W}}^l_{0,2} & \widetilde{\bm{W}}^l_{0,3} \\
\bm{W}^l_{1,0} & \bm{W}^l_{1,1} & \widetilde{\bm{W}}^l_{1,2} & \widetilde{\bm{W}}^l_{1,3}  \\
\widetilde{\bm{W}}^l_{2,0} & \widetilde{\bm{W}}^l_{2,1} & \text{x} & \text{x} \\
\widetilde{\bm{W}}^l_{3,0} & \widetilde{\bm{W}}^l_{3,1} & \text{x} & \text{x}
\end{bmatrix} = \big(\bm{G}_r^T  \otimes \bm{I}_{\frac{N_{l}}{2}} \big) \bm{W}^l \big( \bm{G}_c \otimes \bm{I}_{\frac{N_{l-1}}{2}} \big)\nonumber \\ &=
\bigg( \begin{bmatrix}
1 & 0  \\
0 & 1  \\
1 & 1 \\
1 & - 1 
\end{bmatrix} \otimes \bm{I}_{\frac{N_{l}}{2}} \bigg)
\begin{bmatrix}
\bm{W}^l_{0,0} & \bm{W}^l_{0,1} \\
\bm{W}^l_{1,0} & \bm{W}^l_{1,1} \\
\end{bmatrix}
\big(\begin{bmatrix}
1 & 0  & 1 &  1 \\
0 & 1 & 1 &  2  
\end{bmatrix}\otimes \bm{I}_{\frac{N_{l-1}}{2}}  \big).
\end{align*}
Here x denotes that there is no node in that location, and the corresponding sub-matrix is not required to be computed. Observe that the matrices $\bm{G}_r=\begin{bmatrix}
1 & 0  & 1 & 1 \\
0 & 1  & 1 & -1
\end{bmatrix} $ and $\bm{G}_c= \begin{bmatrix}
1 & 0  & 1 &  1 \\
0 & 1 & 1 &  2  
\end{bmatrix}$ that are used here to generate the linearly independent linear combinations, are actually the generator matrices of two systematic $(4,2)$ MDS codes. Also $\bm{I}_{N_l/2}$ and $\bm{I}_{N_{l-1}/2}$ denote two identity matrices of sizes $N_l/2$ and $N_{l-1}/2$ respectively and, $\otimes$ denotes the Kronecker product of two matrices. The Kronecker product is required as the operations on $\bm{W}^l$ are performed block-wise instead of element-wise. We will again return to this example and explain how this technique can help detect and correct errors due to the MDS code.

\begin{clm}
\label{clm:general_encoding}
For the general case, to be able to correct \textit{any} $t$ errors at each layer after steps O1 or O2, under Error Models $1$ or $2$, the initial encoding is as follows:
\begin{align*}
\left(\bm{G}_r^T  \otimes \bm{I}_{N_{l}/m}  \right) \bm{W}^l \left( \bm{G}_c \otimes \bm{I}_{N_{l-1}/n}\right).
\end{align*}
Here $\bm{G}_r$ and $\bm{G}_c$ are the generator matrices of a systematic $(m+2t, m)$ MDS code and a systematic $(n+2t, n)$ MDS code respectively. The blocks of $\bm{W}^l$ are of size $\frac{N_{l}}{m} \times \frac{N_{l-1}}{n} $, hence the Kronecker product with $\bm{I}_{N_{l}/m} $ and $\bm{I}_{N_{l-1}/n}$ respectively.
\end{clm}

We formally justify this claim in \Cref{appendix:error_detection}.

\noindent \textbf{Remark:}  Recall that, this encoding of  $\bm{W}^l$ is error-free and is done \textit{only once} prior to the start of training. For all subsequent iterations, we will show that the nodes can perform coded updates that maintain the coded structure of $\bm{W}^l$ without the need to encode the entire matrix afresh at every iteration. Thus the cost of encoding $\bm{W}^l$ initially is amortized as we train the network over several iterations.

\subsubsection{Feedforward stage on a single layer}
The feedforward stage  consists of computing a matrix-vector product $\bm{s}^l=\bm{W}^l\bm{x}^l$ (step O1) followed by an element-wise nonlinear operation $f(\bm{s}^l)$. The key idea of error correction (inspired from \cite{lee2018speeding,ABFT1984}) is illustrated in the following example: let us examine only the horizontal partitioning of $\bm{W}^l$, \textit{i.e.},   $\bm{W}^l = \begin{bmatrix}
       \bm{W}^l_{0,0}& \bm{W}^l_{0,1}\\
       \bm{W}^l_{1,0} & \bm{W}^l_{1,1}
       \end{bmatrix} = \begin{bmatrix}
       \bm{W}^l_{0,:}\\
       \bm{W}^l_{1,:}
       \end{bmatrix}$. The MDS coding discussed before, \textit{i.e.},
$\widetilde{\bm{W}}^l_{2,j} = \bm{W}^l_{0,j} + \bm{W}^l_{1,j} $ and 
$\widetilde{\bm{W}}^l_{3,j}= \bm{W}^l_{0,j} - \bm{W}^l_{1,j} $ for all $j=0,1$ actually results in two redundant row blocks which are linearly independent combinations of  $\bm{W}^l_{0,:}$ and $\bm{W}^l_{1,:}$, given by
$\widetilde{\bm{W}}^l_{2,:}= \bm{W}^l_{0,:} + \bm{W}^l_{1,:} $ and 
$\widetilde{\bm{W}}^l_{3,:}= \bm{W}^l_{0,:} - \bm{W}^l_{1,:} $ respectively. Consider these computations:
$$  \begin{bmatrix}
\bm{W}^l_{0,:}\\
       \bm{W}^l_{1,:}\\
       \widetilde{\bm{W}}^l_{2,:} \\
       \widetilde{\bm{W}}^l_{3,:}
\end{bmatrix} \bm{x}^l = \begin{bmatrix}
\bm{W}^l_{0,:}\\
       \bm{W}^l_{1,:}\\
       \bm{W}^l_{0,:} + \bm{W}^l_{1,:} \\
       \bm{W}^l_{0,:} - \bm{W}^l_{1,:}
\end{bmatrix} \bm{x}^l = \begin{bmatrix}
\bm{s}^l_0\\
\bm{s}^l_1\\
\bm{s}^l_0+\bm{s}^l_1\\
\bm{s}^l_0-\bm{s}^l_1
\end{bmatrix} = \begin{bmatrix}
\bm{s}^l_0\\
\bm{s}^l_1\\
\widetilde{\bm{s}}^l_2\\
\widetilde{\bm{s}}^l_3
\end{bmatrix}. $$
We claim that $ \bm{s}^l(=\bm{W}^l\bm{x}^l)$ can still be successfully decoded from these $4$ outputs if one of them is erroneous, under Error Models $1$ and $2$. More generally, the claim is as follows:
\begin{clm} 
\label{clm:error_tolerance}
Using a systematic $(m+2t,m)$ MDS code to encode the horizontally-split blocks of $\bm{W}^l$, the result $ \bm{s}^l=\bm{W}^l\bm{x}^l$ can still be successfully decoded from the $m+2t$ computations $\bm{W}^l_{i,:}\bm{x}^l$ for $0 \leq i \leq m-1$ or  $\widetilde{\bm{W}}^l_{i,:}\bm{x}^l$ for $m \leq i \leq m+2t-1$, in presence of \textit{any} $t$ errors, under both the Error Models $1$ and $2$.
\end{clm}

We formally show this in \Cref{appendix:error_detection}. For an intuition, we return to the example with $m=n=2$ and $t=1$. First let us assume that only one of the $4$ outputs 
 $\bm{s}^l_0, \bm{s}^l_1, \widetilde{\bm{s}}^l_2$ and $\widetilde{\bm{s}}^l_3 $ is erroneous, \textit{i.e.}, corrupted with an additive random noise. As an example, say $\bm{s}^l_1 \rightarrow \bm{s}^l_1 + \bm{e}_1$ where $\bm{e}_1$ is the noise. Consider all subsets of size $3$ as follows: $ \{\bm{s}^l_0, \bm{s}^l_1 + \bm{e}_1, \widetilde{\bm{s}}^l_2 \}  $, $ \{\bm{s}^l_0,\widetilde{\bm{s}}^l_2, \widetilde{\bm{s}}^l_3 \}  $, $\{\bm{s}^l_0, \bm{s}^l_1 + \bm{e}_1, \widetilde{\bm{s}}^l_3\}$ and $\{\bm{s}^l_1 + \bm{e}_1, \widetilde{\bm{s}}^l_2,  \widetilde{\bm{s}}^l_3 \}$. For each subset, one can perform a consistency check for errors as follows:
 \begin{align*}
  \{\bm{s}^l_0, \bm{s}^l_1 + \bm{e}_1, \widetilde{\bm{s}}^l_2 \} &: \; \text{Is } \bm{s}^l_0 + \bm{s}^l_1 + \bm{e}_1 = \widetilde{\bm{s}}^l_2 ? \\
  \{\bm{s}^l_0,\widetilde{\bm{s}}^l_2, \widetilde{\bm{s}}^l_3\}&: \; \text{Is }  2\bm{s}^l_0 = \widetilde{\bm{s}}^l_2 + \widetilde{\bm{s}}^l_3 ? \\
  \{\bm{s}^l_0, \bm{s}^l_1 + \bm{e}_1, \widetilde{\bm{s}}^l_3 \}&: \;\text{Is } \bm{s}^l_0 -(\bm{s}^l_1 + \bm{e}_1)= \widetilde{\bm{s}}^l_3 ?\\
  \{\bm{s}^l_1 + \bm{e}_1, \widetilde{\bm{s}}^l_2,  \widetilde{\bm{s}}^l_3 \}&: \; \text{Is } 2(\bm{s}^l_1 + \bm{e}_1)= \widetilde{\bm{s}}^l_2 -\widetilde{\bm{s}}^l_3 ?
 \end{align*}

If there is only one erroneous output, one of the four subsets of size $3$ would pass the consistency check (in this case $\{\bm{s}^l_0,\widetilde{\bm{s}}^l_2, \widetilde{\bm{s}}^l_3 \}$) and the three correct outputs would be determined. From this subset, one can correctly decode  $\bm{s}^l_0$ and $\bm{s}^l_1= \widetilde{\bm{s}}^l_2-\bm{s}^l_0 $ which constitute $\bm{s}^l$.

\noindent \textbf{General Error-Detection Method:} Instead of checking all possible subsets, a more efficient and general technique to detect errors is to first perform $2t$ consistency checks using the parity check matrix of $\bm{G}_r$. This can be done, as is standard, by examining a full-rank matrix $\bm{H}$ of dimension $2t \times (m+2t)$ such that $\bm{H} \bm{G}_r^T=\bm{0}$. Thus, first we pre-multiply $(\bm{H} \otimes \bm{I}_{N_l/m})$ with the output as a check. If there are no erroneous blocks, this result will be zero under both Error Models $1$ and $2$, as we elaborate in \Cref{appendix:error_detection}. Otherwise, we will perform decoding,~e.g. using decoding algorithms based on sparse-reconstruction~\cite{candes2005decoding}.

Our next claim is that under Error Model $2$, if there is more than $1$ erroneous output, it is still possible to detect the occurrence of error with probability $1$, even though the error cannot be corrected. More generally, our claim is as follows:

\begin{clm} 
\label{clm:error_detection}
If the erroneous outputs are corrupted by an additive noise whose every element is drawn independently from real-valued continuous distribution (Error Model $2$), then it is still possible to detect the occurrence of erroneous outputs with probability $1$, even if their number is more than $t$.
\end{clm}

We formally show this in \Cref{appendix:error_detection}. For the particular example with $m=n=2$ and $t=1$, this means that if there is more than one erroneous output, then all four of the consistency checks will fail with probability $1$. Thus, the system will be able to detect the occurrence of errors, even though it might not be able to correct them.

With this key idea in mind, we now describe the strategy for the feedforward stage (with $m=n=2$ and $t=1$). First, assume that the virtual node $S^{l-1}$ has the vector $\bm{x}^l$ required at a particular iteration for layer $l$. We will justify this assumption at the end of this subsection. We also assume that the sub-matrices (blocks) of the updated $\bm{W}^l$ for the current iteration are also available at the appropriate nodes
\footnote{We encoded the $\bm{W}^l$ matrix and stored it in this manner before the first iteration. We will also show that in the backpropagation stage, all the nodes are able to update their sub-matrices without requiring encoding of matrices at every iteration, and thus the updated sub-matrix or coded sub-matrix of $\bm{W}^l$ is available at every node, prior to each new iteration.}. For the matrix-vector product $\bm{W}^l\bm{x}^l$, we only use the $8$ nodes containing the following sub-matrices (blocks):
\begin{align*}
\begin{bmatrix}
\bm{W}^l_{0,0} & \bm{W}^l_{0,1}  \\
\bm{W}^l_{1,0} & \bm{W}^l_{1,1}  \\
\widetilde{\bm{W}}^l_{2,0} & \widetilde{\bm{W}}^l_{2,1} \\
\widetilde{\bm{W}}^l_{3,0} & \widetilde{\bm{W}}^l_{3,1}
\end{bmatrix} &=
\left( \begin{bmatrix}
1 & 0  \\
0 & 1  \\
1 & 1 \\
1 & - 1 
\end{bmatrix} \otimes \bm{I}_{N_{l}/2} \right)
\begin{bmatrix}
\bm{W}^l_{0,0} & \bm{W}^l_{0,1} \\
\bm{W}^l_{1,0} & \bm{W}^l_{1,1} \\
\end{bmatrix}
\nonumber \\
& = \left(\bm{G_r^T}  \otimes \bm{I}_{N_{l}/2}  \right) \bm{W}^l. 
\end{align*}

For the feedforward stage, each of these $8$ sub-matrices, \textit{i.e.}, $\bm{W}^l_{i,j}$ or $\widetilde{\bm{W}}^l_{i,j}$, are available in $8$ nodes beforehand (laid out in a $4 \times 2$ grid). Observe that,     
$$ \begin{bmatrix}
\bm{s}^l_0 \\
\bm{s}^l_1\\
\widetilde{\bm{s}}^l_2\\
\widetilde{\bm{s}}^l_3
\end{bmatrix} =   \begin{bmatrix}
\bm{W}^l_{0,:}  \\
\bm{W}^l_{1,:}  \\
\widetilde{\bm{W}}^l_{2,:}\\
\widetilde{\bm{W}}^l_{3,:}
\end{bmatrix}   \bm{x}^l =  \begin{bmatrix}
           \bm{W}^l_{0,0}  & \bm{W}^l_{0,1} \\
     \bm{W}^l_{1,0}  & \bm{W}^l_{1,1}  \\
   \widetilde{\bm{W}}^l_{2,0}  &    \widetilde{\bm{W}}^l_{2,1}  \\
     \widetilde{\bm{W}}^l_{3,0}  &   \widetilde{\bm{W}}^l_{3,1} 
          \end{bmatrix}  \begin{bmatrix}
          \bm{x}^l_0\\\bm{x}^l_1
          \end{bmatrix}=
          \begin{bmatrix}
           \bm{W}^l_{0,0} \bm{x}^l_0 + \bm{W}^l_{0,1} \bm{x}^l_1\\
     \bm{W}^l_{1,0} \bm{x}^l_0 + \bm{W}^l_{1,1} \bm{x}^l_1  \\
   \widetilde{\bm{W}}^l_{2,0} \bm{x}^l_0 +    \widetilde{\bm{W}}^l_{2,1} \bm{x}^l_1 \\
     \widetilde{\bm{W}}^l_{3,0} \bm{x}^l_0 +   \widetilde{\bm{W}}^l_{3,1} \bm{x}^l_1 
          \end{bmatrix}.
$$
\noindent Each of the $8$ nodes only requires  either $\bm{x}^l_0$ or $\bm{x}^l_1$ to compute the $8$ small matrix-vector products (e.g. $\bm{W}^l_{0,0} \bm{x}^l_0$)  in parallel. The results are then added along the horizontal dimension to compute $\{\bm{s}^l_0 
,\bm{s}^l_1,
\widetilde{\bm{s}}^l_2,
\widetilde{\bm{s}}^l_3 \}$  in parallel and sent to the virtual node $S^{l}$ for error detection through parity checks (consistency checks)\footnote{Our actual implementation is decentralized with no virtual nodes. These consistency checks are performed using efficient collective communication protocols (All-Reduce)~\cite{chan2007collective} among the relevant nodes, without communicating all the sub-vectors to any particular node, that we will elaborate further when we discuss Decentralized Implementation.}. Now, even if any one of $\{\bm{s}^l_0 
,\bm{s}^l_1,
\widetilde{\bm{s}}^l_2,
\widetilde{\bm{s}}^l_3 \}$ is corrupted by soft-errors, the error-detector/decoder $S^{l}$ can still decode $\bm{s}^l$, apply the nonlinear activation $f(\cdot)$ and generate the feedforward input for the next layer. This also justifies our initial assumption that the virtual node has the feedforward input for each layer at the beginning of feedforward stage in that layer. If there is more than $1$ erroneous output, the virtual node $S^{l}$ detects the occurrence of errors with probability $1$, and reverts the system to the last \textit{checkpoint}. The feedforward stage is illustrated in Fig.~\ref{fig:DNN_feedforward}, where the nodes not used for computation are faded. The steps are as follows: 
\begin{figure}[t]
\centering
\subfloat[Step $1$: Node $S^{l-1}$ multi-casts appropriate portions of $\bm{x}^l$ to the corresponding column of nodes.]{\includegraphics[height=2.6cm]{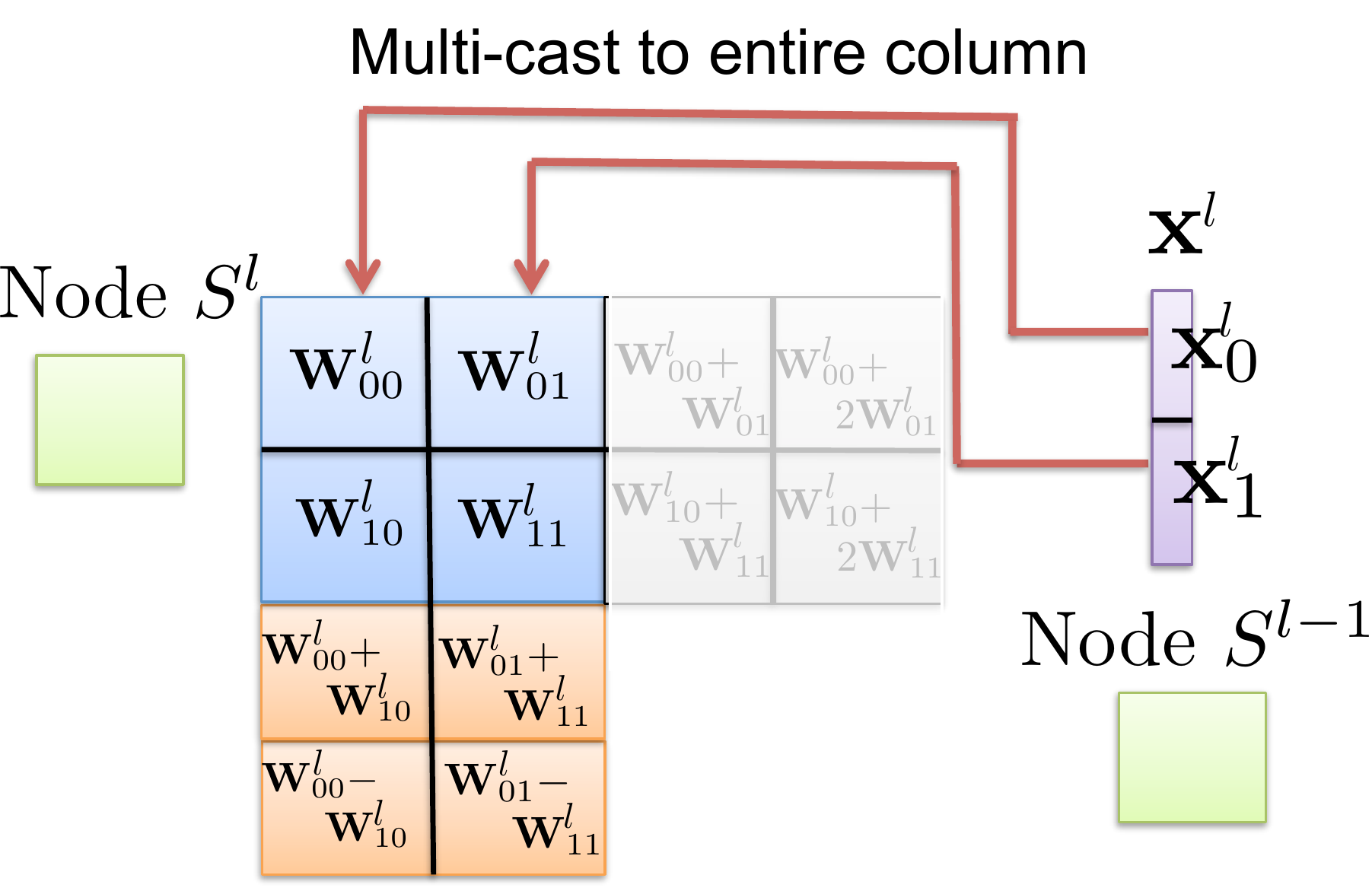}%
\label{fig:DNN_feedforward1}}
\hspace{0.5cm}
\subfloat[Step $2$: Each active node performs individual computations, \textit{i.e.},  $\bm{W}_{i,j}^l\bm{x}^l_j $ or $\widetilde{\bm{W}}^l_{i,j}\bm{x}^l_j $.]{\includegraphics[height=2.6cm]{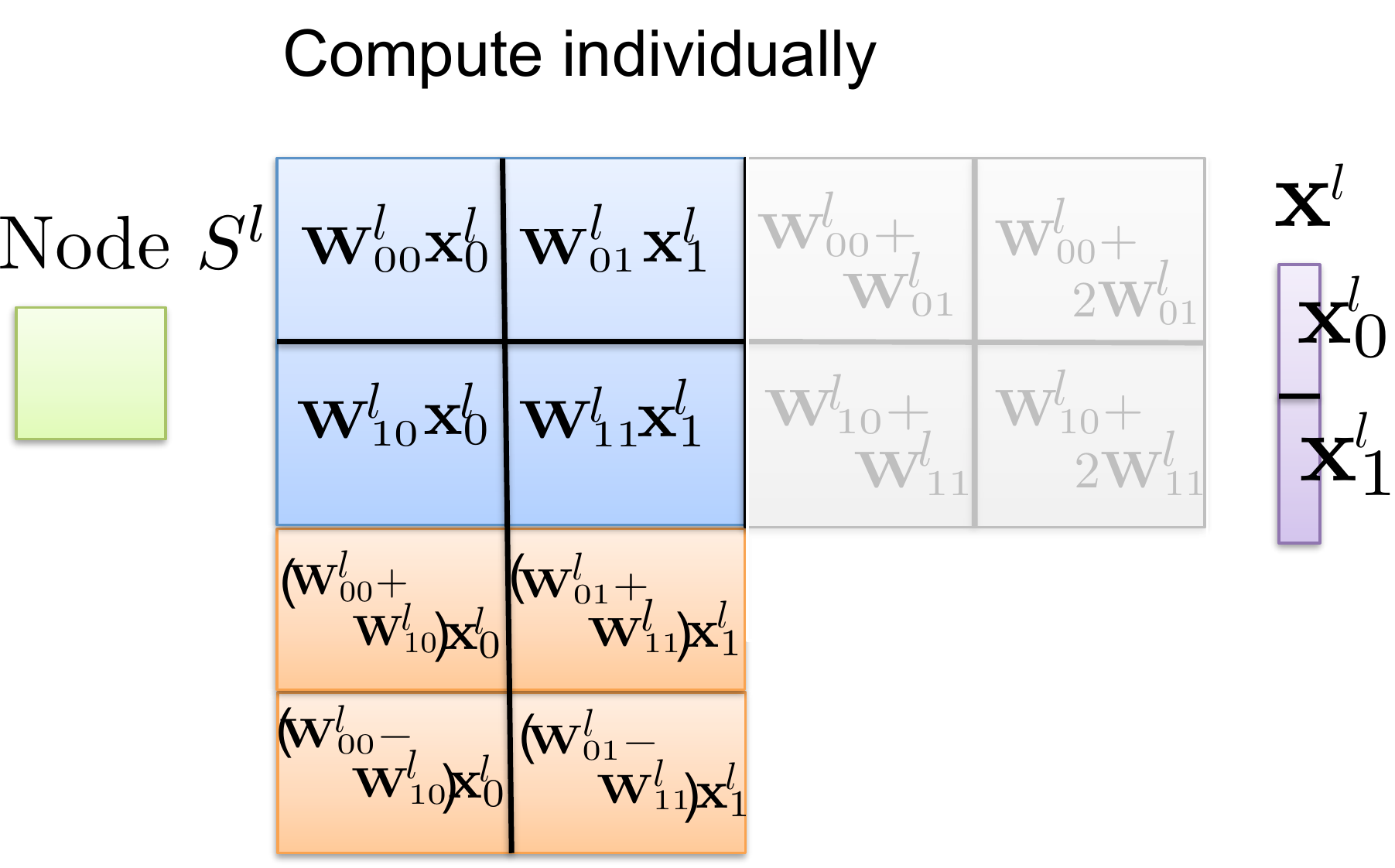}%
\label{fig:DNN_feedforward2}}\\
\subfloat[Step $3$: Nodes compute the sum of partial results horizontally and send to node $S^{l}$.]{\includegraphics[height=2.4cm]{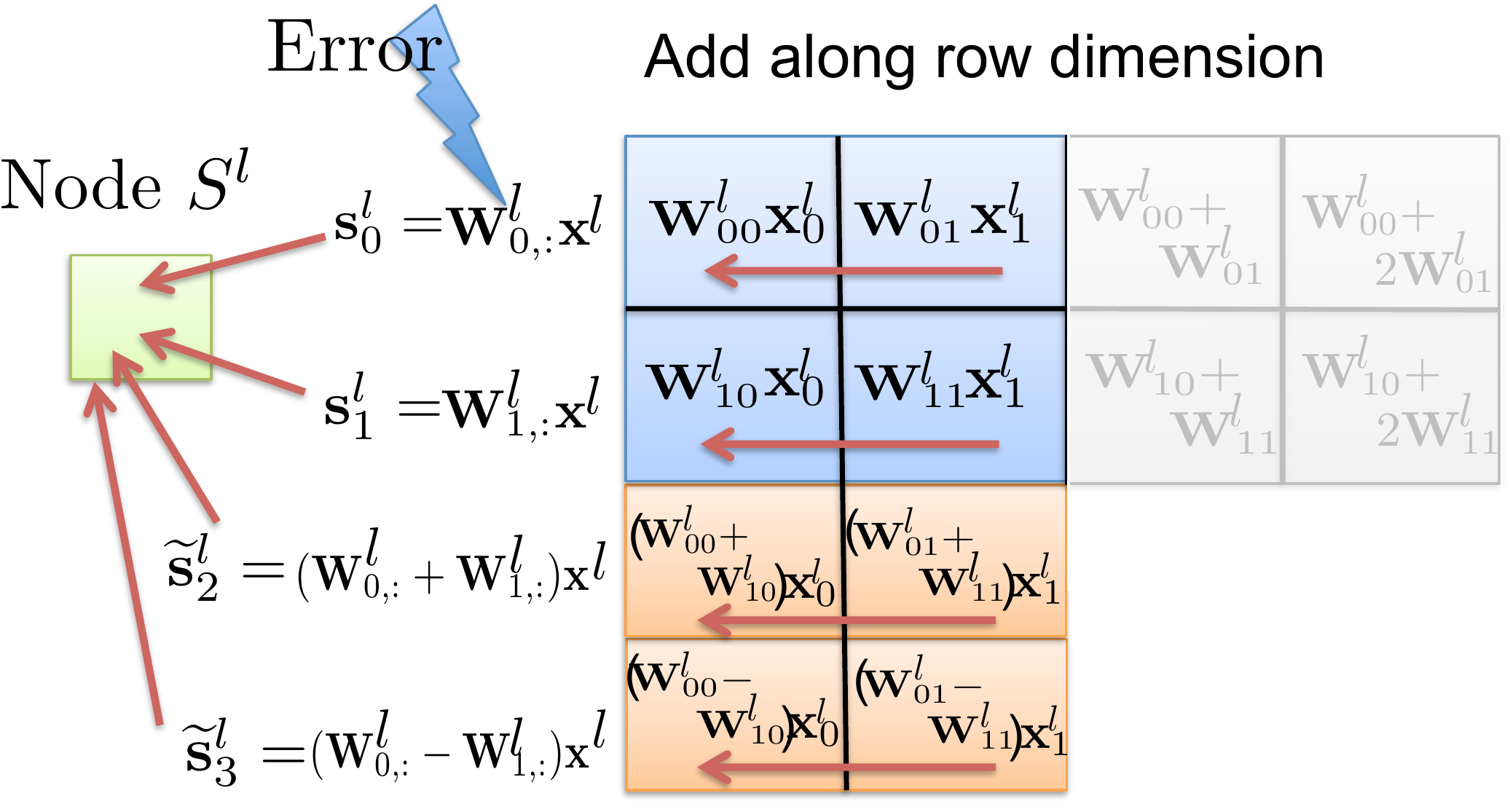}%
\label{fig:DNN_feedforward3}}
\hspace{0.5cm}
\subfloat[Step $4$: Additional Encoding Step (for update).]{\includegraphics[height=2.4cm]{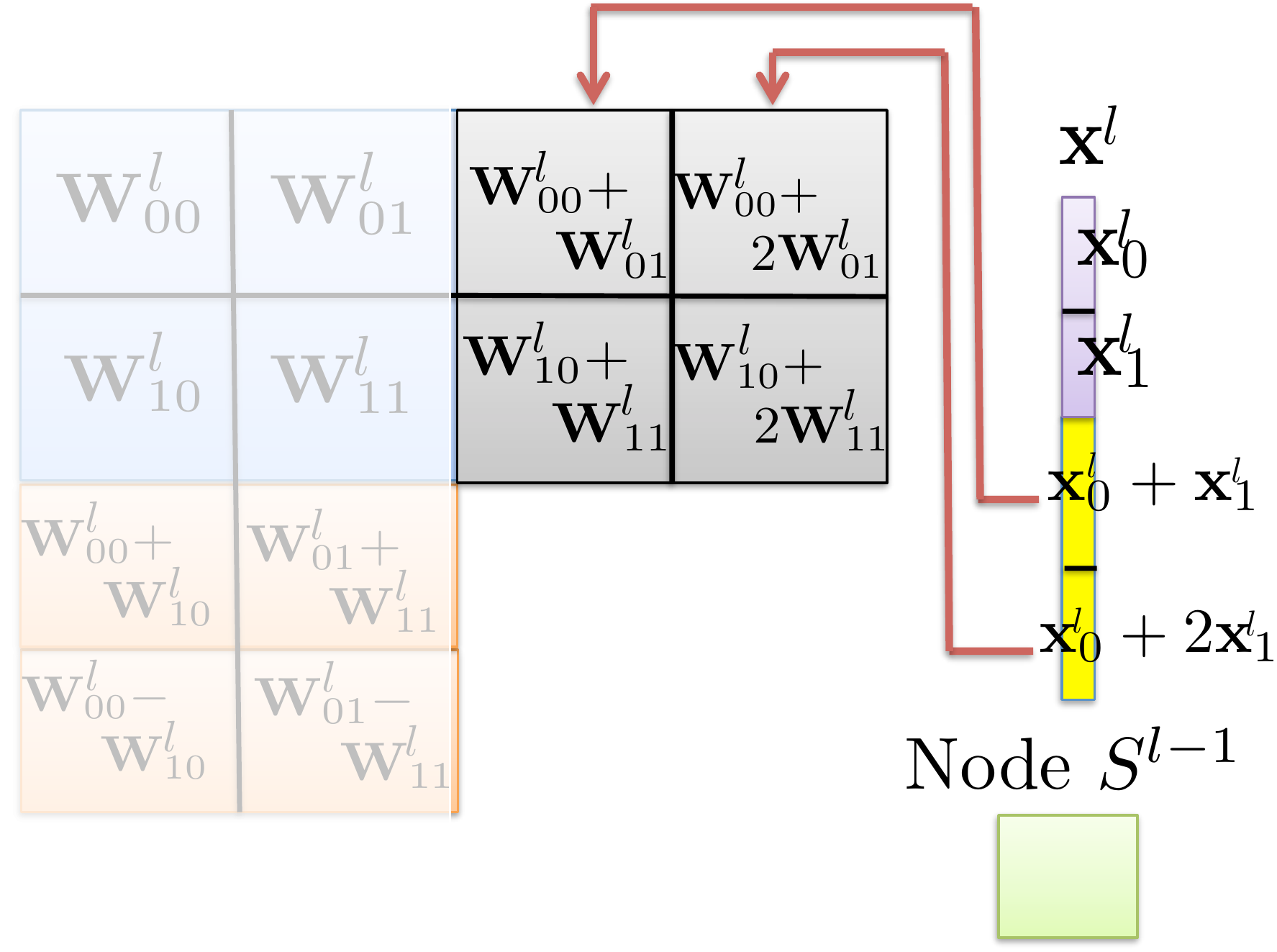}
\label{fig:DNN_feedforward4}}
\caption{Feedforward Stage in CodeNet strategy. }
\label{fig:DNN_feedforward}
\end{figure}

\noindent {\bf{[Step $1$]:}} The virtual node $S^{l-1}$ (which serves as encoder/source node for feedforward stage) divides the vector $\bm{x}^l$ into $ n (=2) $ parts and multi-casts appropriate portions of the vector, \textit{i.e.}, $\bm{x}^l_j$ to the $j$-th column of nodes in the grid, as shown in Fig.~\ref{fig:DNN_feedforward1}.\\
\noindent {\bf{[Step $2$]:}} All the active nodes then perform their individual computations, \textit{i.e.}, $\bm{W}^l_{i,j}\bm{x}^l_j$ or $\widetilde{\bm{W}}^l_{i,j}\bm{x}^l_j$, as shown in Fig.~\ref{fig:DNN_feedforward2}.\\ 
\noindent {\bf{[Step $3$]:}} For each of the $m+2$ ($=4$ here) rows in the grid, the partial results are added along the horizontal dimension and the sum is sent to the sink node $S^{l}$ for error detection through consistency checks (parity checks), as shown in Fig.~\ref{fig:DNN_feedforward3}. The sink node $S^{l}$ can correct upto $t$ errors, under Error Models $1$ or $2$, and thus decodes $\bm{s}^l=\bm{W}^l\bm{x}^l$. Then, it generates input for the next layer $(l+1)$ as follows: $\bm{x}^{(l+1)}=f(\bm{s}^l)$. If there are more errors under Error Model $2$, then the system reverts to the last checkpoint. \\
\noindent {\bf{[Step $4$]:}} Meanwhile, the node $S^{l-1}$ encodes the sub-vectors of $\bm{x}^l$ and sends coded sub-vectors to the nodes that were inactive during the feedforward stage multiplication as shown in Fig.~\ref{fig:DNN_feedforward4}. This additional encoding step does not affect the computation $\bm{s}^l=\bm{W}^l\bm{x}^l$ in the feedforward stage but will be useful in the update stage. The sub-vectors of $\bm{x}^l$ are encoded using the same $(n+2t,n)$ systematic MDS code used to encode the vertically-split blocks of $\bm{W}^l$. E.g.,~for $m=n=2$ and $t=1$, the encoding is, 
$$
\begin{bmatrix}
\bm{x}^l_0\\
\bm{x}^l_{1}\\
\widetilde{\bm{x}}^l_2 \\
\widetilde{\bm{x}}^l_{3}
\end{bmatrix}
= (\bm{G}_c^T \otimes \bm{I}_{N_{l-1}/2} )\begin{bmatrix} 
\bm{x}^l_0 \\
\bm{x}^l_{1}
\end{bmatrix} = \left( \begin{bmatrix}
1 & 0 \\
0 & 1\\
1 & 1\\
1 & 2
\end{bmatrix}\otimes \bm{I}_{N_{l-1}/2} \right)\begin{bmatrix} 
\bm{x}^l_0 \\
\bm{x}^l_{1}
\end{bmatrix}. $$
The $2$ parity sub-vectors (or coded sub-vectors) $\widetilde{\bm{x}}^l_2$ and
$\widetilde{\bm{x}}^l_3$ are sent to the entire column of nodes indexed $2$ and $3$ respectively, as shown in Fig.~\ref{fig:DNN_feedforward4}.


\subsubsection{Backpropagation stage on a single layer}

In the backpropagation stage, the key operation is the computation of the matrix-vector product $(\bm{c}^l)^T=(\bm{\delta}^l)^T\bm{W}^l$ (step O2), followed by a diagonal matrix post-multiplication. The computation of $(\bm{c}^l)^T=(\bm{\delta}^l)^T\bm{W}^l$ is performed in a way similar to the feedforward stage, as shown in Fig.~\ref{fig:DNN_backprop}.  For the matrix-vector product, we now use the nodes that contain
\begin{align*}
& \begin{bmatrix}
\bm{W}^l_{0,0} & \bm{W}^l_{0,1} & \widetilde{\bm{W}}^l_{0,2} & \widetilde{\bm{W}}^l_{0,3} \\
\bm{W}^l_{1,0} & \bm{W}^l_{1,1} & \widetilde{\bm{W}}^l_{1,2} & \widetilde{\bm{W}}^l_{1,3}  
\end{bmatrix} \\ &=
\begin{bmatrix}
\bm{W}^l_{0,0} & \bm{W}^l_{0,1} \\
\bm{W}^l_{1,0} & \bm{W}^l_{1,1} \\
\end{bmatrix}
\left(\begin{bmatrix}
1 & 0  & 1 &  1 \\
0 & 1 & 1 &  2  
\end{bmatrix}\otimes \bm{I}_{N_{l-1}/2}  \right) \\
&=  \bm{W}^l \left( \bm{G_c} \otimes \bm{I}_{N_{l-1}/2}\right).
\end{align*}
Similar to the feedforward stage, each of these $8$ sub-matrices are available at $8$ nodes before backpropagation starts. The node $S^l$ now serves as the encoder/source for the backpropagation stage as the direction of flow of computation is reversed from feedforward stage. Assume that node $S^l$ has the backpropagated error from the $(l+1)$-th layer, \textit{i.e.}, $(\bm{\delta}^l)^T$ at the beginning of backpropagation at layer $l$. We will justify this assumption towards the end of this subsection. Observe that,
\begin{align*}
&\begin{bmatrix}
(\bm{c}^l_0)^T & (\bm{c}^l_1)^T & (\widetilde{\bm{c}}^l_2)^T & (\widetilde{\bm{c}}^l_3)^T
\end{bmatrix}\\
&= \begin{bmatrix}
(\bm{\delta}^l_0)^T \bm{W}^l_{0,0} & (\bm{\delta}^l_0)^T\bm{W}^l_{0,1} & (\bm{\delta}^l_0)^T \widetilde{\bm{W}}^l_{0,2} & (\bm{\delta}^l_0)^T \widetilde{\bm{W}}^l_{0,3} \\
+(\bm{\delta}^l_1)^T\bm{W}^l_{1,0} & +(\bm{\delta}^l_1)^T\bm{W}^l_{1,1} & +(\bm{\delta}^l_1)^T\widetilde{\bm{W}}^l_{1,2} & +(\bm{\delta}^l_1)^T\widetilde{\bm{W}}^l_{1,3}  
\end{bmatrix}.  
\end{align*}

Each node only needs $(\bm{\delta}^l_0)^T$ or $(\bm{\delta}^l_1)^T$ to compute its small matrix-vector product. The partial results are then added along the vertical dimension and sent to decoder $S^{l-1}$. If any $1$ of $\{(\bm{c}^l_0)^T ,(\bm{c}^l_1)^T , (\widetilde{\bm{c}}^l_2)^T ,(\widetilde{\bm{c}}^l_3)^T \}$ is in error, the decoder $S^{l-1}$ can still decode $\bm{c}$ and proceed with the computation. The steps of the backpropagation stage are as follows (also illustrated in Fig.~\ref{fig:DNN_backprop}):\\
\noindent [{\bf{Step $1$}}]: First $S^l$ divides the row-vector $(\bm{\delta}^l)^T$ into $m(=2)$ equal parts and multi-casts $(\bm{\delta}^l_i)^T$ to the $i$-th rows of nodes as shown in Fig.~\ref{fig:DNN_backprop1}. \\
\noindent [{\bf{Step $2$}}]: Then, each active node performs its individual computations, \textit{i.e.}, $(\bm{\delta}^l_i)^T \bm{W}^l_{i,j}$ or $(\bm{\delta}^l_i)^T \widetilde{\bm{W}}^l_{i,j}$, as shown in Fig.~\ref{fig:DNN_backprop2}. \\
\noindent [{\bf{Step $3$}}]: After that, the partial results for each node are summed along the vertical dimension and sent to the virtual node $S^{l-1}$ for error detection through consistency checks (parity checks). The sink node can correct upto $t$ errors, as shown in Fig.~\ref{fig:DNN_backprop3}. Observe that, the virtual node $S^{l-1}$ now has $(\bm{\delta}^l)^T$ and $\bm{x}^l$, and thus it can generate the backpropagated error for the layer $(l-1)$ through the diagonal matrix post-multiplication step. If there are more than $t$ errors under Error Model $2$, the system detects the occurrence of errors with probability $1$, and reverts to the last \textit{checkpoint}.\\
\noindent [{\bf{Step $4$}}]: An additional step at $S^{l}$ is to encode and send coded sub-vectors of $(\bm{\delta}^l)^T$ to the inactive nodes, that were not used for the matrix-vector product in the backpropagation stage, as shown in Fig.~\ref{fig:DNN_backprop4}). The encoding uses the same $(m+2t,m)$ systematic MDS code used to encode the horizontally-split blocks of $\bm{W}^l$. For $m=n=2$ and $t=1$, the encoding is given by,
\begin{align*}
\begin{bmatrix}
(\bm{\delta}_0^l)^T | \  (\bm{\delta}_1^l)^T | \ (\widetilde{\bm{\delta}}_2^l)^T | \ (\widetilde{\bm{\delta}}_3^l)^T 
\end{bmatrix}
&= \begin{bmatrix}
(\bm{\delta}_0^l)^T | \ (\bm{\delta}_1^l)^T
\end{bmatrix}
\left( \bm{G}_r\otimes \bm{I}_{N_{l}/2} \right)\\
& =\begin{bmatrix}
(\bm{\delta}_0^l)^T | \  (\bm{\delta}_1^l)^T
\end{bmatrix}
\left( \begin{bmatrix}
1 & 0 & 1 & 1 \\
0 & 1 & 1 & - 1 
\end{bmatrix} \otimes \bm{I}_{N_{l}/2} \right).
\end{align*}
The parity sub-vectors (or coded sub-vectors) $(\widetilde{\bm{\delta}}_2^l)^T$ and $(\widetilde{\bm{\delta}}_3^l)^T$ are sent to the row of nodes indexed $2$ and $3$ respectively.

\begin{figure}[t]
\centering
\subfloat[Step $1$: Node $S^{l}$ multi-casts appropriate portions of $(\bm{\delta}^l)^T$ to corresponding row of nodes.]{\includegraphics[height=2.3cm]{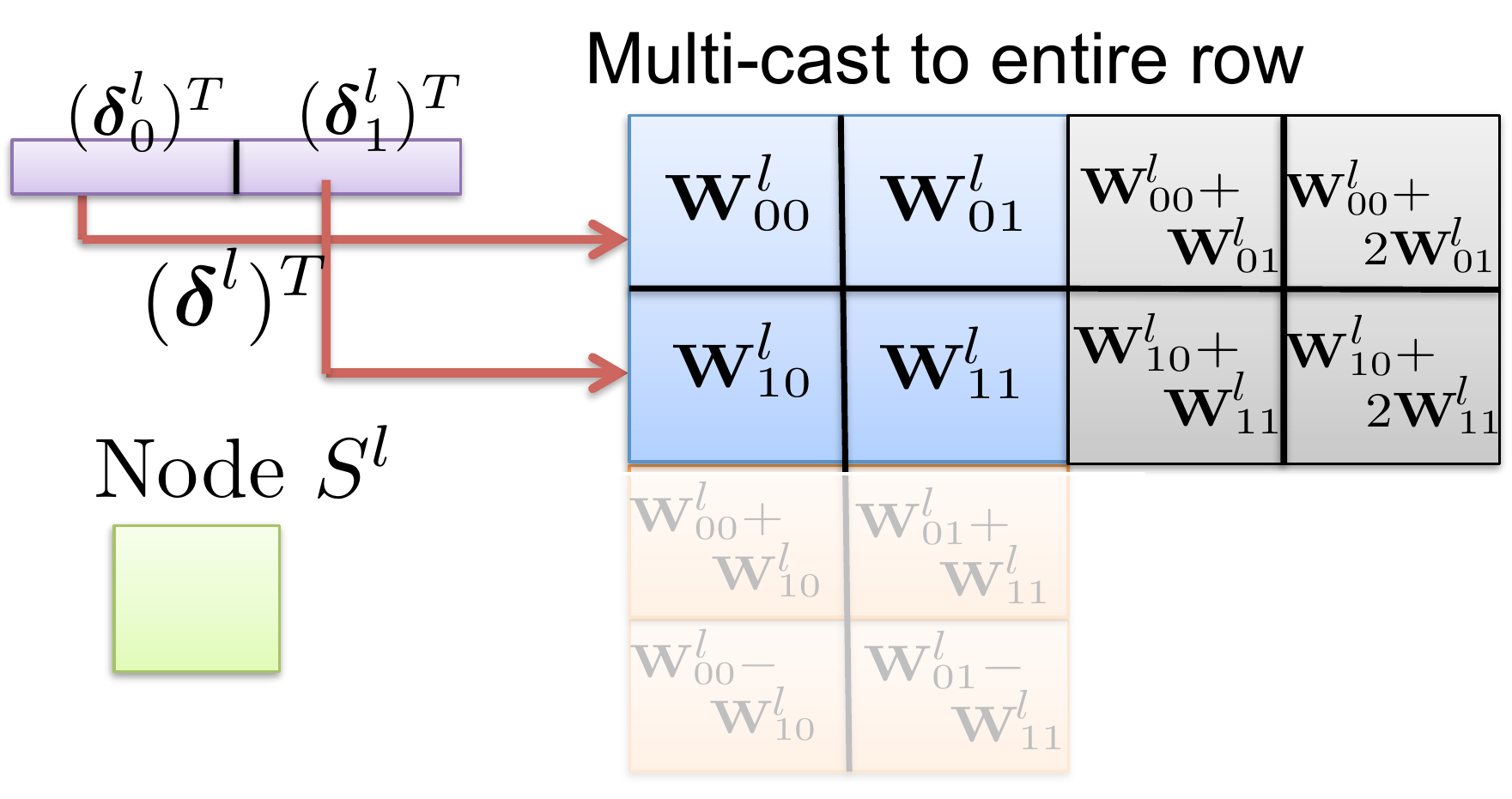}%
\label{fig:DNN_backprop1}}
\hspace{0.5cm}
\subfloat[Step $2$: Each node performs individual computations, \textit{i.e.}, $\bm{\delta}^l_i \bm{W}^l_{i,j} $ or $\bm{\delta}^l_i \widetilde{\bm{W}}^l_{i,j} $.]{\includegraphics[height=2.3cm]{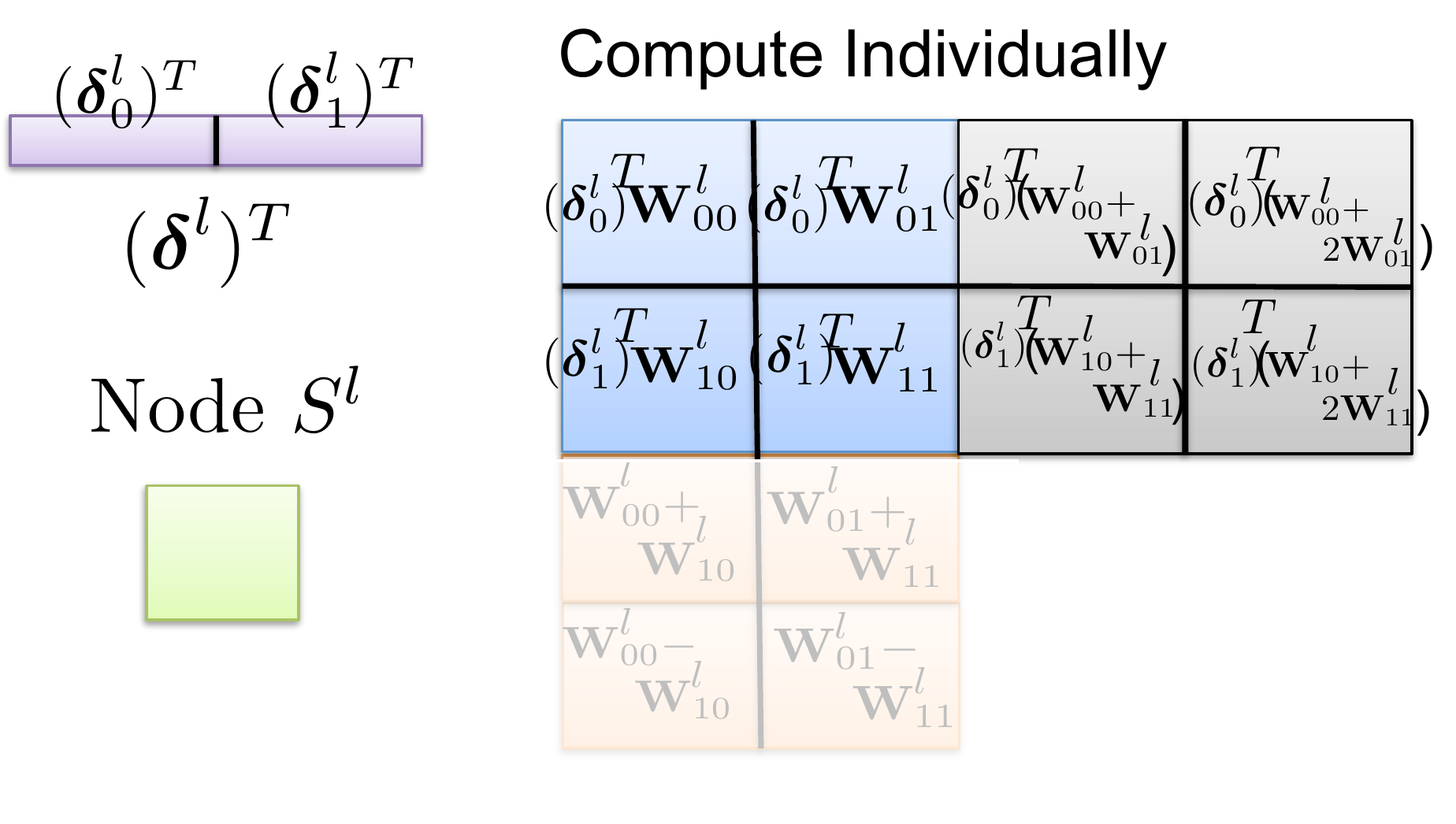}%
\label{fig:DNN_backprop2}} \\
\subfloat[Step $3$: Nodes compute sum of partial results vertically and send to $S^{l-1}$.]{\includegraphics[height=2.6cm, width=4.1cm]{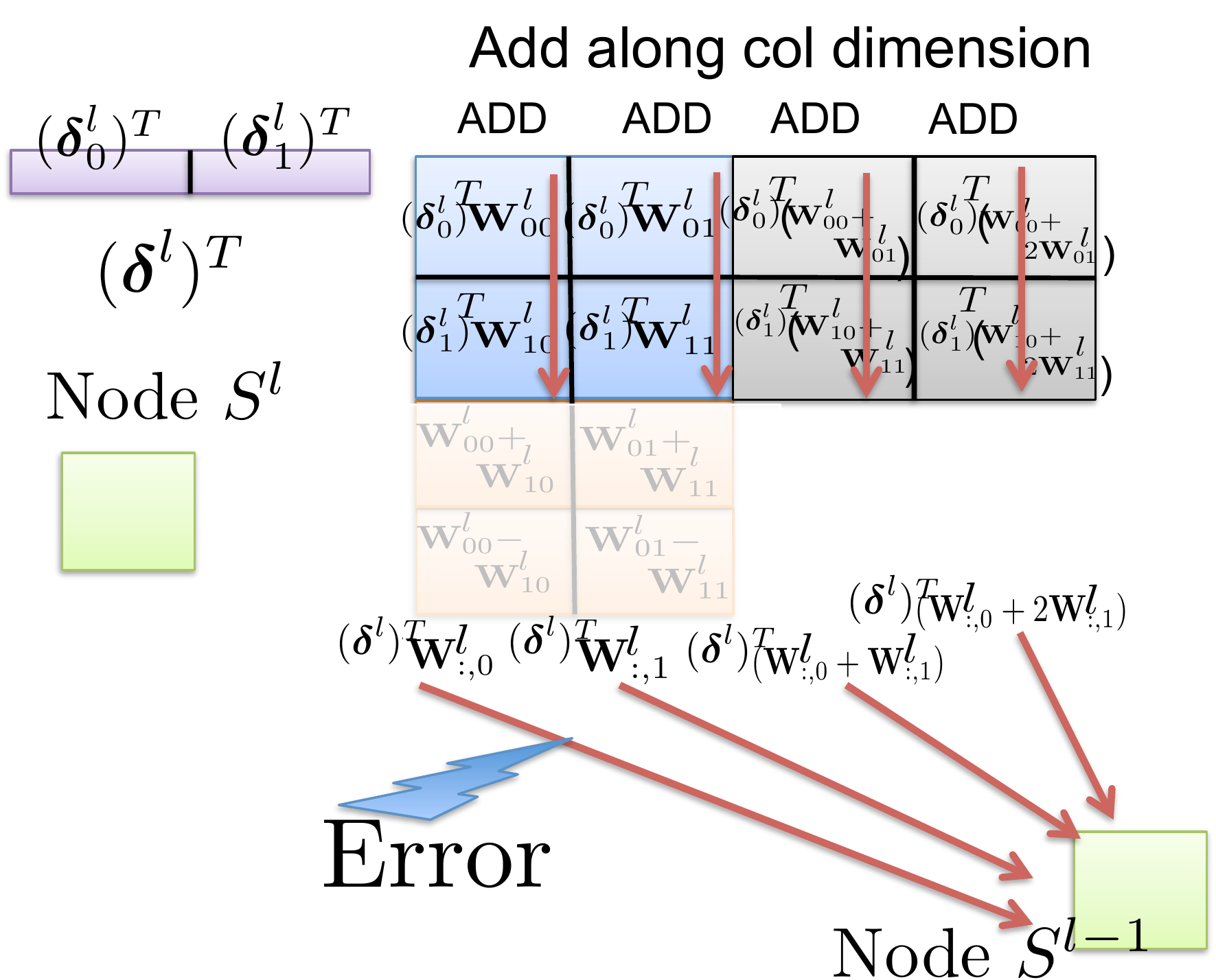}%
\label{fig:DNN_backprop3}}
\hspace{0.5cm}
\subfloat[Step $4$: Additional Encoding Step (for update).]{\includegraphics[height=2.4cm]{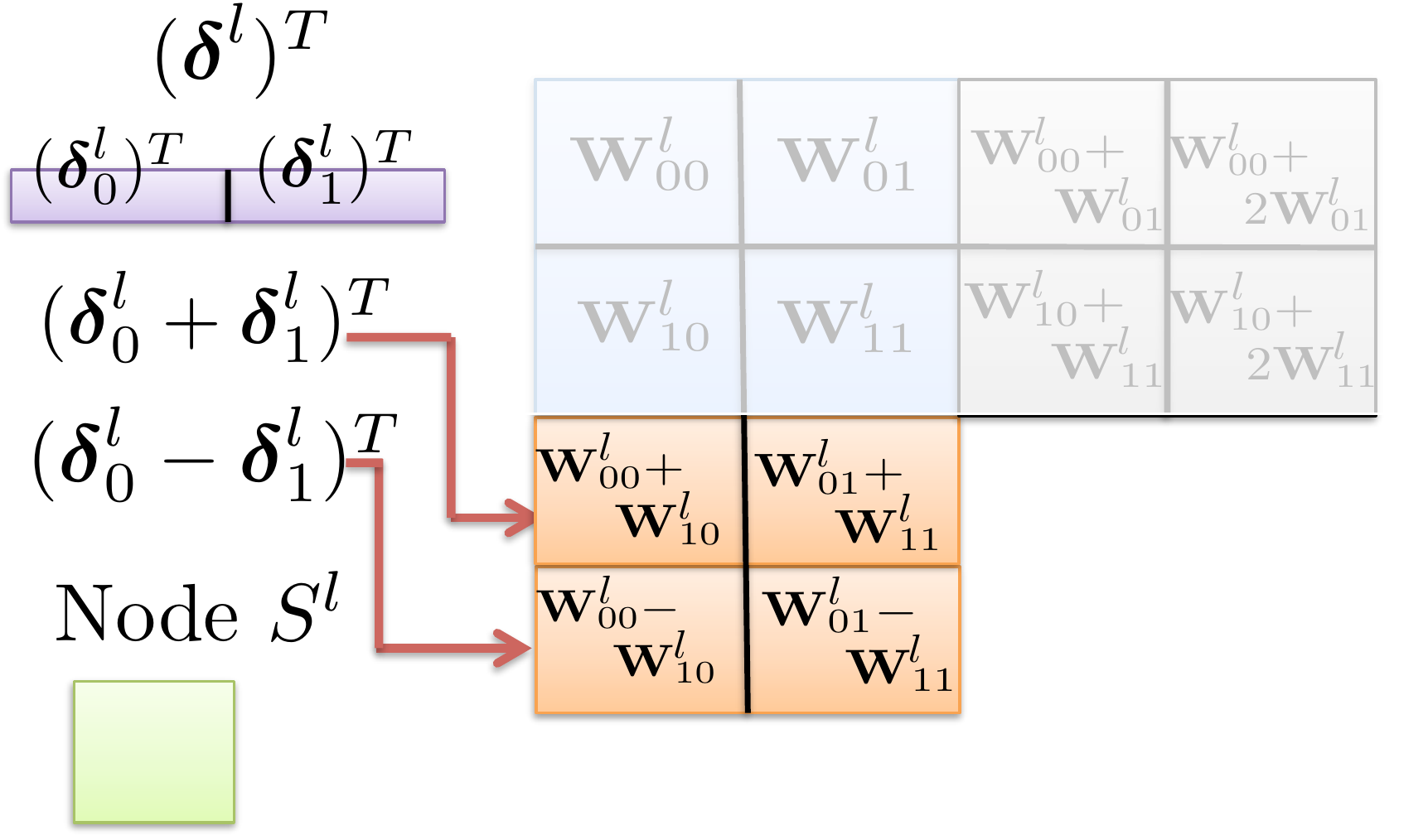}%
\label{fig:DNN_backprop4}}
\caption{Backpropagation Stage in CodeNet strategy}
\label{fig:DNN_backprop}
\end{figure}
\subsubsection{Update stage on a single layer}
After these steps, each node is able to update its own sub-matrix without additional communication, and all overheads associated are negligible compared to the complexity of the update at each node. Recall that the update rule (O3) is given by $\bm{W}^l + \eta \bm{\delta}^l(\bm{x}^l)^T$. Observe that, 
$$
\bm{W}^l + \eta \begin{bmatrix}
\bm{\delta}^l_0 \\ \bm{\delta}^l_1
\end{bmatrix} \begin{bmatrix}
(\bm{x}_0^l)^T (\bm{x}_1^l)^T
\end{bmatrix}=
  \begin{bmatrix}
\bm{W}^l_{0,0} + \eta \bm{\delta}_0^l(\bm{x}_0^l)^T & \bm{W}^l_{0,1} + \eta \bm{\delta}_0^l(\bm{x}_1^l)^T \\
  \bm{W}^l_{1,0} + \eta \bm{\delta}_1^l(\bm{x}_0^l)^T & \bm{W}^l_{1,1}  + \eta \bm{\delta}_1^l(\bm{x}_1^l)^T
\end{bmatrix}. 
$$
Therefore, any sub-matrix of $\bm{W}^l$, e.g.~$\bm{W}^l_{i,j}$, only requires the sub-vectors $\bm{\delta}^l_i$ and $\bm{x}^l_j$ to update itself (crucial observation). 

\begin{clm}
CodeNet ensures that every node has both the appropriate parts of $\bm{\delta}^l$ and $\bm{x}^l$ to update itself, including the coded (or parity) nodes.
\end{clm}

As an example, refer to Fig.~\ref{fig:DNN_backprop5}. The node with $\bm{W}^l_{0,0}$, also receives $\bm{x}^l_0$ in feedforward stage and $\bm{\delta}^l_0$ in the backpropagation stage. Thus, it can update itself as $ \bm{W}^l_{0,0} + \eta \bm{\delta}^l_0(\bm{x}_0^l)^T $ without requiring any further communication. Interestingly, our strategy also ensures this for the coded sub-matrices due to the additional encoding steps at the end of feedforward and backpropagation stages. Consider the node at location $(0,3)$ containing $\widetilde{\bm{W}}^l_{0,3}=\bm{W}^l_{0,0}+ 2\bm{W}^l_{0,1}$. This node gets $\widetilde{\bm{x}}^l_3=\bm{x}^l_0 + 2\bm{x}^l_1$ in the additional encoding step in the feedforward stage and $\bm{\delta}^l_0$ in the backpropagation stage. Thus it can update itself as,
\begin{align}
&\widetilde{\bm{W}}^l_{0,3} + \bm{\delta}^l_0 (\widetilde{\bm{x}}^l_3)^T =  \bm{W}^l_{0,0}+ 2\bm{W}^l_{0,1} + \eta \bm{\delta}^l_0 (\bm{x}^l_0 + 2\bm{x}^l_1)^T \nonumber \\
& = \underbrace{\bm{W}^l_{0,0}+ \eta \bm{\delta}^l_0 (\bm{x}_0^l)^T}_{\text{update of } \bm{W}^l_{0,0}} + 2( \underbrace{\bm{W}^l_{0,1} + \eta \bm{\delta}^l_0 (\bm{x}_1^l)^T}_{\text{update of } \bm{W}^l_{0,1}}).
\end{align}

More generally, every node can update itself using one of the three rules (for the example, $m=n=2$ and $t=1$):
\begin{align}
&\bm{W}^l_{i,j} \leftarrow \bm{W}^l_{i,j} + \bm{\delta}^l_i (\bm{x}_j^l)^T \ \forall \  0 \leq i \leq m-1  \text{ and } 0 \leq j \leq n-1 \\
& \widetilde{\bm{W}}^l_{i,j} \leftarrow \widetilde{\bm{W}}^l_{i,j} + \bm{\delta}_i (\widetilde{\bm{x}}^l_j)^T \ \forall \  0 \leq i \leq m-1  \text{ and }   n \leq j \leq n+2t-1 \\
& \widetilde{\bm{W}}^l_{i,j} \leftarrow \widetilde{\bm{W}}^l_{i,j} + \widetilde{\bm{\delta}}^l_i (\bm{x}_j^l)^T \ \forall \  m \leq i \leq m+2t-1  \text{ and }   0 \leq j \leq n-1
\end{align}
{Any errors that occur during the update stage, \textit{i.e.}, step O3 at layer $l$, corrupt the sub-matrices $\bm{W}^l_{i,j}$ (or $\widetilde{\bm{W}}^l_{i,j}$) and are thus propagated into the next iteration of layer $l$\footnote{Note that this error in the update stage does not affect the computations in the layer $l+1$, \textit{i.e.}, the next \textit{layer} (for the same iteration). It only affects the following \textit{iteration}.}. As an example, suppose error occurs at node $(i,j)$ at layer $l$. Then the updated sub-matrix  $\bm{W}^l_{i,j}$ (or $\widetilde{\bm{W}}^l_{i,j}$) is now erroneous. For $0\leq j\leq n-1$, the erroneous sub-matrix is used to compute an output first after step O1 at layer $l$ in the next iteration, while for $n\leq j\leq n+2t-1$, the output is produced first after step O2 at layer $l$ in the next iteration. Thus, the errors can be detected and corrected (if within CodeNet's error tolerance) after either step O1 or step O2 of layer $l$ in the next iteration, whenever the erroneous node produces an output first for that layer. }

\begin{figure}[ht]
\centering
{\includegraphics[height=2.4cm, width=8cm]{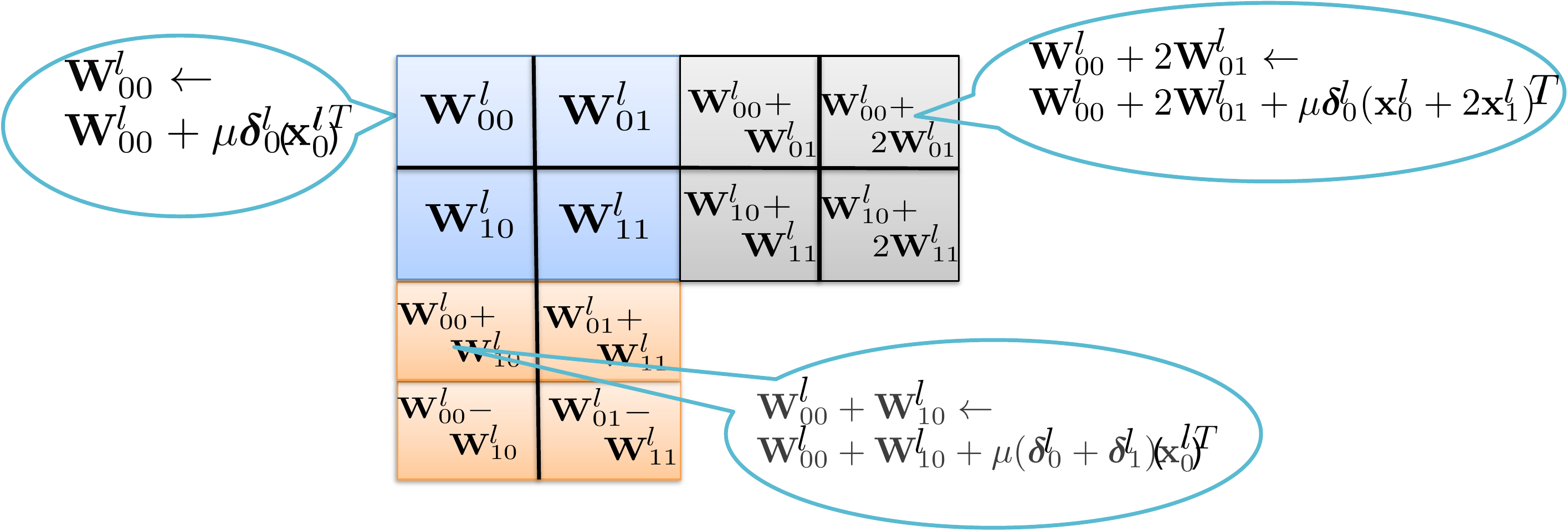}}
\caption{Update Stage: Each node is able to update itself.}
\label{fig:DNN_backprop5}
\end{figure}

\subsection{Decentralized Implementation of CodeNet by Removing Virtual Nodes}
While we described our strategy assuming presence of error-free virtual nodes $S^0, S^1, \ldots, S^L$, in an actual implementation such nodes would become a single point of failure. Thus, in our actual implementation, all the operations of these virtual nodes including encoding/error-detection/decoding/nonlinear activation/diagonal matrix post-multiplication are performed in a decentralized manner (also allowing for errors in these steps under Error Model $2$). To address these errors, intuitively, we replicate the functionality of each of the virtual nodes $S^0, S^1, \ldots, S^L$ at multiple distributed nodes.  

We ensure, using efficient collective communication protocols~\cite{chan2007collective}, that our implementation also has same scaling of additional communication costs as an uncoded or replication strategy under equal storage per node (for fixed $t$; see Theorem~\ref{thm:complexity}). This is desirable because communication is often much more expensive in computational systems~\cite{chan2007collective} than computing. 

Before proceeding further, we introduce two standard ``collective'' communication protocols (see~\cite{chan2007collective} for details) that can be implemented with low communication cost. In a cluster of $P$ nodes, let the $p$-th node initially have only vector $\bm{a}_p$. 

\noindent \textbf{Reduce and All-Reduce:} After Reduce (or All-Reduce), one node (or all the $P$ nodes) gets the sum $\sum_{p=1}^{P}\bm{a}_p$.

\noindent \textbf{Gather and All-Gather:} After Gather (or All-Gather), one node (or all the $P$ nodes) gets all the vectors $\{\bm{a}_p |\ p=1,2,\ldots,P\}$.

The communication protocols will be used to make various operations in the aforementioned algorithm decentralized, that includes aggregating sum of partial results, performing $2t$ consistency checks for error detection, additional encoding etc. Now, we describe our decentralized implementation in detail (see \Cref{appendix:algorithm} for formal description). 

The pre-processing step is performed only once before training as described earlier. Also assume that before every iteration, the updated sub-matrix $\bm{W}^l_{i,j}$ or $\widetilde{\bm{W}}^l_{i,j}$ is available at the respective nodes, because every node will be able to update its own sub-matrix in the update stage.

\subsubsection{Feedforward stage on a single layer} The nodes $(i,j)$ with indices $0\leq i\leq m+2t-1$ and $0 \leq j \leq n-1 $ are the active nodes (not faded in Fig.~\ref{fig:DNN_feedforward1}). As there is no virtual node $S^{l-1}$ to multi-cast portions of $\bm{x}^l$, let us skip Step $1$ for now, and assume that, in every iteration at layer $l$, every active node has $\bm{x}^l_j$ (or an erroneous $\bm{x}^l_j$) at the beginning of computation in that layer. We will justify this assumption shortly. Now, for Step $2$, each active node computes its individual matrix-vector product (which may be erroneous), \textit{i.e.} $\bm{W}^l_{i,j}\bm{x}^l_j$ if $0 \leq i \leq m-1 $or $\widetilde{\bm{W}}^l_{i,j}\bm{x}^l_j$ if $m \leq i \leq m+2t-1$. However, Step $3$ has to be modified as there is no virtual node $S^{l}$. Thus, in the decentralized implementation, \textit{all the active nodes} in each row compute the sum $\bm{s}^l_i= \sum_{j=0}^{n-1} \bm{W}^l_{i,j}\bm{x}^l_j$ or $\widetilde{\bm{s}}^l_i= \sum_{j=0}^{n-1} \widetilde{\bm{W}}^l_{i,j}\bm{x}^l_j$ using an All-Reduce operation, so that \textbf{every active node in row $i$ has $\bm{s}^l_i$ (or $\widetilde{\bm{s}}^l_i$)}.

\paragraph{Error-Detection} Now, instead of communicating the outputs $\bm{s}^l_i$ (or $\widetilde{\bm{s}}^l_i$) to a centralized node $S^{l}$ for consistency checks, each of the active nodes performs those $2t$ consistency checks (or parity checks) \textbf{individually}. However, each of the $2t$ parity check-sums are computed distributedly and shared among the nodes using efficient collective communication protocols, as opposed to sending all the outputs $\bm{s}^l_i$ (or $\widetilde{\bm{s}}^l_i$) to every node. For each column, note that node $(i,j)$ already has $\bm{s}^l_i$ (or $\widetilde{\bm{s}}^l_i$). Thus, each of the $2t$ consistency checks, which essentially consist of a different linear combination of all the $\bm{s}^l_i$'s and $\widetilde{\bm{s}}^l_i$'s, is performed by an All-Reduce communication operation for each column. Thus, we require $2t$ All-Reduce operations for each column in parallel to detect errors.

\paragraph{Errors in Error-Detection under Error Model $2$} To check for errors during this process (only for Error Model $2$), we include an additional verification step where all the active nodes exchange their evaluation of the consistency checks, \textit{i.e.}, a vector of length $2t$ among each other and compare them (additional communication overhead of $\alpha \log{\hat{P}}+2\beta \hat{P}t$). Because the complexity of this verification step is low in scaling sense (does not depend on $N_l,N_{l-1}$ and $\hat{P} \ll N_l,N_{l-1}$), we assume that it is error-free. The probability of errors occurring within such a small duration is negligible as compared to other computations of longer durations. If there is any disagreement among the nodes, the entire DNN (\textit{i.e.}, the weight matrices for every layer) is restored from the last checkpoint. Otherwise, if there is an agreement among all nodes, the algorithm moves forward.

\paragraph{Additional Encoding Step (If no errors)} For the inactive set of nodes of the feedforward stage, the additional encoding step (Step $4$) is performed now, once for each column. The encoding of $\widetilde{\bm{x}}^l_j$ at inactive node $(i,j)$ is performed by linearly combining all the stored $\bm{x}^l_j$s at the active nodes in that row through an All-Reduce operation, in parallel for all rows. 

\paragraph{Generate input for $(l+1)$-th layer (If no errors)}
Each active node in column $j$ needs to generate $\bm{x}^l_j$ for the $(l+1)$-th layer. For this, first each active node in rows $0$ to $m-1$ computes the element-wise nonlinear activation $f(\bm{s}^l_i)=\bm{x}^{(l+1)}_i$ where $\bm{x}^{(l+1)}_i$ is the $i$-th sub-vector obtained when $\bm{x}^{(l+1)}$ is divided into $m$ equal parts (instead of $n$), for $i=0,1,\ldots,m-1$. Each active node in column $j$ then obtains the appropriate elements of $\bm{x}^{(l+1)}_j$ from the nodes that have it in that column through one or more broadcasts. E.g., if $m=n=\sqrt{P}$, only the diagonal nodes can broadcast $\bm{x}^{(l+1)}_j$ to all nodes in that column. Thus, every active node has $\bm{x}^{(l+1)}_j$ at the beginning of the feedforward stage at $(l+1)$-th layer, which also justifies our initial assumption.

\paragraph{Decoding (Only if errors)} When errors are detected, each active column starts an All-Gather in parallel, to get all the $\bm{s}^l_i$'s and $\widetilde{\bm{s}}^l_i$'s from all the rows and attempts to decode $\bm{s}^l$.

\paragraph{Errors in decoding under Error Model $2$} Decoding errors can occur under Error Model $2$. So, we include one more verification step where all nodes exchange their assessment of
node outputs, \textit{i.e.}, a list of nodes that they found erroneous and
compare these lists (additional overhead of $\Theta(\alpha(\log{\hat{P}})+\beta \hat{P}^2)$). If there is a disagreement at one or more nodes during
this process, the decoding is deemed erroneous, and the entire DNN is restored from the last checkpoint. Once again, as the complexity of this step is low, we assume the verification step is error-free.

\paragraph{Regeneration (If no decoding errors)}
If there are no decoding errors and the number of errors is within the tolerance $t$ under Error Models $1$ and $2$, every node is able to detect which nodes were erroneous. Then, all the stored sub-matrices ($\bm{W}^l_{i,j}$) or coded sub-matrices ($\widetilde{\bm{W}}^l_{i,j}$) and the vectors $\bm{x}^l_j$ and $\bm{\delta}^l_i$ at the erroneous nodes are first deleted assuming they are erroneous, and then \textbf{generated} again by accessing some of the nodes that are known to be correct before the algorithm proceeds forward. The process of deleting the stored content and generating it again by accessing other nodes is called regeneration. 

E.g., suppose that during the feedforward stage (after O1), every node found out that one of the outputs, $\bm{s}^l_1$ is erroneous. Then, it is possible that one or both of the nodes located at index $(1,0)$ or $(1,1)$ might be corrupt. Therefore, both $\bm{W}^l_{1,0}$ and $\bm{W}^l_{1,1}$, as well as the stored vectors $\bm{x}^l_0$ or $\bm{x}^l_1$ in these nodes are required to be regenerated. Because the blocks (or sub-matrices) of $\bm{W}^l$ are encoded using $(4,2)$ MDS codes, any $2$ blocks in the same row or column in the grid might be accessed. As an example, $\bm{W}^l_{0,0}$ and $\widetilde{\bm{W}}^l_{2,0}=\bm{W}^l_{0,0}+\bm{W}^l_{1,0}$ in the same column may be accessed to regenerate $\bm{W}^l_{1,0}$. The stored vectors $\bm{x}^l_j$ can be regenerated by accessing any correct node in column $j$. If errors are found during the backpropagation stage (after O2), both $\bm{x}^l_j$ and $\bm{\delta}^l_i$ are regenerated.

Both decoding and regeneration are expensive because sub-vectors and sub-matrices need to be communicated across nodes resulting in communication complexities of $\Theta\big(\frac{N_lN_{l-1}}{P}\big)$, but still the computation proceeds forward to the next iteration as compared to a strategy with no error correction (but that can still detect errors, e.g. replication), where the system would resume from its last checkpoint even with a single error. We also show theoretically in Theorem~\ref{thm:time} (see next section) how error correction with regeneration can provide scaling sense advantages in expected computation time over a comparable replication strategy with no error correction. Note that any error occurring during regeneration of a sub-matrix can be detected and corrected in the next iteration, when these erroneous nodes produce an output for the first time for that layer using the erroneous sub-matrix.
 
After successful \textbf{decoding} and \textbf{regeneration}, the inactive nodes perform the additional encoding step while the active nodes generate $\bm{x}^{(l+1)}_j$ for the next layer (no communication is needed as during decoding they already produced the entire $\bm{s}^l$). Under Error Model $2$, if the number of errors exceeds the error tolerance $t$ (in this case $t=1$), then the nodes can still detect the occurrence of error with probability $1$, even though they cannot locate or correct it. So, the entire DNN is again restored from the last checkpoint.

\subsubsection{Backpropagation and Update Stages on a single layer} For the backpropagation stage at layer $l$, the nodes with indices $(i,j)$ for $0 \leq i \leq m-1$ and $0 \leq j \leq n+2t-1$ are the active nodes. We again skip Step $1$ and assume that every node has $\bm{\delta}_i$ for that layer, similar to the feedforward stage (see \Cref{appendix:algorithm}). The steps $2$ and $3$, \textit{i.e.}, the computation of $(\bm{c}^l_j)^T$ or $(\widetilde{\bm{c}}^l_j)^T $ are also carried out similar to the feedforward stage. When there are no errors, every active node computes $\bm{\delta}^{(l-1)}_j$ by performing the diagonal matrix post-multiplication $(\bm{\delta}^{(l-1)}_j)^T=(\bm{c}^l_j)^T\bm{D}^l_j$ where $\bm{\delta}^{(l-1)}_j$ is a sub-matrix of $\bm{\delta}^{(l-1)}$ when it is divided into $n$ equal parts (instead of $m$) and $\bm{D}^l_j$ is a diagonal matrix whose entries only depend on $\bm{x}^l_j$ which is already available at column $j$. After this, every active node fetches the appropriate parts of $\bm{\delta}^{(l-1)}_i$ from the nodes in that row that have it, through one or more broadcasts. Thus, every active node has $\bm{\delta}^{l-1}_i$ at the beginning of backpropagation in the $(l-1)$-th layer, as assumed. Decoding and regeneration are also performed in a manner similar to the feedforward stage.

Since each node has the vectors $\bm{x}^l_j$ (or $\widetilde{\bm{x}}^l_j$) and $\bm{\delta}^l_i$ (or $\widetilde{\bm{\delta}}^l_i$) for a layer, it can also update itself by computing an outer product. As mentioned before, the errors in the update stage (step O3) will show up as noise in the output of step O1 or O2 at layer $l$, in the next iteration, when the erroneous updated sub-matrix is used next.

\paragraph*{Errors in additional encoding/nonlinear activation/diagonal matrix post-multiplication} Under Error Model $2$, errors can also occur either at the inactive nodes during the additional encoding steps, or at the active nodes during the generation of $\bm{x}^{(l+1)}_j$ (after the nonlinear activation) or generation of $\bm{\delta}^{(l-1)}_i$ (after the diagonal matrix post-multiplication step). If the error is during additional encoding, e.g.~in encoded sub-vector $\widetilde{\bm{x}}^l_j$, then it will corrupt the update of $\widetilde{\bm{W}}^l_{i,j}$. This error will be detected in the next iteration, when $\widetilde{\bm{W}}^l_{i,j}$ is used in computation, because the error will show up as an additive noise in the output. Alternately, if the error is in sub-vector $\bm{x}^{(l+1)}_j$, it propagates into the feedforward sum $\bm{s}^{(l+1)}_i=\sum_{j=0}^{n-1}\bm{W}^{(l+1)}_{i,j}\bm{x}^{(l+1)}_j $ or $\widetilde{\bm{s}}^{(l+1)}_i=\sum_{j=0}^{n-1}\widetilde{\bm{W}}^{(l+1)}_{i,j}\bm{x}^{(l+1)}_j $, and shows up as an additive noise in $\bm{s}^{(l+1)}_i$ or $\widetilde{\bm{s}}^{(l+1)}_i$ respectively in the feedforward stage of layer $(l+1)$. Errors in $\bm{\delta}^{(l-1)}_i$ are also detected in the backpropagation stage of layer $(l-1)$ similarly.

\section{Discussion and Conclusions}
While the decentralized nature of neural network training proposed here is essential for biological plausibility, it is not sufficient. For instance, one would want to code Hebbian learning~\cite{hebbianlearning}, instead of backpropagation, because it is thought to be a closer approximation of how biological neural networks learn. This is left for future work. We therefore believe that this paper is only a step in attaining biologically plausible mechanisms for reliable neural network training. 
Nevertheless, coded computation in biological neural networks is largely ignored in the discussion on efficient coding hypothesis~\cite{barlow1961possible}, even though the hypothesis is largely driven by Shannon-theoretic principles. We believe that understanding how the brain attains reliability using error-prone computation -- the motivation of von Neumann's work~\cite{von1956probabilistic} -- is a key step towards the broader goal of understanding how neural circuits work, but is an aspect that has received little attention from the theory community.

Finally, we note that biologically-plausible algorithms are not simply of interest from a neuroscience and HPC perspective; they are also desirable when chips for artificial neural networks are fabricated, e.g.~in neuromorphic computing. Thus, the results here could be of interest to neuromorphic computing community as well.

\noindent \textbf{Appendix} is provided after References.

%



\bibliography{arxiv}

\begin{thebibliography}{10}
\providecommand{\url}[1]{#1}
\csname url@samestyle\endcsname
\providecommand{\newblock}{\relax}
\providecommand{\bibinfo}[2]{#2}
\providecommand{\BIBentrySTDinterwordspacing}{\spaceskip=0pt\relax}
\providecommand{\BIBentryALTinterwordstretchfactor}{4}
\providecommand{\BIBentryALTinterwordspacing}{\spaceskip=\fontdimen2\font plus
\BIBentryALTinterwordstretchfactor\fontdimen3\font minus
  \fontdimen4\font\relax}
\providecommand{\BIBforeignlanguage}[2]{{%
\expandafter\ifx\csname l@#1\endcsname\relax
\typeout{** WARNING: IEEEtran.bst: No hyphenation pattern has been}%
\typeout{** loaded for the language `#1'. Using the pattern for}%
\typeout{** the default language instead.}%
\else
\language=\csname l@#1\endcsname
\fi
#2}}
\providecommand{\BIBdecl}{\relax}
\BIBdecl

\bibitem{shannon1948}
C.~Shannon, ``A mathematical theory of communication,'' \emph{Bell Syst. Tech.
  J.}, vol.~27, no.~4, pp. 623--656, Oct 1948.

\bibitem{von1956probabilistic}
J.~von Neumann, ``Probabilistic logics and the synthesis of reliable organisms
  from unreliable components,'' \emph{Automata studies}, vol.~34, pp. 43--98,
  1956.

\bibitem{mcculloch1943logical}
W.~S. McCulloch and W.~Pitts, ``A logical calculus of the ideas immanent in
  nervous activity,'' \emph{The bulletin of mathematical biophysics}, vol.~5,
  no.~4, pp. 115--133, 1943.

\bibitem{sreenivasan2011error}
S.~Sreenivasan and I.~Fiete, ``Error correcting analog codes in the brain:
  beyond classical population coding for exponentially precise computation,''
  \emph{Nature Neuroscience}, vol.~14, pp. 1330--1337, 2011.

\bibitem{barlow1961possible}
H.~B. Barlow, ``Possible principles underlying the transformations of sensory
  messages,'' \emph{Sensory Communication}, pp. 217--234, 1961.

\bibitem{yang2017ITTrans}
Y.~Yang, P.~Grover, and S.~Kar, ``{C}omputing {L}inear {T}ransformations {W}ith
  {U}nreliable {C}omponents,'' \emph{IEEE Transactions on Information Theory},
  vol.~63, no.~6, pp. 3729--3756, 2017.

\bibitem{yanushkevich2013introduction}
S.~N. Yanushkevich, S.~Kasai, G.~Tangim, and A.~Tran, \emph{Introduction to
  Noise-Resilient Computing}.\hskip 1em plus 0.5em minus 0.4em\relax Morgan \&
  Claypool Publishers, 2013.

\bibitem{olshausen1996emergence}
B.~A. Olshausen and D.~J. Field, ``Emergence of simple-cell receptive field
  properties by learning a sparse code for natural images,'' \emph{Nature},
  vol. 381, no. 6583, p. 607, 1996.

\bibitem{whatmough2018dnn}
P.~N. Whatmough, S.~K. Lee, D.~Brooks, and G.-Y. Wei, ``Dnn engine: A 28-nm
  timing-error tolerant sparse deep neural network processor for iot
  applications,'' \emph{IEEE Journal of Solid-State Circuits}, vol.~53, no.~9,
  pp. 2722--2731, 2018.

\bibitem{wang2018training}
N.~Wang, J.~Choi, D.~Brand, C.-Y. Chen, and K.~Gopalakrishnan, ``Training deep
  neural networks with 8-bit floating point numbers,'' in \emph{Advances in
  neural information processing systems}, 2018, pp. 7686--7695.

\bibitem{truenorth}
\BIBentryALTinterwordspacing
P.~A. Merolla, J.~V. Arthur, R.~Alvarez-Icaza, A.~S. Cassidy, J.~Sawada,
  F.~Akopyan, B.~L. Jackson, N.~Imam, C.~Guo, Y.~Nakamura, B.~Brezzo, I.~Vo,
  S.~K. Esser, R.~Appuswamy, B.~Taba, A.~Amir, M.~D. Flickner, W.~P. Risk,
  R.~Manohar, and D.~S. Modha, ``A million spiking-neuron integrated circuit
  with a scalable communication network and interface,'' \emph{Science}, vol.
  345, no. 6197, pp. 668--673, 2014. [Online]. Available:
  \url{http://science.sciencemag.org/content/345/6197/668}
\BIBentrySTDinterwordspacing

\bibitem{binaryconnect}
M.~Courbariaux, Y.~Bengio, and J.-P. David, ``Binaryconnect: Training deep
  neural networks with binary weights during propagations,'' in \emph{Advances
  in Neural Information Processing Systems (NIPS)}, 2015, pp. 3123--3131.

\bibitem{sakr2017analytical}
C.~Sakr, Y.~Kim, and N.~Shanbhag, ``Analytical guarantees on numerical
  precision of deep neural networks,'' in \emph{Proceedings of the 34th
  International Conference on Machine Learning-Volume 70}.\hskip 1em plus 0.5em
  minus 0.4em\relax JMLR. org, 2017, pp. 3007--3016.

\bibitem{sakr2019accumulation}
C.~Sakr, N.~Wang, C.-Y. Chen, J.~Choi, A.~Agrawal, N.~Shanbhag, and
  K.~Gopalakrishnan, ``Accumulation bit-width scaling for ultra-low precision
  training of deep networks,'' \emph{arXiv preprint arXiv:1901.06588}, 2019.

\bibitem{wang2019deep}
E.~Wang, J.~J. Davis, R.~Zhao, H.-C. Ng, X.~Niu, W.~Luk, P.~Y. Cheung, and
  G.~A. Constantinides, ``Deep neural network approximation for custom
  hardware: Where we've been, where we're going,'' \emph{arXiv preprint
  arXiv:1901.06955}, 2019.

\bibitem{yang2018bit}
L.~Yang, D.~Bankman, B.~Moons, M.~Verhelst, and B.~Murmann, ``Bit error
  tolerance of a cifar-10 binarized convolutional neural network processor,''
  in \emph{2018 IEEE International Symposium on Circuits and Systems
  (ISCAS)}.\hskip 1em plus 0.5em minus 0.4em\relax IEEE, 2018, pp. 1--5.

\bibitem{rosenblatt1958perceptron}
F.~Rosenblatt, ``The perceptron: A probabilistic model for information storage
  and organization in the brain.'' \emph{Psychological review}, vol.~65, no.~6,
  p. 386, 1958.

\bibitem{rumelhart1986learning}
D.~E. Rumelhart, G.~E. Hinton, and R.~J. Williams, ``Learning representations
  by back-propagating errors,'' \emph{Nature}, vol. 323, no. 6088, p. 533,
  1986.

\bibitem{krizhevsky2012imagenet}
A.~Krizhevsky, I.~Sutskever, and G.~E. Hinton, ``Imagenet classification with
  deep convolutional neural networks,'' in \emph{Advances in Neural Information
  Processing Systems (NIPS)}, 2012, pp. 1097--1105.

\bibitem{taigman2014deepface}
Y.~Taigman, M.~Yang, M.~Ranzato, and L.~Wolf, ``Deepface: Closing the gap to
  human-level performance in face verification,'' in \emph{IEEE conference on
  computer vision and pattern recognition}, 2014, pp. 1701--1708.

\bibitem{he2016deep}
K.~He, X.~Zhang, S.~Ren, and J.~Sun, ``Deep residual learning for image
  recognition,'' in \emph{IEEE conference on computer vision and pattern
  recognition}, 2016, pp. 770--778.

\bibitem{chan2007collective}
E.~Chan, M.~Heimlich, A.~Purkayastha, and R.~Van De~Geijn, ``Collective
  communication: theory, practice, and experience,'' \emph{Concurrency and
  Computation: Practice and Experience}, vol.~19, no.~13, pp. 1749--1783, 2007.

\bibitem{geist2016supercomputing}
A.~Geist, ``Supercomputing's monster in the closet,'' \emph{IEEE Spectrum},
  vol.~53, no.~3, pp. 30--35, 2016.

\bibitem{ziegler1996terrestrial}
J.~F. Ziegler, ``Terrestrial cosmic rays,'' \emph{IBM journal of research and
  development}, vol.~40, no.~1, pp. 19--39, 1996.

\bibitem{Tay_Bel_68}
M.~G. Taylor, ``Reliable information storage in memories designed from
  unreliable components,'' \emph{Bell Syst. Tech. J.}, vol.~47, no.~10, pp.
  2299--2337, 1968.

\bibitem{hadjicostis2005coding}
C.~N. Hadjicostis and G.~C. Verghese, ``Coding approaches to fault tolerance in
  linear dynamic systems,'' \emph{IEEE Transactions on Information Theory},
  vol.~51, no.~1, pp. 210--228, 2005.

\bibitem{Pip_TIT_91}
N.~Pippenger, G.~Stamoulis, and J.~Tsitsiklis, ``On a lower bound for the
  redundancy of reliable networks with noisy gates,'' \emph{IEEE Transactions
  on Information Theory}, vol.~37, no.~3, pp. 639--643, May 1991.

\bibitem{Spi_FCS_96}
D.~A. Spielman, ``Highly fault-tolerant parallel computation,'' in
  \emph{Symposium on Foundations of Computer Science}, Oct 1996, pp. 154--163.

\bibitem{faultbook}
T.~Herault and Y.~Robert, \emph{{Fault-Tolerance Techniques for High
  Performance Computing}}.\hskip 1em plus 0.5em minus 0.4em\relax Springer,
  2015.

\bibitem{ABFT1984}
K.~H. Huang and J.~A. Abraham, ``Algorithm-based fault tolerance for matrix
  operations,'' \emph{IEEE Transactions on Computers}, vol. 100, no.~6, pp.
  518--528, 1984.

\bibitem{NewsletterPaper}
V.~Cadambe and P.~Grover, ``Codes for {D}istributed {C}omputing: {A}
  {T}utorial,'' \emph{IEEE Information Theory Society Newsletter}, vol.~67,
  no.~4, pp. 3--15, Dec. 2017.

\bibitem{gauri2014delay}
G.~Joshi, Y.~Liu, and E.~Soljanin, ``On the delay-storage trade-off in content
  download from coded distributed storage systems,'' \emph{IEEE Journal on
  Selected Areas in Communications}, vol.~32, no.~5, pp. 989--997, 2014.

\bibitem{gauri2015straggler}
D.~Wang, G.~Joshi, and G.~Wornell, ``{Using Straggler Replication to Reduce
  Latency in Large-scale Parallel Computing},'' in \emph{ACM SIGMETRICS
  Performance Evaluation Review}, vol.~43, no.~3, 2015, pp. 7--11.

\bibitem{gauri2014efficient}
D.~Wang, G.~Joshi, and G.~Wornell, ``Efficient {T}ask {R}eplication for {F}ast
  {R}esponse {T}imes in {P}arallel {C}omputation,'' in \emph{ACM SIGMETRICS
  Performance Evaluation Review}, vol.~42, no.~1, 2014, pp. 599--600.

\bibitem{lee2016speeding}
K.~Lee, M.~Lam, R.~Pedarsani, D.~Papailiopoulos, and K.~Ramchandran, ``Speeding
  up distributed machine learning using codes,'' in \emph{IEEE International
  Symposium on Information Theory (ISIT)}, 2016, pp. 1143--1147.

\bibitem{lee2018speeding}
K.~Lee, M.~Lam, R.~Pedarsani, D.~Papailiopoulos, and K.~Ramchandran, ``Speeding
  up distributed machine learning using codes,'' \emph{IEEE Transactions on
  Information Theory}, vol.~64, no.~3, pp. 1514--1529, 2018.

\bibitem{dutta2016short}
S.~Dutta, V.~Cadambe, and P.~Grover, ``{Short-Dot: Computing Large Linear
  Transforms Distributedly Using Coded Short Dot Products},'' in \emph{Advances
  In Neural Information Processing Systems (NIPS)}, 2016, pp. 2092--2100.

\bibitem{dutta2017coded}
S.~Dutta, V.~Cadambe, and P.~Grover, ``Coded convolution for parallel and
  distributed computing within a deadline,'' in \emph{IEEE International
  Symposium on Information Theory (ISIT)}, 2017, pp. 2403--2407.

\bibitem{lee2017matrix}
K.~Lee, C.~Suh, and K.~Ramchandran, ``High-dimensional coded matrix
  multiplication,'' in \emph{IEEE International Symposium on Information Theory
  (ISIT)}, 2017, pp. 2418--2422.

\bibitem{dutta2018optimal}
S.~Dutta, M.~Fahim, F.~Haddadpour, H.~Jeong, V.~Cadambe, and P.~Grover, ``On
  the optimal recovery threshold of coded matrix multiplication,'' \emph{arXiv
  preprint arXiv:1801.10292}, 2018.

\bibitem{dutta2018unified}
S.~Dutta, Z.~Bai, H.~Jeong, T.~M. Low, and P.~Grover, ``A unified coded deep
  neural network training strategy based on generalized polydot codes for
  matrix multiplication,'' \emph{arXiv preprint arXiv:1811.10751}, 2018.

\bibitem{yu2017polynomial}
Q.~Yu, M.~A. Maddah-Ali, and A.~S. Avestimehr, ``{P}olynomial {C}odes: an
  {O}ptimal {D}esign for {H}igh-{D}imensional {C}oded {M}atrix
  {M}ultiplication,'' in \emph{Advances In Neural Information Processing
  Systems (NIPS)}, 2017, pp. 4403--4413.

\bibitem{GC1}
R.~Tandon, Q.~Lei, A.~G. Dimakis, and N.~Karampatziakis, ``Gradient coding,''
  in \emph{Machine Learning Systems Workshop, Advances in Neural Information
  Processing Systems (NIPS)}, 2016.

\bibitem{GC2}
R.~Tandon, Q.~Lei, A.~G. Dimakis, and N.~Karampatziakis, ``Gradient {C}oding:
  Avoiding {S}tragglers in {D}istributed {L}earning,'' in \emph{International
  Conference on Machine Learning (ICML)}, 2017, pp. 3368--3376.

\bibitem{yu2018entangled}
Q.~Yu, M.~A. Maddah-Ali, and A.~S. Avestimehr, ``Straggler mitigation in
  distributed matrix multiplication: Fundamental limits and optimal coding,''
  in \emph{IEEE International Symposium on Information Theory (ISIT)}, 2018,
  pp. 2022 -- 2026.

\bibitem{li2015coded}
S.~Li, M.~A. Maddah-Ali, and A.~S. Avestimehr, ``Coded mapreduce,'' in
  \emph{IEEE Communication, Control, and Computing (Allerton)}, 2015, pp.
  964--971.

\bibitem{GC3}
N.~Raviv, I.~Tamo, R.~Tandon, and A.~G. Dimakis, ``Gradient coding from cyclic
  mds codes and expander graphs,'' \emph{arXiv preprint arXiv:1707.03858},
  2017.

\bibitem{yang2016logistic}
Y.~Yang, P.~Grover, and S.~Kar, ``Fault-tolerant distributed logistic
  regression using unreliable components,'' in \emph{Communication, Control,
  and Computing (Allerton)}, 2016, pp. 940--947.

\bibitem{yang2017encoded}
Y.~Yang, P.~Grover, and S.~Kar, ``{C}omputing {L}inear {T}ransformations {W}ith
  {U}nreliable {C}omponents,'' \emph{IEEE Transactions on Information Theory},
  vol.~63, no.~6, 2017.

\bibitem{yang2016convolution}
Y.~Yang, P.~Grover, and S.~Kar, ``Fault-tolerant parallel linear filtering
  using compressive sensing,'' in \emph{IEEE International Symposium on Turbo
  Codes and Iterative Information Processing (ISTC)}, 2016, pp. 201--205.

\bibitem{yang2016encoded}
Y.~Yang, P.~Grover, and S.~Kar, ``{C}omputing {L}inear {T}ransformations {W}ith
  {U}nreliable {C}omponents,'' in \emph{IEEE International Symposium on
  Information Theory (ISIT)}, 2016.

\bibitem{Salman1}
S.~Li, M.~Maddah-Ali, Q.~Yu, and A.~S. Avestimehr, ``A {F}undamental {T}radeoff
  {B}etween {C}omputation and {C}ommunication in {D}istributed {C}omputing,''
  \emph{IEEE Transactions on Information Theory}, vol.~64, no.~1, pp. 109--128,
  2018.

\bibitem{Salman2}
S.~Li, M.~A. Maddah-Ali, and A.~S. Avestimehr, ``A {U}nified {C}oding
  {F}ramework for {D}istributed {C}omputing with {S}traggling {S}ervers,'' in
  \emph{Globecom Workshops (GC Wkshps)}, 2016, pp. 1--6.

\bibitem{Salman3}
S.~Li, S.~Supittayapornpong, M.~A. Maddah-Ali, and A.~S. Avestimehr, ``Coded
  {T}era{S}ort,'' in \emph{IEEE International Parallel and Distributed
  Processing Symposium Workshops (IPDPSW)}, 2017, pp. 389--398.

\bibitem{Salman4}
S.~Li, M.~A. Maddah-Ali, and A.~S. Avestimehr, ``{C}oded {D}istributed
  {C}omputing: {S}traggling {S}ervers and {M}ultistage {D}ataflows,'' in
  \emph{Communication, Control, and Computing (Allerton)}, 2016, pp. 164--171.

\bibitem{Emina1}
M.~Aktas, P.~Peng, and E.~Soljanin, ``Effective {S}traggler {M}itigation:
  {W}hich {C}lones {S}hould {A}ttack and {W}hen?'' \emph{ACM SIGMETRICS
  Performance Evaluation Review}, vol.~45, no.~2, pp. 12--14, 2017.

\bibitem{Emina2}
M.~Aktas, P.~Peng, and E.~Soljanin, ``Straggler {M}itigation by {D}elayed
  {R}elaunch of {T}asks,'' \emph{ACM SIGMETRICS Performance Evaluation Review},
  vol.~45, no.~2, pp. 224--231, 2018.

\bibitem{Virtualization}
M.~Aliasgari, J.~Kliewer, and O.~Simeone, ``{Coded Computation Against
  Straggling Decoders for Network Function Virtualization},'' \emph{arXiv
  preprint arXiv:1709.01031}, 2017.

\bibitem{heterogeneousclusters}
A.~Reisizadeh, S.~Prakash, R.~Pedarsani, and A.~S. Avestimehr, ``Coded
  computation over heterogeneous clusters,'' in \emph{IEEE International
  Symposium on Information Theory (ISIT)}, 2017, pp. 2408--2412.

\bibitem{GC4}
W.~Halbawi, N.~Azizan-Ruhi, F.~Salehi, and B.~Hassibi, ``Improving distributed
  gradient descent using reed-solomon codes,'' \emph{arXiv preprint
  arXiv:1706.05436}, 2017.

\bibitem{Suhas1}
C.~Karakus, Y.~Sun, and S.~Diggavi, ``Encoded distributed optimization,'' in
  \emph{IEEE International Symposium on Information Theory (ISIT)}, 2017, pp.
  2890--2894.

\bibitem{Suhas2}
C.~Karakus, Y.~Sun, S.~Diggavi, and W.~Yin, ``{Straggler Mitigation in
  Distributed Optimization through Data Encoding},'' in \emph{Advances in
  Neural Information Processing Systems (NIPS)}, 2017, pp. 5440--5448.

\bibitem{yang2017NIPS}
Y.~Yang, P.~Grover, and S.~Kar, ``{Coded Distributed Computing for Inverse
  Problems},'' in \emph{Advances in Neural Information Processing Systems
  (NIPS)}, 2017, pp. 709--719.

\bibitem{Ramtin1}
A.~Reisizadeh and R.~Pedarsani, ``{Latency Analysis of Coded Computation
  Schemes over Wireless Networks},'' \emph{arXiv preprint arXiv:1707.00040},
  2017.

\bibitem{lee2017multicore}
K.~Lee, R.~Pedarsani, D.~Papailiopoulos, and K.~Ramchandran, ``Coded
  computation for multicore setups,'' in \emph{IEEE International Symposium on
  Information Theory (ISIT)}, 2017, pp. 2413--2417.

\bibitem{jeongFFT}
H.~Jeong, T.~M. Low, and P.~Grover, ``{Masterless Coded Computing: A
  Fully-Distributed Coded FFT Algorithm},'' \emph{Communication, Control, and
  Computing (Allerton)}, 2018.

\bibitem{baharav2018straggler}
T.~Baharav, K.~Lee, O.~Ocal, and K.~Ramchandran, ``{Straggler-Proofing
  Massive-Scale Distributed Matrix Multiplication with D-Dimensional Product
  Codes},'' in \emph{IEEE International Symposium on Information Theory
  (ISIT)}, 2018, pp. 1993--1997.

\bibitem{suh2017matrix}
G.~Suh, K.~Lee, and C.~Suh, ``Matrix sparsification for coded matrix
  multiplication,'' in \emph{Communication, Control, and Computing (Allerton)},
  2017, pp. 1271--1278.

\bibitem{mallick2018rateless}
A.~Mallick, M.~Chaudhari, and G.~Joshi, ``{Rateless Codes for Near-Perfect Load
  Balancing in Distributed Matrix-Vector Multiplication},'' \emph{arXiv
  preprint arXiv:1804.10331}, 2018.

\bibitem{wang2018coded}
S.~Wang, J.~Liu, and N.~Shroff, ``Coded sparse matrix multiplication,''
  \emph{arXiv preprint arXiv:1802.03430}, 2018.

\bibitem{aliasgari2019coded}
M.~Aliasgari, J.~Kliewer, and O.~Simeone, ``Coded computation against
  processing delays for virtualized cloud-based channel decoding,'' \emph{IEEE
  Transactions on Communications}, vol.~67, no.~1, pp. 28--38, 2019.

\bibitem{wang2018fundamental}
S.~Wang, J.~Liu, N.~Shroff, and P.~Yang, ``{Fundamental Limits of Coded Linear
  Transform},'' \emph{arXiv preprint arXiv: 1804.09791}, 2018.

\bibitem{severinson2018block}
A.~Severinson, A.~G. i~Amat, and E.~Rosnes, ``Block-diagonal and lt codes for
  distributed computing with straggling servers,'' \emph{IEEE Transactions on
  Communications}, 2018.

\bibitem{ye2018communication}
M.~Ye and E.~Abbe, ``Communication-computation efficient gradient coding,''
  \emph{arXiv preprint arXiv:1802.03475}, 2018.

\bibitem{haddadpour2018codes}
F.~Haddadpour and V.~R. Cadambe, ``Codes for distributed finite alphabet
  matrix-vector multiplication,'' in \emph{IEEE International Symposium on
  Information Theory (ISIT)}, 2018, pp. 1625--1629.

\bibitem{haddadpour2018straggler}
F.~Haddadpour, Y.~Yang, M.~Chaudhari, V.~R. Cadambe, and P.~Grover,
  ``Straggler-resilient and communication-efficient distributed iterative
  linear solver,'' \emph{arXiv preprint arXiv:1806.06140}, 2018.

\bibitem{yang2018ISIT}
Y.~Yang, P.~Grover, and S.~Kar, ``Coding for a single sparse inverse problem,''
  in \emph{IEEE International Symposium on Information Theory (ISIT)}, 2018,
  pp. 1575--1579.

\bibitem{ferdinand2016anytime}
N.~S. Ferdinand and S.~C. Draper, ``Anytime coding for distributed
  computation,'' in \emph{Communication, Control, and Computing (Allerton)},
  2016, pp. 954--960.

\bibitem{ferdinand2018hierarchical}
N.~Ferdinand and S.~C. Draper, ``Hierarchical coded computation,'' in
  \emph{IEEE International Symposium on Information Theory (ISIT)}, 2018, pp.
  1620--1624.

\bibitem{song2017pliable}
L.~Song, C.~Fragouli, and T.~Zhao, ``A pliable index coding approach to data
  shuffling,'' \emph{arXiv preprint arXiv:1701.05540}, 2017.

\bibitem{kosaian2018learning}
J.~Kosaian, K.~Rashmi, and S.~Venkataraman, ``Learning a code: Machine learning
  for approximate non-linear coded computation,'' \emph{arXiv preprint
  arXiv:1806.01259}, 2018.

\bibitem{seth2018bigdata}
U.~Sheth, S.~Dutta, M.~Chaudhari, H.~Jeong, Y.~Yang, J.~Kohonen, T.~Roos, and
  P.~Grover, ``{An Application of Storage-Optimal MatDot Codes for Coded Matrix
  Multiplication: Fast k-Nearest Neighbors Estimation},'' in \emph{IEEE Big
  Data (Short Paper)}, 2018.

\bibitem{jeong2018locally}
H.~Jeong, F.~Ye, and P.~Grover, ``{Locally Recoverable Coded Matrix
  Multiplication},'' in \emph{Communication, Control, and Computing
  (Allerton)}, 2018.

\bibitem{straggler_tail}
J.~Dean and L.~A. Barroso, ``The tail at scale,'' \emph{Communications of the
  ACM}, vol.~56, no.~2, pp. 74--80, 2013.

\bibitem{dutta2018DNN2}
S.~Dutta, Z.~Bai, H.~Jeong, T.~M. Low, and P.~Grover, ``{A Unified Coded Deep
  Neural Network Training Strategy based on Generalized PolyDot codes},'' in
  \emph{IEEE International Symposium on Information Theory (ISIT)}, 2018, pp.
  1585--1589.

\bibitem{ProductCodes}
K.~Lee, C.~Suh, and K.~Ramchandran, ``High-dimensional coded matrix
  multiplication,'' in \emph{IEEE International Symposium on Information Theory
  (ISIT)}, 2017, pp. 2418--2422.

\bibitem{bouteiller2015algorithm}
A.~Bouteiller, T.~Herault, G.~Bosilca, P.~Du, and J.~Dongarra,
  ``Algorithm-based fault tolerance for dense matrix factorizations, multiple
  failures and accuracy,'' \emph{ACM Transactions on Parallel Computing},
  vol.~1, no.~2, p.~10, 2015.

\bibitem{summa}
R.~A. Van De~Geijn and J.~Watts, ``{SUMMA: Scalable universal matrix
  multiplication algorithm},'' \emph{Concurrency-Practice and Experience},
  vol.~9, no.~4, pp. 255--274, 1997.

\bibitem{dally2015}
W.~Dally, ``High-performance hardware for machine learning,'' \emph{NIPS
  Tutorial}, 2015.

\bibitem{li2007memory}
X.~Li, K.~Shen, M.~C. Huang, and L.~Chu, ``{A Memory Soft Error Measurement on
  Production Systems},'' in \emph{USENIX Annual Technical Conference}, 2007, p.
  275–280.

\bibitem{lecun1998mnist}
Y.~LeCun, C.~Cortes, and C.~J. Burges, ``The {MNIST} database of handwritten
  digits,'' \url{http://yann. lecun. com/exdb/mnist}, 1998.

\bibitem{ryan2009channel}
W.~Ryan and S.~Lin, \emph{Channel codes: {C}lassical and {M}odern}.\hskip 1em
  plus 0.5em minus 0.4em\relax Cambridge University Press, 2009.

\bibitem{candes2005decoding}
E.~J. Candes and T.~Tao, ``Decoding by linear programming,'' \emph{IEEE
  Transactions on Information Theory}, vol.~51, no.~12, pp. 4203--4215, 2005.

\bibitem{hebbianlearning}
D.~O. Hebb, \emph{{The Organizations of Behavior: a Neuropsychological
  Theory}}.\hskip 1em plus 0.5em minus 0.4em\relax Chapman and Hall, 1957.

\bibitem{donoho2006compressed}
D.~L. Donoho, ``Compressed sensing,'' \emph{IEEE Transactions on Information
  Theory}, vol.~52, no.~4, pp. 1289--1306, 2006.

\end{thebibliography}
\bibliographystyle{IEEEtran}

\noindent \textbf{Acknowledgements:} 
%
The authors thank Haewon Jeong, Viveck Cadambe, Mohammad Fahim, Farzin Haddadpour, Yaoqing Yang, Anit Sahu, Gauri Joshi and Ankur Mallick for helpful conversations. This work was supported in part by NSF CNS-1702694, and CCF-1350314.  This work was also supported in part by Systems on Nanoscale Information fabriCs (SONIC), one of the six SRC STARnet Centers, sponsored by MARCO and DARPA.\\


\appendices
\crefalias{section}{app}
\crefalias{subsection}{app}
\section{DNN Background and Notations}
\label{appendix:DNN_background}

We  follow the standard notation used in DNN literature (e.g.~\cite{rumelhart1986learning}), so familiar readers can merely skim this part. We assume that we are training a DNN with layers $l=1,2,\ldots,L$ with backpropagation algorithm using Stochastic Gradient Descent (with batch size $1$). We also use the index $k$ to denote the iteration number of the training. At the $k$-th iteration (for every $k$), the neural network is trained based on a single data point using three stages: a feedforward stage, a backpropagation stage and an update stage. At the $l$-th layer, $N_{l}$ denotes the number the neurons. We assume the following:

\begin{enumerate}

\item $N_1,N_2,\ldots,N_L :$ Number of neurons at layers $1,2,\ldots,L$. We also introduce the notation $N_0$ to denote the dimension of the original data vector, which serves as the input to the first layer.

\item $W_{i,j}^l (k):$ At iteration $k$, the weight of the connection from neuron $j$ on layer $l-1$ to neuron $i$ on layer $l$ for $i=0,1,\ldots,N_l-1$ and $ j=0,1,\ldots,N_{l-1}-1$. Note that the weights  form a matrix $\bm{W}^l(k)$ of dimension $N_l \times N_{l-1}$ for layer $l$.

\item $\bm{x}^l(k) \in \mathcal{R}^{N_{l-1}}  : $ The input of layer $l$ at the $k$-th iteration. Note that, for the first layer, $\bm{x}^1(k)$ is the data used for the $k$-th iteration of training.

\item $\hat{\bm{y}}^l(k) \in \mathcal{R}^{N_{l}}  : $ The output of layer $l$ at the $k$-th iteration.

\item $ \bm{s}^l(k): $ The summed output of the neurons of layer $l$ before a nonlinear function $f(.)$ is applied on it, at the $k$-th iteration. Note that, for $i=0,1,\ldots,N_l-1$, the scalar $s_i^l(k)$ is the $i$-th entry of the vector $\bm{s}^l(k)$, \textit{i.e.} the summed output of neuron $i$ on layer $l$.

\end{enumerate}

\noindent We now make some observations:
\begin{enumerate}
\item Input to any layer is the output of the previous layer (except, of course, for the first layer whose input is the actual data vector itself):

$$\bm{x}^l(k)=
\begin{cases}
    \hat{\bm{y}}^{l-1}(k) \in \mathcal{R}^{N_{l-1}} ,& \text{if } l=2,3,\ldots,L\\
    \bm{x}^1(k),              & \text{otherwise}
\end{cases} $$

\item At each layer, the input to that layer is summed with appropriate weights of that layer (elements $W_{i,j}^l$) to produce the summed output of each neuron, given by:
\begin{align}
s_i^l(k) & =\bm{W}^l_{i,:}(k)\bm{x}^{l}(k) = \sum_{j=0}^{N_{l-1}-1} W_{i,j}^{l}(k) x^{l}_j(k) \\
& =  \sum_{j=0}^{N_{l-1}-1} W_{i,j}^{l}(k) x^{l}_j(k)  = \sum_{j=0}^{N_{l-1}-1} W_{i,j}^{l}(k) \hat{y}^{l-1}_j(k)  
\end{align}

\item The final output of each layer is given by a nonlinear function applied on the summed output of each neuron as below:
\begin{align}
\hat{y}_i^l (k) = f(s_i^l(k)) = f\left(\sum_{j=1}^{N_{l-1}} W_{i,j}^{l}(k) \hat{y}^{l-1}_j(k)   \right)
\end{align}

\item Observe that $\hat{y}^L(k)$ of the last layer denotes the estimated output of the DNN and is to be compared with the true label vector $\bm{y}(k)$ for the corresponding data point $\bm{x}^1(k)$.

\end{enumerate}

\noindent \textbf{Training of a DNN:}
\noindent We now detail the steps of the training algorithm. Note that the ``loss function'' in the $k$-th iteration is defined as:

\begin{equation}
\epsilon^2(k) = \sum_{i=0}^{N_L-1} \epsilon_i^2=\sum_{i=0}^{N_L-1} (\hat{y}_i^l(k)-y_i)^2 
\end{equation}

While we choose a squared error loss function here, the analysis easily generalizes to any loss function, such as regularized squared error loss, or soft-max loss function.
\noindent At the $k$-th iteration, the weights of every layer $l$ of the DNN are updated as follows:
\begin{equation}
W_{i,j}^l(k+1)=W_{i,j}^l(k) - \eta \frac{\partial \epsilon^2(k)}{\partial W_{i,j}^l(k) }
\end{equation}

\noindent Backpropagation algorithm helps us to compute the errors and updates in a recursive form, so that the update of any layer $l$ depends only on the backpropagated error vector of its succeeding layer, \textit{i.e.} layer $l+1$ and not on the layers $\{l+1, l+2, \ldots ,L\}$. Here, we provide the rules of the update stage. Let us define the backpropagated error vector as $\bm{\delta}^l(k)$:
\begin{equation}
\delta_i^l(k)= -\frac{\partial \epsilon^2(k)}{\partial s_i^l(k) } \ \ \forall \ i=0, 1, \ldots, N_l-1
\end{equation}
As a special case, we have that,
\begin{equation}
\delta_i^L(k)= 2 \epsilon_i(k) f'(s_i^L(k))
\end{equation}

\noindent Now, observe that $\delta_i^l(k)$ can be calculated from $\bm{\delta}^{l+1}(k)$ as we derive here in Lemma \ref{lem:DNN}.

\begin{lem}
\label{lem:DNN}
During the training of a neural network using backpropagation algorithm, the backpropagated error vector $\delta_i^{l}(k)$ for any layer $l$ can be expressed as a function of the backpropagated error vector of the previous layer as given by:
\begin{equation}
\delta_i^{l}(k) = \left( \sum_{j=0}^{N_{l+1}-1} \delta_j^{l+1}(k) W_{j,i}^l(k)  \right) f'(s_i^l(k))
\end{equation}
\end{lem}

\begin{proof}[Proof of Lemma \ref{lem:DNN}]
\begin{align}
\delta_i^l(k) & = -\frac{\partial \epsilon^2(k)}{\partial s_i^l(k) } = -\sum_{j=0}^{N_{l+1}-1} \frac{\partial \epsilon^2(k)}{\partial s_j^{l+1}(k) } \frac{\partial s_j^{l+1}(k)}{\partial s_i^l(k) } \\
& = -\sum_{j=0}^{N_{l+1}-1} \frac{\partial \epsilon^2(k)}{\partial s_j^{l+1}(k) } W^l_{j,i}(k) f'(s^l_i(k)) = \left( \sum_{j=0}^{N_{l+1}-1} \delta_j^{l+1}(k) W_{j,i}^l(k)  \right) f'(s_i^l(k)).
\end{align}
\end{proof}
Now, using the fact that $s_i^l(k)  = \sum_{j=0}^{N_{l-1}-1} W_{i,j}^{l}(k) x^{l}_j(k)$, we have
\begin{equation}
\frac{\partial \epsilon^2(k)}{\partial W_{i,j}^l(k) } = \frac{\partial \epsilon^2(k)}{\partial s_i^l(k) }\frac{\partial s_i^l(k)}{\partial W_{i,j}^l(k) } = -\delta_i^l(k) x^{l}_j(k).
\end{equation}
Thus, the update rule for backpropagation algorithm is derived as follows:
\begin{equation}
W_{i,j}^l(k+1)=W_{i,j}^l(k) - \eta \frac{\partial \epsilon^2(k)}{\partial W_{i,j}^l(k) }= W_{i,j}^l(k) + \eta \delta_i^l(k) x^{l}_j(k).
\end{equation}

\noindent \textbf{Training of a single layer:}

Let us first look into the operations that are performed in a single layer of DNN in matrix notation:

\begin{itemize}
\item Feedforward stage: In this stage, the data for the $k$-th iteration, \textit{i.e.} $\bm{x}^1(k)$ is obtained and the goal is to pass it through all the layers to compute the estimate of the label $\hat{\bm{y}}^{L}(k)$. Thus, the operation at each layer is given by:

\begin{itemize}
\item Obtain $\bm{x}^l(k)$ from $(l-1)$-th layer
\item Compute $\bm{s}^l(k)= \bm{W}^l(k) \bm{x}^l(k)$
\item Send $\bm{x}^{l+1}(k)=\hat{\bm{y}}(k)=f(\bm{s}^l(k))$ to $(l+1)$-th layer 
\end{itemize}

\item Backpropagation stage: Given the true label vector $\bm{y}$ and the estimate $\hat{\bm{y}}^{L}(k)$, the goal is to find the backpropagated errors for all the layers, and update the weights $W_{i,j}^l(k)$. We also assume that the nonlinear function is such that $f'(u)=g(f(u))$ for some function $g(u)$.
\begin{itemize}
\item Obtain $\bm{\delta}^l(k)$ from $(l+1)$-th layer
\item Compute backpropagated error for $(l-1)$-th layer as: \\ $$\bm{\delta}^{l-1}(k)= \begin{bmatrix}
g(x_0^l(k)) & 0 & 0 \\
\vdots & \ddots & \vdots \\
0 & 0 &
g(x_{N_{l-1}-1}^l(k)) \end{bmatrix}
[\bm{W}^l(k)]^T \bm{\delta}^l(k) = \begin{bmatrix}
g(x_0^l(k)) & 0 & 0 \\
\vdots & \ddots & \vdots \\
0 & 0 &
g(x_{N_{l-1}-1}^l(k)) \end{bmatrix}\bm{c}^l(k)$$
\item Send $\bm{\delta}^{l-1}(k)$ to $(l-1)$-th layer 
\end{itemize}

\item Update stage:
\begin{itemize}
\item Update: 
$\bm{W}^l(k+1)=\bm{W}^l(k) + \eta \bm{\delta}^l(k) [\bm{x}^l(k)]^T$
\end{itemize}
\end{itemize}

\noindent \textbf{Remark:} {Note that, the operations of the DNN training are same across all the layers, with the matrix-vector products and rank-$1$ update being the most complexity-intensive steps (also proved formally in Theorem~\ref{thm:complexity}). Moreover, to circumvent the nonlinear activation step between two consecutive layers, we apply linear coding on each layer separately, making the overall strategy very similar for each layer. Thus, it is sufficient to describe the strategy for a single layer in a single iteration. For the purpose of understanding as it naturally applies to all the layers at all iterations. So we simplify our notations in the main part of the paper as follows: For any layer, we denote the feedforward input $\bm{x}^l(k)$, the weight matrix $\bm{W}^l(k)$ and the backpropagated error $\bm{\delta}^l(k)$ as $\bm{x}^l$, $\bm{W}^l$ and $\bm{\delta}^l$ respectively. }



\section{Error Detection and Correction Mechanism of CodeNet}
\label{appendix:error_detection}

Here, we describe the error detection and correction mechanism of CodeNet and also provide proofs of Theorem~\ref{thm:recovery_threshold}, and Claims~\ref{clm:general_encoding}, \ref{clm:error_tolerance} and \ref{clm:error_detection}. For the general case, to be able to correct $(t_1+t_3)$ errors after step O1 and $(t_2+t_3)$ errors after step O2, we use a set of $P=mn$ ``base'' nodes arranged in an $m \times n$ grid, and an additional $2n(t_1+t_3)+2m(t_2+t_3)$ ``redundant'' nodes for error resilience. We add $2(t_1+t_3))$ rows of $n$ nodes each for error-detection and correction after step O1 and $2(t_2+t_3)$ columns of $m$ nodes each for the same after step O2. For simplicity, we restrict our description to the case with $(t_1+t_3)=(t_2+t_3)=t$, though the same techniques easily generalizes.

For every layer, the weight matrix $\bm{W}^l$ of dimension $N_{l} \times N_{l-1}$ is first divided into $P$ equal sub-matrices across the grid of $m \times n$ base nodes, thus satisfying the storage constraint (each node can only store $\frac{N_{l} N_{l-1} }{P}$ elements) for each layer. Thus,
\begin{align*}  
\bm{W}^l & = \begin{bmatrix}
            \bm{W}^l_{0,0} & \ldots &\bm{W}^l_{0,n-1} \\
            \vdots & \ddots & \vdots \\
            \bm{W}^l_{m-1,0} & \ldots  &\bm{W}^l_{m-1,n-1}
          \end{bmatrix} = \begin{bmatrix}
            \bm{W}^l_{0,:}   \\
            \vdots \\
            \bm{W}^l_{m-1,:} 
          \end{bmatrix} \Bigg\rbrace \text{Row Blocks} \\
          & = \underbrace{\begin{bmatrix}
             \ \bm{W}^l_{:,0} \ & \ldots &  \ \bm{W}^l_{:,n-1} \      \end{bmatrix}.}_{\text{Column Blocks}} 
\end{align*}

\noindent \textbf{Brief Overview of the coding technique:} To be able to correct any $t$ errors under Error Model $1$ or $2$ after steps O1 and O2, the encoding is as follows:
\begin{align}
\left(\bm{G}_r^T  \otimes \bm{I}_{N_{l}/m}  \right) \bm{W}^l \left( \bm{G}_c \otimes \bm{I}_{N_{l-1}/n}\right).
\end{align}
Here $\bm{G}_r$ and $\bm{G}_c$ are the generator matrices of a systematic $(m+2t, m)$ MDS code and a systematic $(n+2t, n)$ MDS code respectively. The blocks of $\bm{W}^l$ are of size $\frac{N_{l-1}}{m} \times \frac{N_l}{n} $, which explains the Kronecker product with $\bm{I}_{N_{l-1}/m} $ and $\bm{I}_{N_l/n}$ respectively. To remind the reader, the Kronecker Product between two matrices $\bm{A}$ and $\bm{B}$ is defined as follows: Suppose $\bm{A}$ is an $m \times n$ matrix and $\bm{B}$ is a $p \times q$ matrix, then $\bm{A} \otimes \bm{B}$ is an $ mp \times nq$ matrix given by $ \begin{bmatrix}
A_{0,0}\bm{B} & \ldots & A_{0,n-1}\bm{B} \\
\vdots & \ddots & \vdots \\
A_{m-1,0}\bm{B} & \ldots & A_{m-1,n-1}\bm{B} 
\end{bmatrix}$.

Thus, after encoding, the matrices on the grid are as follows:
\begin{align*}  
  \begin{bmatrix}
            \bm{W}^l_{0,0} & \ldots &\bm{W}^l_{0,n-1} &  \widetilde{\bm{W}}^l_{0,n} & \ldots & \widetilde{\bm{W}}^l_{0,n+2t-1} \\
            \vdots & \ddots & \vdots & \vdots & \ddots & \vdots \\
            \bm{W}^l_{m-1,0} & \ldots  &\bm{W}^l_{m-1,n-1} & \widetilde{\bm{W}}^l_{m-1,n} & \ldots &\widetilde{\bm{W}}^l_{m-1,n+2t-1} \\
            \widetilde{\bm{W}}^l_{m,0} & \ldots  &\widetilde{\bm{W}}^l_{m,n-1} & \text{x} & \text{x} & \text{x}\\
           \vdots & \ddots  & \vdots & \text{x} & \text{x} & \text{x}\\
           \widetilde{\bm{W}}^l_{m+2t-1,0} & \ldots  &\widetilde{\bm{W}}^l_{m+2t-1,n+2t-1} & \text{x} & \text{x} & \text{x}\\
          \end{bmatrix}    
          \end{align*}
As before, x means that there is no node in that location in the grid and a superscript $\widetilde{\cdot}$ denotes a coded block.

Recall that, this pre-processing step is performed only once at the beginning of training. Subsequently, in each iteration, 
before the feedforward or backpropagation stage at any layer, assume that every node has the appropriate updated sub-matrix or coded sub-matrix of $\bm{W}$ for that iteration. This is because each node is able to update its own sub-matrix without additional communication, and all overheads associated with coding are negligible compared to the complexity of the rank-$1$ update at each node. 

Then, the feedforward stage and backpropagation stage proceed as described in the main section. For ease of understanding, one can stick to the strategy description with virtual nodes for now. 

Recall that the update rule for backpropagation algorithm is given by $\bm{W}^l + \eta \bm{\delta}^l(\bm{x}^l)^T$ (step O3). Observe that:
\begin{align}
&\bm{W}^l + \eta \begin{bmatrix}
\bm{\delta}^l_0 \\ \vdots \\ \bm{\delta}^l_{m-1}
\end{bmatrix} \begin{bmatrix}
(\bm{x}^l_0)^T & \ldots & (\bm{x}^l_{n-1})^T
\end{bmatrix} \\& =
  \begin{bmatrix}
\bm{W}^l_{0,0} + \eta \bm{\delta}^l_0(\bm{x}^l_0)^T  & \bm{W}^l_{0,1} + \eta \bm{\delta}^l_0(\bm{x}^l_1)^T & \ldots & \bm{W}^l_{0,n-1} + \eta \bm{\delta}^l_0(\bm{x}^l_{n-1})^T\\
\vdots & \vdots & \ddots & \vdots\\
 \bm{W}^l_{m-1,0} + \eta \bm{\delta}^l_{m-1}(\bm{x}^l_0)^T & \bm{W}^l_{m-1,1}  + \eta \bm{\delta}_{m-1}^l(\bm{x}^l_{1})^T & \ldots & \bm{W}^l_{m-1,n-1} + \eta \bm{\delta}_{m-1}^l(\bm{x}^l_{n-1})^T  
\end{bmatrix} 
\end{align}
Therefore, any sub-matrix of $\bm{W}^l$, e.g.~$\bm{W}^l_{i,j}$, only requires the sub-vectors $\bm{\delta}^l_i$ and $\bm{x}^l_j$ to update itself. Interestingly, our strategy also ensures this for the coded sub-matrices due to the additional encoding steps at the end of feedforward and backpropagation. Thus, every node can update itself in the update stage using one of the three rules:
\begin{align}
&\bm{W}^l_{i,j} \leftarrow \bm{W}^l_{i,j} + \bm{\delta}^l_i (\bm{x}_j^l)^T \ \forall \  0 \leq i \leq m-1  \text{ and } 0 \leq j \leq n-1 \\
& \widetilde{\bm{W}}^l_{i,j} \leftarrow \widetilde{\bm{W}}^l_{i,j} + \bm{\delta}^l_i (\widetilde{\bm{x}}^l_j)^T \ \forall \  0 \leq i \leq m-1  \text{ and }   n \leq j \leq n+2t-1 \\
& \widetilde{\bm{W}}^l_{i,j} \leftarrow \widetilde{\bm{W}}^l_{i,j} + \widetilde{\bm{\delta}}^l_i (\bm{x}^l_j)^T \ \forall \  m \leq i \leq m+2t-1  \text{ and }   0 \leq j \leq n-1
\end{align}
Note that if any of $\bm{x}^l_j$, $\widetilde{\bm{x}}^l_j$, $\bm{\delta}^l_i $ or $\widetilde{\bm{\delta}}^l_i$ are erroneous, the updated sub-matrix matrix $\bm{W}^l_{i,j}$ or $\widetilde{\bm{W}}^l_{i,j}$ is erroneous. However, in the next feedforward stage in that layer, the sub-matrix $\bm{W}^l_{i,j}$ or $\widetilde{\bm{W}}^l_{i,j}$ would again be used to compute matrix-vector products and hence would produce erroneous outputs that would propagate into $\bm{s}^l_i$ or $\widetilde{\bm{s}}^l_i$ and eventually be detected.

Now we provide proof of Theorem~\ref{thm:recovery_threshold} and then justify our Claims~\ref{clm:general_encoding},~\ref{clm:error_tolerance} and \ref{clm:error_detection} in the paper. For the proofs, we will use the following lemma.

\begin{lem}
\label{lem:error_tolerance}
Consider a matrix $\bm{W}^l_{ N_{l} \times N_{l-1}}$ which is split horizontally into $m$ equal sized row-blocks $\bm{W}^l_{0,:},\bm{W}^l_{1,:},\ldots,\bm{W}^l_{m-1,:}$, and then encoded using the generator matrix $\bm{G}_r$ of a systematic $(m+2t,m)$ MDS code as follows: $$ \begin{bmatrix}
\bm{W}^l_{0,:} \\
\vdots\\
\bm{W}^l_{m-1,:} \\
\widetilde{\bm{W}}^l_{m,:}\\
\vdots \\
\widetilde{\bm{W}}^l_{m+2t-1,:}
\end{bmatrix} =
\left( \bm{G}_r^T \otimes \bm{I}_{N_{l}/m}  \right) \bm{W}^l.$$ Then, the following holds:
\begin{enumerate}
\item The result $\bm{s}^l=\bm{W}^l\bm{x}^l$ can be decoded correctly from the $m+2t$ computations $\bm{W}^l_{i,:}\bm{x}^l$ for $0\leq i \leq m-1$ and $\widetilde{\bm{W}}^l_{i,:}\bm{x}^l$ for $m \leq i \leq m+2t-1$ if \textit{any} $t$ or fewer computations suffer from a soft-error, under either of the Error Models $1$ or $2$.
\item Moreover, under Error Model $2$, if the number of erroneous computations is more than $t$, and the error is assumed to be an additive noise whose individual elements are drawn independently from a real-valued continuous distribution, then the occurrence of erroneous computations can still be detected with probability $1$, even if they cannot be corrected. 
\end{enumerate}
\end{lem}

The proof of this lemma follows from the properties of MDS codes, as we discuss towards the end of this Appendix. First, we provide the proof of Theorem~\ref{thm:recovery_threshold} using this Lemma~\ref{lem:error_tolerance}. Again, we will stick to the simpler case $(t_1+t_3)=(t_2+t_3)=t$, though the technique easily generalizes.

\begin{proof}[Proof of Theorem~\ref{thm:recovery_threshold}]
Observe that in the feedforward stage, CodeNet uses $P=mn$ base nodes arranged in an $m \times n$ grid, with $2t$ rows of $n$ nodes each. The matrix $\bm{W}^l$ is divided into blocks both horizontally and vertically. In the vertical direction, one uses a systematic $(m+2t,m)$ MDS code to encode the horizontally-split blocks, and performs the $m+2t$ computations $\bm{W}^l_{i,:}\bm{x}^l$ for $0\leq i \leq m-1$ and $\widetilde{\bm{W}}^l_{i,:}\bm{x}^l$ for $m \leq i \leq m+2t-1$. Thus, in the feedforward stage at any layer, after step O1, CodeNet is able to correct \textit{any} $t$ of fewer errors (using Lemma~\ref{lem:error_tolerance}). Note that, CodeNet might be able to correct more than $t$ errors for some error patterns, e.g., when more than one node in a row is erroneous but still the total number of erroneous rows is at most $t$. However, in the worst case, e.g. error patterns where $t$ errors occur in different rows in the grid, it can only correct $t$ errors. 

Similarly, in the backpropagation stage, the main computation (step O2) is $(\bm{c}^l)^T=(\bm{\delta}^l)^T\bm{W}^l$, which is also a matrix-vector product. CodeNet reuses the $P=mn$ base nodes arranged in an $m \times n$ grid, with $2t$ additional columns of $m$ nodes each for error correction. The matrix $\bm{W}^l$ is divided into blocks both vertically and horizontally. One uses a systematic $(n+2t,n)$ MDS code to encode the vertically-split blocks, and performs $n+2t$ computations $(\bm{\delta}^l)^T\bm{W}^l_{:,j}$ for $0\leq j \leq n-1$ and $(\bm{\delta}^l)^T\widetilde{\bm{W}}^l_{:,j}$ for $n \leq j \leq n+2t-1$. Thus, again using Lemma~\ref{lem:error_tolerance}, CodeNet can correct any $t$ (or fewer) errors in the worst case, and the worst case pattern occurs when all the errors occur in different columns.

Errors that occur during the update stage corrupt the stored $\bm{W}^l_{i,j}$ or $\widetilde{\bm{W}}^l_{i,j}$
Thus, using a total of $mn+2tn+2tm = P+2(m+n)t$ nodes, CodeNet can correct \textit{any} $t$ errors in the worst case at any layer during feedforward or backpropagation stages respectively.

\end{proof}

\noindent \textbf{Remark:} Substituting $m=n=\sqrt{P}$, we get the number of required nodes to be $P+4\sqrt{P}t$.

\begin{proof}[Justification of Claims~\ref{clm:general_encoding}, \ref{clm:error_tolerance} and \ref{clm:error_detection}]
Claim~\ref{clm:general_encoding} follows directly from Theorem~\ref{thm:recovery_threshold}. 
It might also be noted that Claim~\ref{clm:error_tolerance} is justified from the first condition of Lemma~\ref{lem:error_tolerance}, and Claim~\ref{clm:error_detection} from the second condition of Lemma~\ref{lem:error_tolerance} respectively.
\end{proof}

Now what remains to be proved is Lemma~\ref{lem:error_tolerance}.

\begin{proof}[Proof of Lemma~\ref{lem:error_tolerance}] Let $\bm{s}^l=\begin{bmatrix}
\bm{s}^l_0\\
\vdots \\
\bm{s}^l_{m-1}
\end{bmatrix}. $ Observe that,
$$
\begin{bmatrix}
\bm{W}^l_{0,:} \\
\vdots\\
\bm{W}^l_{m-1,:} \\
\widetilde{\bm{W}}^l_{m,:}\\
\vdots \\
\widetilde{\bm{W}}^l_{m+2t-1,:}
\end{bmatrix} \bm{x}^l =  \left( \bm{G}_r^T \otimes \bm{I}_{N_{l}/m}  \right) \bm{W}^l\bm{x}^l= \left( \bm{G}_r^T \otimes \bm{I}_{N_{l}/m}  \right) \bm{s}^l = \begin{bmatrix}
\bm{s}^l_0\\
\vdots \\
\bm{s}^l_{m-1} \\
\widetilde{\bm{s}}^l_{m}\\
\vdots \\
\widetilde{\bm{s}}^l_{m+2t-1}
\end{bmatrix}.
$$
Thus, the $m+2t$ computations actually form a codeword obtained by encoding $\bm{s}^l_0,\bm{s}^l_1,\ldots,\bm{s}^l_{m-1}$ using the systematic $(m+2t,m)$ MDS code with generator matrix $\bm{G}_r$. Now suppose that \textit{any} $t$ of the computations $\bm{s}^l_0,\ldots,\bm{s}^l_{m-1} ,\widetilde{\bm{s}}^l_{m},\ldots,\widetilde{\bm{s}}^l_{m+2t-1}$ are erroneous. Then, using the properties of MDS codes\cite{ryan2009channel}, any $t$ erroneous computations can be corrected in the codeword of an $(m+2t,m)$ MDS code. This proves the first part of the lemma. However, for completeness, we also provide a rigorous proof using linear algebra arguments.

Let $\bm{out}$ denote the vector consisting of $m+2t$ sub-vectors, each of length $\frac{N_{l}}{m}$, which is the entire codeword corrupted by an additive vector $\bm{e}$. Thus,
$$\bm{out} = \begin{bmatrix}
\bm{out}_0\\
\vdots \\
\bm{out}_{m+2t-1}
\end{bmatrix}= \begin{bmatrix}
\bm{s}^l_0\\
\vdots \\
\bm{s}^l_{m-1} \\
\widetilde{\bm{s}}^l_{m}\\
\vdots \\
\widetilde{\bm{s}}^l_{m+2t-1}
\end{bmatrix} + \bm{e} \ = \begin{bmatrix}
\bm{s}^l_0\\
\vdots \\
\bm{s}^l_{m-1} \\
\widetilde{\bm{s}}^l_{m}\\
\vdots \\
\widetilde{\bm{s}}^l_{m+2t-1}
\end{bmatrix} + \begin{bmatrix}
\bm{e}_0\\
\vdots \\
\bm{e}_{m-1}\\
\bm{e}_{m}\\
\vdots\\
\bm{e}_{m+2t-1}
\end{bmatrix} = \left( \bm{G}_r^T \otimes \bm{I}_{N_{l}/m}  \right) \begin{bmatrix}
\bm{s}^l_0\\
\vdots \\
\bm{s}^l_{m-1} 
\end{bmatrix} + 
\begin{bmatrix}
\bm{e}_0\\
\vdots \\
\bm{e}_{m-1}\\
\bm{e}_{m}\\
\vdots\\
\bm{e}_{m+2t-1}
\end{bmatrix}.$$
Note that $\bm{e}$ is a vector, consisting of $m+2t$ sub-vectors, each of length $\frac{N_l}{m}$. When any of the $m+2t$ sub-vectors, e.g., $\bm{out}_i$, is erroneous, then the sub-vector $\bm{e}_i$ is non-zero (in accordance with our Error Models). Otherwise, if there are no errors in $\bm{out}_i$, then $\bm{e}_i =\bm{0}$. The non-zero sub-vectors of $\bm{e}$ thus correspond to the locations of errors among the $m+2t$ computations. In what follows, we utilize ideas from~\cite{candes2005decoding,donoho2006compressed} to show that any vector $\bm{e}$ with number of non-zero sub-vectors (\textit{i.e.},  $\bm{e}_0,\bm{e}_1,\ldots,\bm{e}_{m+2t-1}$) less that or equal to $t$ can be reconstructed uniquely, and thus given $\bm{out}$ and $\bm{G}_r$, the vector $\bm{s}^l$ can also be reconstructed uniquely because $\bm{G}_r$ is full-rank.

Observe that, because the matrix $\bm{G}_r^T$ of size $(m+2t)\times m$ is full-rank, there exists a full-rank annihilating matrix $\bm{H}$ of dimension $2t\times (m+2t) $ such that $\bm{H} \bm{G}_r^T =\bm{0}$. This leads to: $$(\bm{H} \otimes  \bm{I}_{N_{l}/m}) \bm{e}= (\bm{H} \otimes  \bm{I}_{N_{l}/m}) \ \bm{out} \triangleq \widehat{\bm{out}}.$$ 

The solution $\bm{e}$, for this linear system of equations, is unique if the number of non-zero sub-vectors in $\bm{e}$ (\textit{i.e.}, the non-zero sub-vectors out of $\bm{e}_0,\bm{e}_1,\ldots,\bm{e}_{m+2t-1}$) is less than $\frac{Spark(\bm{H})}{2} $. Here, $Spark(\bm{H})$ is the minimum number of columns of $\bm{H}$ that are linearly dependent (see~\cite{candes2005decoding,donoho2006compressed}), or equivalently the minimum number of non-zero elements in any vector in $Null-space(\bm{H})\setminus \{\bm{0}\}$. Also note that for the Kronecker product $\left(\bm{H} \otimes \bm{I}_{N_l/m}\right)$, the $Spark(\bm{H})$ becomes equal to the minimum number of non-zero sub-vectors of length $N_l/m$ out of the total $m+2t$ sub-vectors in any vector in $Null-space\left(\bm{H} \otimes \bm{I}_{N_l/m}\right)\setminus \{\bm{0} \}$.

Now we show that $Spark(\bm{H})=2t+1$. We dismiss the case of $Spark(\bm{H})> 2t+1$ because $Spark(\bm{H})\leq Rank(\bm{H})+1$, and $Rank(\bm{H})=2t$. Now, for proof by contradiction, we assume that $Spark(\bm{H}) < 2t+1$. Then, there must be a vector $\bm{\beta}$ such that $\bm{\beta} \in \ Null-space(\bm{H})$, and has less than $2t+1$ non-zero elements. Note that, the dimension of the null-space of $\bm{H}$ is $m$ because the $m$ linearly independent columns of the matrix $\bm{G}_r$ lie in the null-space of $\bm{H}$ and also form a basis. Thus
$$
\bm{G}_r \bm{\alpha} = \bm{\beta} \text{ for some vector } \bm{\alpha}.
$$
The vector $\bm{\beta}$ has strictly more than $(m+2t-(2t+1))$ zeros, \textit{i.e.}, $m-1$ zeros. Or, $\bm{\beta}$ has at least $m$ zeros. However, this is not possible as the sub-matrix formed by picking any $m$ rows of $\bm{G}_r^T$ is non-singular. This contradicts with our assumption.

Thus, $Spark(\bm{H}) = (2t + 1)$. Thus, if $$\text{No. of non-zero  sub-vectors in } \bm{e} \ \leq \ t \ < \  \frac{2t+1}{2} =\frac{Spark(\bm{H})}{2}, $$ then there exists a unique $\bm{e}$ that satisfies $\widehat{\bm{out}} = (\bm{H} \otimes  \bm{I}_{N_{l}/m}) \bm{e}$. Thus, the first part of the lemma is proved. \\

Now, we proceed to the next part. We need to show that if each element in the \textit{non-zero} sub-vectors of $\bm{e}$ is drawn independently from a real-valued continuous distribution, then the occurrence of erroneous $\bm{out}_i$'s can still be detected with probability $1$, even if these errors can not be corrected.

Observe that if there are no erroneous sub-vectors in $\bm{out}$ and $\bm{e}=\bm{0}$, then $\widehat{\bm{out}}= \left(\bm{H} \otimes \bm{I}_{N_l/m}\right) \bm{out} = \left(\bm{H} \otimes \bm{I}_{N_l/m}\right)  \bm{e} =\left(\bm{H} \otimes \bm{I}_{N_l/m}\right) \bm{0} =  \bm{0}$. Thus, if there are actually erroneous sub-vectors in $\bm{out}$, then we will obtain $\widehat{\bm{out}}= \left(\bm{H} \otimes \bm{I}_{N_l/m}\right) \bm{out}=\left(\bm{H} \otimes \bm{I}_{N_l/m}\right) \bm{e}\neq\bm{0}$, unless $\bm{e}$ lies in $Null-space(\left(\bm{H} \otimes \bm{I}_{N_l/m}\right)) \setminus \{ \bm{0} \}$. Utilizing the assumption (see our Error Model $2$) that each element in a non-zero block of $\bm{e}$ is drawn independently from a real-valued continuous distribution, we now show that the event that $\bm{e} \ \in \ Null-space\left(\bm{H} \otimes \bm{I}_{N_l/m}\right) \setminus \{ \bm{0} \}$, given that $\bm{e}\neq \bm{0}$ occurs with probability $0$. 

Suppose that  $\bm{e} \ \in \ Null-space\left(\bm{H} \otimes \bm{I}_{N_l/m}\right) \setminus \{ \bm{0} \}$ have $k$ non-zero blocks, for any $k>0$. 


Let $\mathcal{A}_k \subset \{0,1,\ldots,m+2t-1\}$ denote the set of all the $k$ indices of the non-zero blocks (or sub-vectors) of $\bm{e}$. Clearly, $k \ \geq \ 2t + 1$ since $Spark(\bm{H})=2t+1$ and thus no vector can lie in the $Null-space\left(\bm{H} \otimes \bm{I}_{N_l/m}\right) \setminus \bm{0}$ which has less than $2t+1$ non-zero blocks. Now,
$$\left( \bm{H}_{\mathcal{A}_k}\otimes \bm{I}_{N_l/m}\right)\bm{e}_{\mathcal{A}_k} = \bm{0} $$
where $\bm{H}_{\mathcal{A}_k}$ is a sub-matrix consisting of the columns of $\bm{H}$ that are indexed in  $\mathcal{A}_k$ and $\bm{e}_{\mathcal{A}_k}$ denotes the column vector consisting of only the non-zero blocks of $\bm{e}$, \textit{i.e.}, the blocks whose locations are indexed in $\mathcal{A}_k$. 

We will first show that the dimension of the null-space of $\bm{H}_{\mathcal{A}_k}$ is $k-2t$. Recall that the $Spark(\bm{H})=2t+1$ which implies that \textit{any} $2t$ columns of $\bm{H}$ are linearly independent. Thus, any $2t$ columns of $\bm{H}_{\mathcal{A}_k}$ are also linearly independent. The rank of $\bm{H}_{\mathcal{A}_k}$ is thus $2t$ and the dimension of $\bm{H}_{\mathcal{A}_k}$ is $2t \times k$. Using Rank-Nullity theorem, the dimension of the null-space of $\bm{H}_{\mathcal{A}_k}$ is thus $k-2t$.

Now observe that $\bm{e}_{\mathcal{A}_k}$ has $k(>k-2t)$ non-zero blocks. Because every element of $\bm{e}_{\mathcal{A}_k}$ are iid real-valued continuous random variables, the probability that $\bm{e}_{\mathcal{A}_k}$ with $k N_l/m $ independent elements lies in $Null-space\left(\bm{H}_{\mathcal{A}_t} \otimes \bm{I}_{N_l/m}\right)$ of \textit{lower} dimension $(k-2t)N_l/m$, is $0$. In other words, $Null-space\left(\bm{H}_{\mathcal{A}_t} \otimes \bm{I}_{N_l/m}\right)$ of \textit{lower} dimension $(k-2t)N_l/m$ becomes a measure $0$ subset in a the space of $k N_l/m $ independent dimensions. Thus, the event that $\bm{e} \ \in \ Null-space\left(\bm{H}\otimes \bm{I}_{N_l/m}\right) \setminus \{ \bm{0} \} $, given that $\bm{e}\neq \bm{0}$, occurs with probability $0$. 
\end{proof}

\noindent \textbf{Decoding Algorithm:} Using Lemma~\ref{lem:error_tolerance}, we arrive at the following decoding technique: first compute $\left(\bm{H}\otimes \bm{I}_{N_l/m}\right) \ \bm{out} = \widehat{\bm{out}}$. If $\widehat{\bm{out}}=\bm{0}$, one can declare that there are no errors with probability $1$. Otherwise, one can use standard sparse reconstruction algorithms (e.g.~\cite{candes2005decoding,donoho2006compressed}) to find a solution for the under-determined  system of linear equations $\left(\bm{H}\otimes \bm{I}_{N_l/m}\right)\bm{e}= \widehat{\bm{out}}$ with number of non-zero blocks of $\bm{e}$ (\textit{i.e.}, sub-vectors $\bm{e}_0,\bm{e}_1,\ldots,\bm{e}_{m+2t-1}$) less than or equal to $t$. If a solution is found with number of non-zero sub-vectors at most $t$, the errors are corrected. Otherwise, the node declares that the errors cannot be uniquely determined, and thus reverts to the last checkpoint.\\

\noindent \textbf{Remark:} One might also consider the following possibility: if the true error vector $\bm{e}= \bm{e}^{(1)} + \bm{h}$ where $\bm{e}^{(1)}$ is any vector with at most $t$ non-zero blocks or sub-vectors and $\bm{h} \in \ Null-space\left(\bm{H}\otimes \bm{I}_{N_l/m}\right)\setminus \{ \bm{0}\}$. In this case, $\left(\bm{H}\otimes \bm{I}_{N_l/m}\right)\bm{e}=\left(\bm{H}\otimes \bm{I}_{N_l/m}\right)\bm{e}^{(1)} = \widehat{\bm{out}} $, and the sparse reconstruction algorithm would give $\bm{e}^{(1)}$ as the solution even though the correct solution is $\bm{e}$. 
However, we now show that the probability of the occurrence of an error $\bm{e}$ of the form $\bm{e}^{(1)} + \bm{h}$ is also $0$ under Error Model $2$. 


Let $\bm{h}$ have $k$ non-zero blocks, indexed by $\mathcal{B}_{k} \subset \{0,1,\ldots,m+2t-1\}$. Because the number of non-zero blocks in $\bm{e}^{(1)}$ is at most $t$, and $\bm{e}=\bm{e}^{(1)} + \bm{h}$, there are at least $ k - t$ non-zero blocks in $\bm{e}$ that match exactly with those in $\bm{h}$. 
We also let $\mathcal{B}_{k}^{(a)}$ and $\mathcal{B}^{(b)}_{k}$ denote two disjoint sets of indices of the non-zero blocks of $\bm{h}$ that exactly match, and do not match, with $\bm{e}$ respectively. Note that $\mathcal{B}_{k}^{(a)} \cup \mathcal{B}_{k}^{(b)} = \mathcal{B}_{k}$ and $|\mathcal{B}_{k}^{(a)}| + |\mathcal{B}_{k}^{(b)}| = |\mathcal{B}_{k}|$. 

First observe that  $|\mathcal{B}_{k}^{(a)}| \geq k - t \geq  2t +1 - t = t+1$ where the second inequality follows since $k \geq Spark(\bm{H})=2t+1$. We also have, $|\mathcal{B}_{k}^{(b)}| \leq t$ since $\bm{e}^{(1)}$ has at most $t$ non-zero blocks. The dimension of the $Range-space(\bm{H}_{\mathcal{B}_{k}^{(b)}} )$ can thus be at most: $$|\mathcal{B}_{k}^{(b)}| \ \leq \ t \ < \ t+1 \ \leq \min \{|\mathcal{B}_{k}^{(a)}|, 2t \}. $$

Now from definition of these sets,
$$\left(\bm{H}_{\mathcal{B}_{k}}\otimes \bm{I}_{N_l/m}\right) \bm{h}_{\mathcal{B}_{k}} = \left( \bm{H}_{\mathcal{B}_{k}^{(a)}} \otimes \bm{I}_{N_l/m}\right) \bm{h}_{\mathcal{B}_{k}^{(a)}} + \left( \bm{H}_{\mathcal{B}_{k}^{(b)}} \otimes \bm{I}_{N_l/m}\right) \bm{h}_{\mathcal{B}_{k}^{(b)}} = \bm{0}$$
 implying  $$\left( \bm{H}_{\mathcal{B}_{k}^{(a)}} \otimes \bm{I}_{N_l/m}\right) \bm{h}_{\mathcal{B}_{k}^{(a)}} = -\left( \bm{H}_{\mathcal{B}_{k}^{(b)}} \otimes \bm{I}_{N_l/m}\right) \bm{h}_{\mathcal{B}_{k}^{(b)}}.  $$
Observe that the $|\mathcal{B}_{k}^{(a)}|$ sub-vectors (or blocks) of $\bm{h}_{\mathcal{B}_{k}^{(a)}}$ exactly match with $\bm{e}_{\mathcal{B}_{k}^{(a)}}$ and are thus drawn independently from a continuous multivariate distribution of $\frac{N_l}{m}$ i.i.d. random variables. Thus, the vector $\left(\bm{H}_{\mathcal{B}_{k}^{(a)}}\otimes \bm{I}_{N_l/m}\right) \bm{h}_{\mathcal{B}_{k}^{(a)}}$ of length $2t(N_l/m)$ also has a real-valued, continuous distribution where $Rank\left(\bm{H}_{\mathcal{B}_{k}^{(a)}}\otimes \bm{I}_{N_l/m}\right) = \min\{ |\mathcal{B}_{k}^{(a)}|(N_l/m), 2t(N_l/m) \}$. Now, the $Range-space\left(\bm{H}_{\mathcal{B}_{k}^{(b)}} \otimes \bm{I}_{N_l/m} \right)$ becomes a measure $0$ subset in a subspace of dimension: $$\min\{ |\mathcal{B}_{k}^{(a)}|(N_l/m), 2t(N_l/m) \},$$ as $|\mathcal{B}_{k}^{(b)}|(N_l/m) < \min\{ |\mathcal{B}_{k}^{(a)}|(N_l/m), 2t(N_l/m) \}$. Thus, the probability that $\left(\bm{H}_{\mathcal{B}_{k}^{(a)}}\otimes \bm{I}_{N_l/m}\right) \bm{h}_{\mathcal{B}_{k}^{(a)}}$ lies in a lower dimensional space, \textit{i.e.}, in $Range-space\left(\bm{H}_{\mathcal{B}_{k}^{(b)}} \otimes \bm{I}_{N_l/m} \right)$, is $0$.

\section{Decentralized Algorithm for CodeNet}
\label{appendix:algorithm}
Here, we formally describe the decentralized algorithm in Algorithm~\ref{algo}, to complement the description in the section on Decentralized Implementation in the main paper. Note that, for data and label access, we assume that there is a source or shared memory from where all nodes can access the data (for the first layer, during feedforward stage) and its label (for the last layer, while transitioning from feedforward stage to backpropagation stage).
\begin{algorithm}
\caption{Decentralized CodeNet Algorithm for General $t$}
\label{algo}
\begin{algorithmic}[1]
\State Pre-processing Step: Encode and store appropriate sub-matrices of $\bm{W}^l$ initially as described
\State In each iteration: \textbf{If} iteration number $\% I_0 = 0$, \textbf{Checkpoint at Disk}.
\State \textbf{FEEDFORWARD STAGE} (Active Nodes: $0\leq i \leq m+2t-1$, $0 \leq j \leq n-1$)
\State \textbf{For} layers $l=1,2,\ldots,L$ serially:
\State \hspace{0.1cm} Compute  $\bm{W}^l_{i,j}\bm{x}^l_j$ or $\widetilde{\bm{W}}^l_{i,j}\bm{x}^l_j$ at all \textit{active} nodes in parallel
\State \hspace{0.1cm} [All-Reduce] Sum $\bm{s}^l_i=\sum_{j=0}^{n-1}\bm{W}^l_{i,j}\bm{x}^l_j$ or $\widetilde{\bm{s}}^l_i=\sum_{j=0}^{n-1}\widetilde{\bm{W}}^l_{i,j}\bm{x}^l_j$  at each \textit{active} row in parallel 
\State \hspace{0.1cm} [All-Reduce $2t$ times] Perform $2t$ consistency checks (error-detection) using all $\bm{s}^l_i$ or $\widetilde{\bm{s}}^l_i$'s at each \textit{active} column in parallel
\State \hspace{0.1cm} Verification step to check for disagreement among \textit{active} nodes (implies errors during error-detection step):
\State \hspace{0.5cm} \textbf{If} disagreement, \textbf{Return} to last checkpoint
\State \hspace{0.5cm} \textbf{Else If} no errors detected: 
\State \hspace{0.8cm}\textbf{Additional Encoding} step at all \textit{inactive} nodes:  
\State \hspace{1.0cm} [Reduce $2t$ times] Encode $\widetilde{\bm{x}}^l_j$ for all \textit{inactive} columns, at rows $0$ to $m-1$ in parallel
\State \hspace{0.8cm}
\textbf{Generate} $\bm{x}^{(l+1)}_j$ at all active nodes:\\
\State \hspace{1.0cm} Compute $\bm{x}^{(l+1)}_i=f(\bm{s}^l_i)$ at rows $0$ to $m-1$ in parallel
\State \hspace{1.0cm} [Broadcast] Fetch appropriate parts of $\bm{x}^{l+1}_j$ from the nodes that have it, at each \textit{active} column in parallel
\State \hspace{0.5cm} \textbf{Else} \textbf{Decoding} at all active nodes attempting to correct detected errors:
\State \hspace{1.0cm} [All-gather] Get all $\bm{s}^l_i$ or $\widetilde{\bm{s}}^l_i$'s at each column in parallel and attempt to decode $\bm{s}^l$ at every \textit{active} node
\State \hspace{1.0cm} \textbf{If} more than $t$ errors, \textbf{Return} to last checkpoint
\State \hspace{1.0cm} \textbf{Else} Verification step to check for decoding errors:
\State \hspace{1.6cm} \textbf{If} disagreement, \textbf{Return} to last checkpoint
\State \hspace{1.6cm} \textbf{Else} \textbf{Regeneration} of possible erroneous nodes,
\State \hspace{1.6cm} \textbf{Additional Encoding} step and  \textbf{Generation} of $\bm{x}^{(l+1)}_j.$ 
\algstore{myalg}
\end{algorithmic}
\end{algorithm}

\begin{algorithm}                     
\begin{algorithmic} [1]                   
\algrestore{myalg}
\State \textbf{BACKPROPAGATION STAGE}
(Active Nodes: $0\leq i \leq m-1$, $0 \leq j \leq n+2t-1$)
\State \textbf{For} layers $l=L,L-1,\ldots,1$ serially:
\State \hspace{0.1cm} Compute  $(\bm{\delta}^l_i)^T\bm{W}^l_{i,j}$ or $(\bm{\delta}^l_i)^T\widetilde{\bm{W}}^l_{i,j}$ at \textit{active} nodes in parallel
\State \hspace{0.1cm} [All-reduce] Sum $(\bm{c}^l_j)^T=\sum_{i=0}^{m-1}(\bm{\delta}^l_i)^T\bm{W}^l_{i,j}$ or $(\widetilde{\bm{c}}^l_j)^T=\sum_{i=0}^{m-1}(\bm{\delta}^l_i)^T\widetilde{\bm{W}}^l_{i,j}$ at each \textit{active} column in parallel
\State \hspace{0.1cm} [All-Reduce $2t$ times] Perform $2t$ consistency checks (error-detection) using all $(\bm{c}^l_j)^T$ or $(\widetilde{\bm{c}}^l_j)^T$'s at each row in parallel
\State \hspace{0.1cm} Verification step to check for disagreement among all nodes (implies errors during error-detection):
\State \hspace{0.5cm} \textbf{If} disagreement, \textbf{Return} to last checkpoint
\State \hspace{0.5cm} \textbf{Else If} no errors detected: 
\State \hspace{0.8cm} \textbf{Additional Encoding} step at all inactive nodes:
\State \hspace{1.0cm} [Reduce $2t$ times] Encode $(\widetilde{\bm{\delta}}^l_i)^T$ for all \textit{inactive} rows, at columns $0$ to $n-1$ in parallel 
\State \hspace{0.8cm} \textbf{Generate} $\bm{\delta}^{l-1}_i$ at all active nodes:
\State \hspace{1.0cm} Compute $(\bm{\delta}^{l-1}_j)^T= (\bm{c}^l_j)^T \bm{D}^l_j$ at cols. $0$ to $n-1$ in parallel
\State \hspace{1.0cm} [Broadcast] Fetch appropriate parts of $\bm{\delta}^{l-1}_i$ from the nodes that have it, at each \textit{active} row in parallel
\State \hspace{0.5cm} \textbf{Else} \textbf{Decoding} at all active nodes:
\State \hspace{1.0cm} [All-gather] Get all $(\bm{c}^l_j)^T$ or $(\widetilde{\bm{c}}^l_j)^T$'s at each row in parallel and attempt to decode $(\bm{c}^l)^T$ at every \textit{active} node
\State \hspace{1.0cm} \textbf{If} more than $t$ errors, \textbf{Return} to last checkpoint
\State \hspace{1.0cm} \textbf{Else} Verification step to check for decoding errors:
\State \hspace{1.6cm} \textbf{If} disagreement, \textbf{Return} to last checkpoint.
\State \hspace{1.6cm} \textbf{Else} \textbf{Regeneration} of possible erroneous nodes,
\State \hspace{1.6cm} \textbf{Additional Encoding} step and  \textbf{Generation} of $\bm{\delta}^{l-1}_i$. 
\State \textbf{UPDATE STAGE} (All nodes)\\
$
\bm{W}^l_{i,j} \leftarrow \bm{W}^l_{i,j} + \bm{\delta}^l_i (\bm{x}_j^l)^T \ \forall \  0 \leq i \leq m-1  \text{ and } 0 \leq j \leq n-1 $\\
$ \widetilde{\bm{W}}^l_{i,j} \leftarrow \widetilde{\bm{W}}^l_{i,j} + \bm{\delta}^l_i (\widetilde{\bm{x}}^l_j)^T \ \forall \  0 \leq i \leq m-1  \text{ and }   n \leq j \leq n+2t-1 $\\
$ \widetilde{\bm{W}}^l_{i,j} \leftarrow \widetilde{\bm{W}}^l_{i,j} + \widetilde{\bm{\delta}}^l_i (\bm{x}^l_j)^T \ \forall \  m \leq i \leq m+2t-1  \text{ and }   0 \leq j \leq n-1.
$
\end{algorithmic}
\end{algorithm}

\clearpage

\section{Theoretical Analysis of Runtime}
\label{appendix:runtime}

We first elaborate the assumptions here:

\noindent 
\begin{itemize}
\item[(i)] Errors (under Error Model $2$) at any node may or may not have dependence on other nodes. Error events in an iteration are simply divided into three disjoint sets: (1) zero errors (probability $p_0$); (2) error patterns correctable by CodeNet (probability $p_1$); and
(3) error patterns not correctable by CodeNet (probability $p_2=1-p_0-p_1$). E.g., if CodeNet can correct any $t$ errors, then $p_1$ includes the probability of all error-patterns with at most $t$ errors. Observe that, for the errors captured by $p_1$, CodeNet proceeds forward to the next iteration after regenerating the corrupt sub-matrices, but the replication strategy simply reverts to the iteration of its last checkpoint. For error patterns captured by $p_2$, both strategies have to revert to their last checkpoint. 
\item[(ii)]  Period of checkpointing (number of iterations after which we checkpoint), $I_0$ is fixed, but can vary for different strategies. 
\item[(iii)]  Time taken by both replication and CodeNet for an error-free iteration is $\tau_f$. We justify this in Theorem~\ref{thm:complexity} by showing that the communication and computation complexities of replication and CodeNet are comparable. 

\item[(iv)] Denoting the time to resume from previous checkpoint (by reading from disc) by $\tau_b$, we assume that $\tau_b \gg \tau_f $ as fetching data from the disk is extremely time-intensive. We also pessimistically assume that the time to run an iteration with error correction and regeneration by CodeNet be $\tau_b$. When an error occurs, CodeNet proceeds forward in computation, but it also regenerates some of the sub-matrices of $\bm{W}^l$. This regeneration comes with extra communication cost, increasing the time of the iteration. But because it does not require reading from the disk, letting it be as high as $\tau_b$ is a actually a pessimistic assumption in evaluation of the performance of our strategy. 
\item[(v)] Time to save the entire state using checkpointing is $\tau_{cpt}$.  
\end{itemize}

In particular, if the number of errors in an iteration follow Poisson distribution\cite{li2007memory} with parameter $\lambda$, then the ratio of expected time scales with $\lambda$. Again, if any node fails independently with probability $p$ whenever it is used for one of the $3$ most complexity-intensive ($\Theta(N_lN_{l-1}/P)$) computations (O1, O2, or O3), in each layer, in an iteration, then $p_0=(1-p)^3\hat{P}L$ and $p_1 \geq\binom{3\hat{P}L}{1}p^1(1-p)^{3\hat{P}L-1}$ where $\hat{P}$ is the total number of nodes used and $3L$ is the number of times a node is used in an iteration, \textit{i.e.}, thrice for each layer. 

Now, we proceed to the proof.

\begin{proof}[Proof of Theorem \ref{thm:time}]
Let us consider the iterations over a single checkpointing iteration-period of $I_0$. Within one period, let $k$ denote the iteration index. The index takes values from $k=0$ to $k=I_0$. For all $k>0$, let $T_k$ be a random variable that denotes the time taken to reach iteration $k$ starting from $k=0$. We need to find $\mathbb{E}[T_{I_0}]$, \textit{i.e.} the expected time to reach iteration $k=I_0$ starting at $k=0$. Then, the expected time to complete $M$ iterations would be determined by $\frac{M}{I_0}\tau_{cpt} + \frac{M}{I_0}\mathbb{E}[T_{I_0}]$.
The figure below shows the Markov Chain of the different states starting at state $0$ till state $k+1$. Note that, at any state the algorithm can perform one of the following:
\begin{enumerate}
\item Move one step forward taking time $\tau_f$ with probability $p_0$, \textit{i.e.} when no error occurs.
\item Move one step forward but taking time $\tau_b$ with probability $p_1$, \textit{i.e.} when errors occur but are corrected.
\item Move backward to last checkpoint taking time $\tau_b$ with probability $p_2$, \textit{i.e.} when errors occur and cannot be corrected.
\end{enumerate}

\begin{tikzpicture}
[align=center,node distance=2cm]
        \node[state]             (0) {0};
        \node[state, right=of 0] (1) {1};
        \node[state, right=of 1] (2) {2};
        \node[draw=none, right=of 2, minimum width=2cm] (3-k) {$\cdots$};
 		\node[state, right=of 3-k] (k) {k};
        \node[state, right=of k] (k+1) {k+1};
        \draw[every loop]
            (0) edge[bend left,auto=left] node {$p_0,\tau_f$} (1)
            (0) edge[auto=left] node {$p_1,\tau_b$} (1)
            (1) edge[bend left,auto=left] node {$p_0,\tau_f$} (2)
            (1) edge[auto=left] node {$p_1,\tau_b$} (2)
            (2) edge[bend left,auto=left] node {$p_0,\tau_f$} (3-k)
            (2) edge[auto=left] node {$p_1,\tau_b$} (3-k)
            (3-k) edge[bend left, auto=left] node {$p_0,\tau_f$} (k)
            (3-k) edge[ auto=left] node {$p_1,\tau_b$} (k)
            (k) edge[bend left, auto=left] node {$p_0,\tau_f$} (k+1)
            (k) edge[ auto=left] node {$p_1,\tau_b$} (k+1)
            (1) edge[bend left,auto=right] node {$p_2,\tau_b$} (0)
            (2) edge[bend left, auto=right] node {$p_2,\tau_b$} (0)
            (k) edge[bend left, auto=right] node {$p_2,\tau_b$} (0);
\end{tikzpicture}

Now we use the following lemma.
\begin{lem}
\label{lem:markov_chain}
The expected time to reach state $k$ starting from state $0$ for this Markov Chain is given by 
$$\mathbb{E}[T_{k}] = (\tau_f p_0 + \tau_b(1-p_0)) \sum_{i=0}^{k-1} \frac{1}{(p_0+p_1)^i} $$
\end{lem}
The proof of Lemma \ref{lem:markov_chain} is provided at the end of this section. First, using Lemma \ref{lem:markov_chain}, we can derive that the total time for $n$ iterations, including the time for checkpointing. This  is given by $$\frac{M}{I_0}\tau_{cpt} + \frac{M}{I_0}(\tau_f p_0 + \tau_b(1-p_0)) \frac{\frac{1}{(p_0+p_1)^{I_0}}-1}{\frac{1}{(p_0+p_1)}-1}. $$

For the replication strategy, the algorithm moves forward with probability $p_0$, and moves to the last checkpoint with probability $1-p_0$. The expected time can be derived by setting $p_1=0$ in the expression of CodeNet strategy. Thus, the total time for $M$ iterations, including the time for checkpointing, is given by $$\frac{M}{I_0}\tau_{cpt} + \frac{M}{I_0}(\tau_f p_0 + \tau_b(1-p_0)) \frac{\frac{1}{p_0^{I_0}}-1}{\frac{1}{p_0}-1}. $$ 
\end{proof}

What remains now is the proof of Lemma \ref{lem:markov_chain}.\\

\begin{proof}[Proof of Lemma \ref{lem:markov_chain}]

Observe the following:

\begin{align}
\mathbb{E}[T_{k+1}] &= \mathbb{E}[T_{k+1} |  \text{no error}] \Pr(\text{no error}) + \mathbb{E}[T_{k+1} | \text{recoverable error}] \Pr(\text{recoverable error}) \nonumber \\
& \hspace{5cm} + \mathbb{E}[T_{k+1} | \text{non-recoverable error}] \Pr(\text{non-recoverable error}) \nonumber \\
& = (\mathbb{E}[T_{k}+\tau_f])p_0 + (\mathbb{E}[T_{k}+\tau_b])p_1 + (\mathbb{E}[T_{k}+\tau_b + T_{k+1}])p_2 \nonumber \\
&= \mathbb{E}[T_{k}] + (\tau_f p_0 + \tau_b p_1 + \tau_b p_2) + p_2 \mathbb{E}[T_{k+1}] \nonumber \\
& = \mathbb{E}[T_{k}] + (\tau_f p_0 + \tau_b(1-p_0)) + p_2 \mathbb{E}[T_{k+1}] 
\end{align}

This leads to the following recursion:
\begin{align}
\mathbb{E}[T_{k+1}]  = \frac{1}{1-p_2} \left( \mathbb{E}[T_{k}] + (\tau_f p_0 + \tau_b(1-p_0)) \right) = \frac{1}{p_0 +p_1} \left(\mathbb{E}[T_{k}] + (\tau_f p_0 + \tau_b(1-p_0))\right)
\label{recursion}
\end{align}

We solve the recursion using induction.

\noindent 
\textbf{Induction Hypothesis:}
$$\mathbb{E}[T_{k}] = (\tau_f p_0 + \tau_b(1-p_0)) \sum_{i=0}^{k-1} \frac{1}{(p_0+p_1)^i} $$

\textbf{Initial Case:} $k=1$ 

\begin{align}
& \mathbb{E}[T_{1}]  =  (\tau_f p_0 + \tau_b p_1 ) + (\tau_b + \mathbb{E}[T_{1}])p_2 \nonumber \\
& \implies \mathbb{E}[T_{1}] = (\tau_f p_0 + \tau_b(1-p_0))  \ \ \text{(satisfies induction hypothesis)}.
\end{align}
Now assume that the induction hypothesis holds for any $k >1$. Then, using~\eqref{recursion}, we now show that it also holds for $k+1$:
\begin{align}
\mathbb{E}[T_{k+1}]  &= \frac{1}{p_0 +p_1} \left(\mathbb{E}[T_{k}] + (\tau_f p_0 + \tau_b(1-p_0))\right) \nonumber \\
& = \frac{1}{p_0 +p_1} \left( (\tau_f p_0 + \tau_b(1-p_0)) \sum_{i=0}^{k-1} \frac{1}{(p_0+p_1)^i} + (\tau_f p_0 + \tau_b(1-p_0))\right) \nonumber \\
& = (\tau_f p_0 + \tau_b(1-p_0)) \sum_{i=0}^{k} \frac{1}{(p_0+p_1)^i}.
\end{align}

Thus, from induction we obtain that,
\begin{align}
\mathbb{E}[T_{I_0}] = (\tau_f p_0 + \tau_b(1-p_0)) \sum_{i=0}^{I_0-1} \frac{1}{(p_0+p_1)^i} = (\tau_f p_0 + \tau_b(1-p_0)) \frac{\frac{1}{(p_0+p_1)^{I_0}}-1}{\frac{1}{(p_0+p_1)-1}}.
\end{align}
\end{proof}

\section{Computational and Communication Complexity Analysis}
\label{appendix:complexity}

Before proceeding to the proof of Theorem \ref{thm:complexity}, we first briefly discuss our definition of communication complexity and then briefly state the computation and communication complexities of some of the standard collective communication protocols that we discussed\cite{chan2007collective}.

\begin{defn}[Commnication Complexity]
The communication complexity of communicating $N$ real-valued numbers from one node to another in a single round is defined as $\alpha + \beta N$ where $\alpha$ and $\beta$ are system dependent constants. 
\end{defn}

\noindent \textbf{Remark:} Note that $\alpha$ denotes the latency associated with setting up the communication protocol between two nodes and $\beta$ denotes the cost or bandwidth consumed for each item sent. We also use the constant $\gamma$ as a system dependent scaling constant for computation complexity. Typically $\gamma \ll \beta \ll \alpha $.

Now, we will be stating the computation and communication complexities of the standard collective communication protocols that will be used in the proof of Theorem \ref{thm:complexity}. A detailed analysis and discussion is provided in \cite{chan2007collective}.

\begin{itemize}
\item Reduce: In a cluster of $P$ total nodes, each node initially has a vector,~ e.g., $\bm{a}_p$ of length $N$. After Reduce, one node gets the sum $\sum_{p=1}^{P}\bm{a}_p$. The computation cost is $\gamma \frac{P-1}{P}N $ and the communication cost is $\alpha \log{P} + \beta N $.

\item All-Reduce: In a cluster of $P$ total nodes, each node initially has a vector,~ e.g., $\bm{a}_p$ of length $N$. After All-Reduce, each node gets the sum $\sum_{p=1}^{P}\bm{a}_p$. The computation cost is $\gamma \frac{P-1}{P}N $ and the communication cost is $\alpha \log{P} + 2\beta \frac{P-1}{P} N $.

\item Gather and All-Gather: In a cluster of $P$ total nodes, each node initially has a vector,~ e.g., $\bm{a}_p$ of length $N$. After Gather (or All-Gather), one node (or all the nodes) gets all the vectors $\{\bm{a}_p |p=1,2,\ldots,P\}$. There is no computation cost, and the communication cost is $\alpha \log{P} + 2\beta (P-1) N  $.

\item Broadcast: In a cluster of $P$ total nodes, one node initially has a vector,~ e.g., $\bm{a}_p$ of length $N$. After Broadcast, all the nodes get the vector $\bm{a}_p$. There is no computation cost, and the communication cost is $\alpha \log{P} + \beta N  $.
\end{itemize}

\noindent \textbf{Remark:} These communication complexities are achievable using tree-type communication algorithms \cite{chan2007collective}, that efficiently use all the nodes to reduce the total communication cost.

\begin{proof}[Proof of Theorem \ref{thm:complexity}] 
For the first part of the theorem, we will calculate the communication cost of CodeNet and Replication at a single layer, in an error-free iteration. We let $m=n=\sqrt{P}$ for simplicity.

For CodeNet, the steps involving communication in the feedforward stage are as follows:
\begin{itemize}
\item All-Reduce operation to compute the sums at each active row in parallel: $$\text{Communication complexity is: }  \alpha \log{(\sqrt{P})} + 2\beta (\sqrt{P}-1)\frac{N_l}{P} .$$
\item All-Reduce operation $2t$ times to detect errors at each active column in parallel: $$ \text{Communication complexity is: }  \alpha \log{(\sqrt{P}+2t)} + 4t\beta \frac{(\sqrt{P}+2t-1)}{\sqrt{P}+2t} \frac{N_l}{\sqrt{P}} .$$
\item Verification Step under Error Model $2$ to check for errors during detection: $$ \text{Communication complexity is: }  \alpha \log{(\hat{P})} + \beta \hat{P}t .$$
\item Reduce operation for the additional encoding at the $2t$ inactive columns, at rows $0$ to $m-1$ in parallel: 
$$\text{Communication complexity is: }  \alpha \log{(\sqrt{P}+2t)} + 2t\beta \frac{N_{l-1}}{\sqrt{P}} .$$
\item Broadcast $\bm{x}^{(l+1)}_j$ from node $(j,j)$ to all nodes in the $j$-th active column, for each columns $0$ to $n-1$ in parallel:
$$\text{Communication complexity is: } \alpha \log{(\sqrt{P}+2t)} + \beta \frac{N_{l}}{\sqrt{P}}.$$
\end{itemize}

Similarly, the steps involving communication in the backpropagation stage are as follows:
\begin{itemize}
\item All-Reduce operation to compute the sums at each active column in parallel: $$\text{Communication complexity is: }  \alpha \log{(\sqrt{P})} + 2\beta (\sqrt{P}-1)\frac{N_{l-1}}{P} .$$
\item All-Reduce operation $2t$ times to detect errors at each active row in parallel: $$ \text{Communication complexity is: }  \alpha \log{(\sqrt{P}+2t)} + 4t\beta \frac{(\sqrt{P}+2t-1)}{\sqrt{P}+2t} \frac{N_{l-1}}{\sqrt{P}} .$$
\item Verification Step under Error Model $2$ to check for errors during detection: $$ \text{Communication complexity is: }  \alpha \log{(\hat{P})} + \beta \hat{P}t .$$
\item Reduce operation for the additional encoding at the $2t$ inactive rows, at columns $0$ to $n-1$ in parallel: 
$$\text{Communication complexity is: }  \alpha \log{(\sqrt{P}+2t)} + 2\beta t\frac{N_{l}}{\sqrt{P}} .$$
\item Broadcast $\bm{\delta}^{(l-1)}_i$ from node $(i,i)$ to all nodes in the $i$-th active row, for rows $0$ to $m-1$ in parallel:
$$\text{Communication complexity is: } \alpha \log{(\sqrt{P}+2t)} + \beta \frac{N_{l-1}}{\sqrt{P}}.$$
\end{itemize}

There is no more communication in the update stage. Thus, the total communication cost of CodeNet at a single layer, in an error-free iteration, is upper-bounded by:
\begin{align}
&8\alpha \log{(\sqrt{P}+2t)} + \beta\left( 2\frac{\sqrt{P}-1}{\sqrt{P}} + 4t \frac{\sqrt{P}+2t-1 }{\sqrt{P}+2t}  +2t +1 \right)\left(\frac{N_l}{\sqrt{P}}+ \frac{N_{l-1}}{\sqrt{P}} \right) + 2( \alpha \log{(\hat{P})} + \beta \hat{P}t ) \nonumber \\
&\leq 8\alpha \log{(\sqrt{\hat{P}})} + \beta(6t+3)\left(\frac{N_l}{\sqrt{P}}+ \frac{N_{l-1}}{\sqrt{P}} \right)+ 2( \alpha \log{(\hat{P})} + \beta \hat{P}t ) \nonumber \\
& = 6\alpha \log{(\hat{P})} + \beta (6t+3) \left(\frac{N_l}{\sqrt{P}}+ \frac{N_{l-1}}{\sqrt{P}} \right) + 2\beta \hat{P}t .
\end{align}
To compare, the steps of the replication strategy that involve communication are primarily a Reduce and a Broadcast at the matrix-vector products in the feedforward and backpropagation stages, along with additional communication across nodes computing the same output for comparing the outputs (or parts of it) to detect errors. The additional communication for exchange of outputs can be kept low, particularly if, only a part of the whole output and not the whole of it is exchanged. As an optimistic estimate for replication, we therefore ignore the communication cost for exchanging the outputs for error detection, and calculate the communication cost only from the Reduce and Broadcast. The total complexity is thus given by:
\begin{align}
4\alpha \log{(\sqrt{P})} + 2\beta  \left( \frac{N_l}{\sqrt{P} } + \frac{N_{l-1} }{\sqrt{P} } \right) = 2\alpha \log{(P)} + 2\beta  \left( \frac{N_l}{\sqrt{P} } + \frac{N_{l-1} }{\sqrt{P} } \right) .
\end{align}

Thus, an upper-bound on the ratio of the communication complexities of CodeNet to replication for a single layer, in an error free iteration, is given by:
\begin{align}
\frac{\text{Comm. Complexity (CodeNet) } }{ \text{Comm. Complexity (replication) } } \leq \frac{6\alpha \log{(\hat{P})} + \beta (6t+3) \left(\frac{N_l}{\sqrt{P}}+ \frac{N_{l-1}}{\sqrt{P}} \right) + 2\beta  \hat{P}t}{2\alpha \log{(P)} + 2\beta \left( \frac{N_l}{\sqrt{P} } + \frac{N_{l-1} }{\sqrt{P} } \right)  }.
\end{align}
In the limit of $\hat{P},P,N_{l},N_{l-1} \to \infty$, this ratio scales as $\mathcal{O}(3t)$ as long as $\hat{P}^{3/2}=o(\min \{N_l, N_{l-1} \} ) $.

Next, we prove the second part of the theorem which compares the computational complexity of CodeNet and replication. In the feedforward stage for both CodeNet and replication, each node primarily computes a matrix-vector product of complexity $2\gamma \frac{N_l N_{l-1} }{P}$, followed by either an All-Reduce or Reduce operation, both with the same computational complexity of $\gamma \frac{\sqrt{P}-1}{\sqrt{P}} \frac{N_{l}}{\sqrt{P}}$. Then, CodeNet performs $2t$ consistency checks to detect errors which contribute to a computational complexity of $ \gamma (2t)\frac{\sqrt{P}+2t-1}{\sqrt{P}+2t}  \frac{N_{l}}{\sqrt{P}}$. Under Error Model $2$, CodeNet then performs a verification step to check for errors during the detection step of computational complexity $\gamma(P+2t\sqrt{P})t \leq \gamma \hat{P}t $, followed by the additional encoding step of computational complexity $\gamma (2t) \frac{\sqrt{P}+2t-1}{\sqrt{P}+2t}  \frac{N_{l-1}}{\sqrt{P}} $.  

Similarly, in the backpropagation stage for both CodeNet and replication, each node primarily computes a matrix-vector product of complexity $2\gamma \frac{N_l N_{l-1} }{P}$, followed by either an All-Reduce or Reduce operation, both with the same computational complexity of $\gamma \frac{\sqrt{P}-1}{\sqrt{P}} \frac{N_{l-1}}{\sqrt{P}}$. Then, CodeNet performs $2t$ consistency checks to detect errors which contribute to a computational complexity of $ \gamma (2t)\frac{\sqrt{P}+2t-1}{\sqrt{P}+2t}  \frac{N_{l-1}}{\sqrt{P}}$. Under Error Model $2$, CodeNet then performs a verification step to check for errors during the detection step of computational complexity $\gamma (P+2t\sqrt{P})t \leq \gamma \hat{P}t $, followed by the additional encoding step of computational complexity $\gamma (2t) \frac{\sqrt{P}+2t-1}{\sqrt{P}+2t}  \frac{N_{l}}{\sqrt{P}} $.  

Finally, in the update stage, the computational complexity is again $2\gamma \frac{N_l N_{l-1} }{P}$ which is the complexity of performing a rank-$1$ update on a matrix of size $\frac{N_l}{\sqrt{P}}\times \frac{N_{l-1}}{\sqrt{P}}$.

Thus, the ratio of the communication complexities of CodeNet to replication is upper-bounded by:
\begin{align}
\frac{\text{Comp. Complexity (CodeNet)}}{\text{Comp. Complexity (Replication)}}\leq \frac{6\gamma \frac{N_l N_{l-1} }{P} + \gamma(1+4t)\left(\frac{N_l}{\sqrt{P} }+\frac{N_{l-1}}{\sqrt{P}}\right) + 2\gamma\hat{P}t }{6\gamma \frac{N_l N_{l-1} }{P} + \gamma \left(\frac{N_l}{\sqrt{P} }+\frac{N_{l-1}}{\sqrt{P}}\right) }.
\end{align}
In the limit of $\hat{P},P,N_l,N_{l-1}\to \infty$, the ratio scales as $\Theta(1)$ as long as $\hat{P}^{3/2}=o(\min \{N_l, N_{l-1} \} ) $.
\end{proof}

\end{document}